\documentclass[11pt]{article}
\usepackage{fullpage}
\usepackage{url}
\usepackage[colorlinks=true,linkcolor=Emerald,citecolor=RoyalBlue]{hyperref}
\usepackage{amsmath,nccmath, amsfonts,amsthm,multirow,mdframed,amssymb}
\usepackage{times}
\usepackage{paralist}
\usepackage{graphicx}
\usepackage{floatpag}
\usepackage{float}
\usepackage{bbold}
\usepackage{enumitem}
\usepackage{subcaption} 
\usepackage{soul}
\usepackage{comment}
\usepackage{calc}
\usepackage{upgreek}
\usepackage{mathtools}
\usepackage{color}
\usepackage[dvipsnames]{xcolor}
\usepackage{xspace}
\usepackage{tcolorbox}
\usepackage{cleveref}
\usepackage{stackengine}
\usepackage{tocloft}

\stackMath
\allowdisplaybreaks

\DeclareMathOperator*{\E}{\mathbb{E}}

\DeclareMathOperator*{\Ind}{\mathbb{1}}

\DeclarePairedDelimiter\ip{\langle}{\rangle}
\newcommand\skipi{{\vskip 10pt}}
\def \F {{\mathbb F}}
\def \Q {{\mathbb Q}}
\def \R {\mathbb{R}}
\def \Z {\mathbb{Z}}
\def \N {\mathbb{N}}

\def \cA {\mathcal{A}}
\def \cB {\mathcal{B}}
\def \cC {\mathcal{C}}
\def \cD {\mathcal{D}}
\def \cE {\mathcal{E}}

\def \cH {\mathcal{H}}

\def \cP {\mathcal{P}}
\def \cR {\mathcal{R}}
\def \cS {\mathcal{S}}

\def \cV {\mathcal{V}}

\def \sym {\text{Sym}}

\def \poly {\mathrm{poly}}

\def \eps {{\varepsilon}}

\def \val {{\sf val}}
\def \viol {\mathrm{viol}}

\def \maj {\operatorname{Maj}}
\def \agr {\operatorname{agr}}

\def\ggg{\gtrsim}
\def\lll{\lesssim}

\def \rt {\text{rt}}

\def \spn {\text{span}}
\def \good {\text{Good}}
\def \csat {\text{CircuitSAT}}
\def \deg {\text{deg}}
\def \qe {\text{QE}}
\def \out {\textsc{out}}
\def \In {\textsc{in}}
\def \pl {\sf{polylog}}

\def \sp {\text{SP}}

\def \size {{\sf{size}}}

\renewcommand{\leq}{\leqslant}
\renewcommand{\le}{\leqslant}
\renewcommand{\geq}{\geqslant}
\renewcommand{\ge}{\geqslant}

\newcommand{\wt}[1]{\widetilde{#1}}
\renewcommand{\bar}[1]{\overline{#1}}

\usepackage{tikz}
\newcommand\rcirc{\tikz[baseline=(char.base)]{
            \node[shape=circle,draw,inner sep=0.8pt] (char) {\small r};}}
\newcommand\zz{\tikz[baseline=(char.base)]{
            \node[shape=circle,draw,inner sep=0.8pt] (char) {\small z};}}

\newtheorem{theorem}{Theorem}[section]

\newtheorem{fact}[theorem]{Fact}

\newtheorem{prop}[theorem]{Proposition}
\newtheorem{corollary}[theorem]{Corollary}
\newtheorem{cor}[theorem]{Corollary}

\newtheorem{lemma}[theorem]{Lemma}
\newtheorem*{lemma*}{Lemma}
\newtheorem*{theorem*}{Theorem}
\newtheorem{claim}[theorem]{Claim}
\newtheorem*{claim*}{Claim}

\newtheorem{remark}[theorem]{Remark}
\newtheorem*{remark*}{Remark}
\newtheorem{definition}[theorem]{Definition}
\theoremstyle{definition}

\newtheorem{agr-test}{Agreement-Test}
\newtheorem{list-agr-test}{List-Agreement-Test}

\makeatletter
\let\c@fconjecture\c@conjecture
\makeatother

\makeatletter
\let\c@fconj\c@conj
\makeatother

\title{Quasi-Linear Size PCPs with Small Soundness from HDX}

\author{Mitali Bafna \thanks{Department of Mathematics, Massachusetts Institute of Technology.}
 \and
 	Dor Minzer\footnotemark[1] \\ \\ With an Appendix by Zhiwei Yun\footnotemark[1]
\and Nikhil Vyas\thanks{SEAS, Harvard University.}}

\date{\vspace{-5ex}}
\begin{document}
\maketitle

\begin{abstract}
We construct 2-query, quasi-linear size  probabilistically checkable proofs (PCPs) with arbitrarily small constant soundness, improving upon Dinur's  2-query quasi-linear size PCPs with soundness $1-\Omega(1)$. As an immediate corollary, we get that under the exponential time hypothesis, for all $\eps>0$ no approximation algorithm for $3$-SAT can obtain an approximation ratio of $7/8+\eps$ in time $2^{n/\log^C n}$, where $C$ is a constant depending on $\eps$. Our result builds on a recent line of works showing the existence of linear size direct product testers with small soundness~\cite{BafnaMinzer,DiksteinD-agreement,BLM24,dikstein2023swap,DDL24}. 

The main new ingredient in our proof is a technique that embeds a given 2-CSP into a 2-CSP on a prescribed graph, provided that the latter is a graph underlying a sufficiently good high-dimensional expander (HDX). We achieve this by establishing a novel connection between PCPs and fault-tolerant distributed computing, more precisely, to the almost-everywhere reliable transmission problem from~\cite{Dwork,Upfal,Chandran,chandran2012edge,JRV}. We instantiate this connection by showing that graphs underlying HDXs admit routing protocols that are tolerant to adversarial edge corruptions, also improving upon the state of the art constructions of sparse edge-fault-tolerant networks in the process.

Our PCP construction requires variants of the aforementioned direct product testers with poly-logarithmic degree. The existence and constructability of these variants is shown in~\Cref{app:cl} by Zhiwei Yun.
\end{abstract}

\thispagestyle{empty}
\setcounter{page}{0}

\newpage
\pagenumbering{gobble}
\setcounter{tocdepth}{2}
\tableofcontents
\thispagestyle{empty}
\setcounter{page}{0}

\newpage

\pagenumbering{arabic}
\section{Introduction}
The PCP Theorem~\cite{FGLSS,AS,ALMSS} is a cornerstone of
theoretical computer science, with many applications in
hardness of approximation, cryptography and interactive 
protocols. It will be convenient for us to take the following
combinatorial view of PCPs, using the language of the 
Label Cover problem.
\begin{definition}\label{def:label-cover}
An instance of Label Cover $\Psi = (G=(L\cup R,E), \Sigma_L, \Sigma_R, \Phi = \{\Phi_{e}\}_{e\in E})$ consists of a bipartite 
graph $G$, alphabets $\Sigma_L, \Sigma_R$ and constraints 
$\Phi_e\subseteq \Sigma_L\times \Sigma_R$, one for each edge.
Each one of the constraints is a projection constraint, meaning
that for every $e\in E$ there is a map $\phi_e\colon \Sigma_L\to\Sigma_R$ such that
\[
\Phi_e = \{(\sigma,\phi_e(\sigma))~|~\sigma\in \Sigma_L\}.
\]
\end{definition} 
Given a label cover instance $\Psi$, the goal is to find assignments $A_L\colon L\to\Sigma_L$ and $A_R\colon R\to\Sigma_R$ that satisfy as many of the constraints as possible, namely that maximize the quantity
\[
{\sf val}_{\Psi}(A_L,A_R) = \frac{1}{|E|}\left|\{e=(u,v)\in E~|~(A_L(u),A_R(v))\in \Phi_e\}\right|,
\qquad
{\sf val}(\Psi) = \max_{A_L,A_R}{\sf val}_{\Psi}(A_L,A_R).
\]
We denote by gap-LabelCover$[c,s]$ the promise problem wherein the input is an instance $\Psi$ of label cover promised to either have ${\sf val}(\Psi)\geq c$ or else ${\sf val}(\Psi)\leq s$, and the goal is to distinguish between these two cases. 
In this language, versions of 
the PCP theorem assert that the problem gap-LabelCover$[c,s]$ 
is NP-hard, and there are a few parameters of
interest:
\begin{enumerate}
    \item Completeness - the completeness of the PCP is the parameter $c$, and it
    should be as large as possible. In 
    this paper we will always have perfect completeness, that is, $c=1$.
    \item Soundness - the soundness of a PCP is the parameter 
    $s$, and one wants it to be as small as possible.
    \item Alphabet size - the alphabet size is  
    $\max(|\Sigma_L|, |\Sigma_R|)$, and one often wants it 
    to be of constant size.
    \item Instance size - the instance size is the blow-up of the reduction showing that 
    gap-LabelCover$[c,s]$ is NP-hard. The assertion
    of NP-hardness means that there is a polynomial time reduction mapping 
    $3$-CNF formulas $\phi$ to instances of label cover $\Psi$
    such that: if $\phi$ is satisfiable, then ${\sf val}(\Psi)\geq c$, and if $\phi$ is unsatisfiable, then 
    ${\sf val}(\Psi)\leq s$. Letting the size of the $3$-CNF
    formula $\phi$ be denoted by $n$, the size of the PCP is 
    measured by the size of $\Psi$ as a function of $n$.
\end{enumerate}
The original proof of the PCP Theorem~\cite{FGLSS,AS,ALMSS} was able to achieve perfect completeness, soundness $s\leq 1-\eps$ for some 
absolute (but tiny) constant $\eps>0$, constant size alphabet and polynomial instance size. Using the parallel repetition theorem of Raz~\cite{Raz,Holenstein,Rao}, one is able to get a stronger version of the PCP theorem wherein the soundness parameter $s$ can be taken to be an arbitrarily small constant, while incurring a polynomial size blow up. The question of whether
there are other soundness amplification techniques with milder instance size increase has received a significant amount of attention over the years, and is often referred to as the “derandomized parallel repetition problem” in the literature.

\subsection{Near-linear Size PCPs}
Following the proof of the PCP theorem, 
the question of consructing size efficient
PCPs has naturally emerged.
Polishchuk and Spielman~\cite{polishchuk1994nearly}
were the first to construct nearly linear size 2-query PCPs, and their construction has soundness $1-\eps$ and size $n^{1+c(\eps)}$, where $c(\eps)$
approaches $0$ as $\eps$ tends to $0$. 

In her combinatorial proof of the PCP theorem, Dinur~\cite{Dinur07} established a quasi-linear version of the PCP theorem, namely a 2-query PCP with size $n\cdot{\sf poly}(\log n)$, soundness $s= 1-\Omega(1)$ and constant alphabet size. Her proof used a novel gap-amplification procedure for PCPs via graph powering. This procedure can be interpreted as a form of ``derandomized'' parallel repetition, that decreases the soundness (like parallel repetition) but only increases the instance size mildly. Bogdanov~\cite{bogdanov2005gap} proved that the soundness achievable by Dinur's approach plateaus at $1/2$, and in particular the graph powering approach cannot lead to arbitrarily small soundness. Other barriers towards derandomizing parallel repetition were also shown by~\cite{FeigeKillian,MoshkovitzNoGo}. These barriers suggest that new techniques are necessary for constructing nearly-linear size 2-query PCPs with arbitrarily small soundness.

Moshkovitz and Raz~\cite{MoshkovitzR08} were the first to construct such PCPs, and they used a different approach based on low-degree testing and a novel form of composition. Specifically, they proved the hardness of Label Cover with soundness $s = \frac{1}{k^{\Omega(1)}}$, size $n\cdot 2^{\Theta(\sqrt{\log n})}$ and alphabet $2^{\Theta(k)}$, for any $k\leq \log n$. The most notable feature of this work is that it allows one to even get sub-constant soundness, namely one that vanishes with the instance size (at the price of having quite  a large alphabet). The
blow-up in size of the PCP is $2^{\Theta(\sqrt{\log n})}$ though, which is larger than the poly-logarithmic blowup in Dinur's PCP.

The above discussion brings us to the main result of this paper:
\begin{theorem}\label{thm:main}
    For all $\delta>0$, there is $C = C(\delta)>0$ and a 
    polynomial time procedure  that given an instance 
    $\phi$ of $3$-SAT of size $n$, produces a label cover 
    instance $\Psi$ with the following properties:
    \begin{enumerate}
    \item The size of $\Psi$ is at most $n (\log n)^{C}$ and 
    the alphabet size of $\Psi$ is at most $O_{\delta}(1)$.
    \item If $\phi$ is satisfiable, then ${\sf val}(\Psi) = 1$.
    \item If $\phi$ is unsatisfiable, then ${\sf val}(\Psi)\leq \delta$.
    \end{enumerate}
\end{theorem}
\noindent In words, Theorem~\ref{thm:main} gives a version of the PCP theorem of Dinur~\cite{Dinur07} in the low soundness regime.
\vspace{-2ex}
\paragraph{The structure of the proof of 
Theorem~\ref{thm:main}:}
the proof of Theorem~\ref{thm:main} has three  components. The first component involves reducing an arbitrary 2-query PCP, viewed as a 2-CSP, to a 2-CSP whose constraint graph underlies an HDX. To achieve this, we build on the techniques in~\cite{polishchuk1994nearly,DinurMeir} and show a novel connection between PCPs and fault-tolerant distributed computing: one can embed an arbitrary 2-CSP on a prescribed graph $G$ while maintaining the soundness, provided that $G$ has fault-tolerant routing protocols. We then show such routing protocols for HDX, using their spectral expansion properties and the presence of numerous well-distributed dense subgraphs. This is the primary contribution of our work, and we draw inspiration from fault-tolerant distributed computing~\cite{Dwork, Upfal, Chandran, chandran2012edge, JRV}.\footnote{As a by-product, we improve upon the best-known construction of sparse networks with edge-fault-tolerant protocols of~\cite{chandran2012edge}.} The second component uses a size-efficient direct product tester from HDX, which was conjectured to exist in~\cite{DinurK17} and recently established in a series of works~\cite{BafnaMinzer, BLM24, DiksteinD-agreement, DDL24}. This allows us to perform derandomized parallel repetition for 2-CSPs on an HDX, and reduce their soundness to a small constant close to $0$ while incurring a poly-logarithmic blow-up in size. Combining the first and second component we 
get a quasi-linear PCP with small soundness, albeit 
with a large alphabet. Thus, the third and final component is the classical alphabet reduction technique, and we use the results of~\cite{MoshkovitzR08, DinurH13} to reduce the alphabet size to a constant, thus concluding the proof of \Cref{thm:main}.

\subsection{Implications of Our Results }\label{sec:implications}
Many hardness of approximation results in the literature start off with the hardness of Label Cover with small soundness, which is typically achieved by parallel repetition or by appealing to the result of~\cite{MoshkovitzR08}. We are able to 
reproduce any such result that does not use any other features of parallel repetition (such as the so-called covering property), and obtain a better run-time lower bound on approximation algorithms under the exponential time hypothesis (ETH) of~\cite{IP}. First, using the work of H\r{a}stad~\cite{Hastad01}, we have the following two corollaries.
\begin{corollary}\label{cor:3SAT}
    Assuming ETH, for any $\eps>0$ there exists $C>0$ such that solving gap-3SAT$[1,7/8+\eps]$ requires time at 
    least $2^{n/\log^C n}$.\footnote{We remark that in his proof of the hardness of $3$-SAT~\cite[Lemma 6.9]{Hastad01}, H\r{a}stad uses a feature of the outer PCP construction that our PCP lacks. However, as observed by Khot~\cite{KhotSmooth}, 
a weaker property called ``smoothness'' suffices
for the analysis of the reduction, which our PCP construction in Theorem~\ref{thm:main} has. This is because in the end we compose with the Hadamard code and the associated manifold versus point test, which is smooth.}
\end{corollary}

In the 3-LIN problem one is given a system of
linear equations over $\mathbb{F}_2$ where each
equation contains 3 variables, and the goal is
to find an assignment satisfying as many of the constraints as possible. We have:
\begin{corollary}
    Assuming ETH, for any $\eps>0$ there exists $C>0$ such that solving gap-3LIN$[1-\eps,1/2+\eps]$ requires time at 
    least $2^{n/\log^C n}$.
\end{corollary}
Next, the set cover problem admits dynamic programming based algorithms that solve it exactly in $O^*(2^n)$ time. The works of~\cite{BourgeoisEP09,CyganKW09} showed that it can be approximated to a $(1+\ln r)$-factor in $O^*(2^{n/r})$ time for every rational $r\ge 1$. By applying the label cover transformation in \cite{Moshkovitz15} to~\Cref{thm:main} and then applying Feige’s reduction to the set cover problem \cite{Feige98}, we get the 
following hardness result:
 \begin{corollary}
Assuming ETH, for every rational $r>1$, there is some constant $C$ such that no algorithm running in time $2^{n/(\log n)^C}$ can approximate the set cover problem to a factor better than $r$.
\end{corollary} 
The above corollary also implies similar conditional run-time lower bound for approximation algorithms for the max coverage problem with a $(1-1/e)$-approximation factor and various other problems which can be reduced from the set cover problem (such as clustering; see \cite{GK99}).
\skipi
Our proof techniques can also be 
used to obtain improvements for the almost-everywhere reliable transmission problem~\cite{Dwork}. This problem is very similar to the routing problem (given in Definition~\ref{def:adv-comm-protocol}) which is central to this paper. The almost everywhere reliable transmission problem involves designing a sparse graph $G$ along with communication protocols between all pairs of vertices of $G$, that are resilient against vertex/edge corruptions. Our results nearly match the performance of the best-known protocols for vertex corruptions~\cite{Chandran, JRV}, while addressing the more challenging scenario of edge corruptions, thus improving upon the results of~\cite{chandran2012edge}. More specifically, in Section~\ref{sec:cor-proof} we show:
\begin{theorem}\label{thm:edge-routing}
For all \( n \in \mathbb{N} \), there exists a regular graph \( G = (V, E) \) on \( \Theta(n) \) vertices with degree $\pl n$ and \( O(\log n) \)-round protocols \(\{ \mathcal{R}_{u, v}\}_{u,v\in V} \) on it such that, for all adversaries corrupting \( \eps |E| \) edges, all but $O(\eps)$ of the message transfers between pairs of vertices $(u, v)$ will be successful. Further, running a single $\mathcal{R}_{u, v}$ protocol can be done in $\pl n$ computation across all nodes.
\end{theorem}

The rest of this introductory section is 
organized as follows. In Section~\ref{sec:history} we discuss prior work on gap amplification for PCPs. In Section~\ref{sec:ae-intro}
we discuss the almost-everywhere reliable transmission problem from distributed computing and in Section~\ref{sec:techniques} we 
discuss our techniques.

\subsection{History of Gap Amplification}\label{sec:history}
\subsubsection{Parallel Repetition}
The parallel repetition theorem of Raz~\cite{Raz,Holenstein,Rao} is a powerful technique
in hardness of approximation and interactive protocols. Most relevant to us is its application to PCPs, wherein it is used to boost the soundness of a given PCP construction. Indeed, given 
a label cover instance $\Psi$ and an integer $t\in\mathbb{N}$,
we consider the $t$-fold repeated game $\Psi^{\otimes t}$ which consists of
the graph $G_t = (L^t\cup R^t, E_t)$ where 
\[
E_t = \{((u_1,\ldots,u_t),(v_1,\ldots,v_t))~|~(u_i,v_i)\in E,\forall i=1,\ldots,t\},
\]
as well as alphabets $\Sigma_L^t, \Sigma_R^t$ and constraints 
$\Phi' = \{{\Phi'}_e\}_{e\in E_t}$ where
\[
\Phi_{(\vec{u},\vec{v})}
=\{((\sigma_1,\ldots,\sigma_t),(\tau_1,\ldots,\tau_t))~|~(\sigma_i,\tau_i)\in \Phi_{(u_i,v_i)}~\forall i=1,\ldots,t\}.
\]
It is clear that if ${\sf val}(\Psi) = 1$ then 
${\sf val}(\Psi^{\otimes t})=1$, and the content of the
parallel repetition theorem asserts that if ${\sf val}(\Psi)\leq 1-\eps$, then ${\sf val}(\Psi^{\otimes t})\leq (1-\eps')^t$, where $\eps' = {\sf poly}(\eps)$. Thus, parallel repetition can be used to decrease the soundness to be 
as close to $0$ as we wish. However, as ${\sf size}(\Psi^{\otimes t}) = {\sf size}(\Psi)^t$, parallel repetition cannot be used on a given PCP construction to get small
soundness while maintaining nearly-linear size.

In light of this issue, it makes sense to 
look for sparse analogs
of parallel repetition that still amplify soundness. This task is known as derandomizing parallel repetition in the literature -- given a label cover instance
$\Psi$ and an integer $t$, one would like to come up with
sparse subsets of $L^{t}$ and $R^t$ 
such that the induced label over instance $\Psi^{\otimes t}$ on them would still have significantly smaller
soundness than the original instance $\Psi$. Ideally, one would like the subsets to be 
as small as $O_t(|L|)$, $O_t(|R|)$ while 
retaining arbitrarily small constant soundness
$s$. Towards this end, it makes sense to consider the simpler combinatorial analog of this question known as \emph{direct product testing}, which we define next.

\subsubsection{Direct Product Testing}
In an effort to simplify the proof of
the PCP theorem, Goldreich and Safra~\cite{GoldreichSafra} introduced the notion of direct product testing. 
In direct product testing, one
wishes to encode a function $f\colon [n]\to \Sigma$ (which, in the context of PCPs is thought of as an assignment) via local views in a way that admits local testing. The most natural direct product encoding,
which we refer to as the Johnson direct
product scheme, has a parameter $k\in\mathbb{N}$ which is thought of as a large constant. The function 
$f$ is encoded via the assignment $F\colon \binom{[n]}{k}\to\Sigma^k$ defined as 
$F(\{a_1,\ldots,a_k\}) = (f(a_1),\ldots,f(a_k))$.\footnote{We fix an arbitrary ordering on $[n]$.} The natural 2-query direct product
test associated with this encoding is the 
following consistency check:
\begin{enumerate}
    \item Sample $B\subseteq [n]$ of size $\sqrt{k}$.
    \item Sample $A,A'\supseteq B$ independently
    of size $k$.
    \item Read $F[A]$ and $F[A']$ and check
    that $F[A]|_{B} = F[A']|_{B}$.
\end{enumerate}
It is clear that if $F$ is a 
valid encoding of a function $f$, then the
tester passes with probability $1$. The work of~\cite{DinurG08} shows that this test is also sound. Namely, for all $\delta>0$, taking $k\in\mathbb{N}$ large enough, if an assignment $F\colon \binom{[n]}{k}\rightarrow\Sigma^k$ passes
the Johnson direct product test with probability at least $\delta$, 
then there is a function $f'\colon [n]\to\Sigma$ such that
\begin{equation}\label{eq1_intro}
\Pr_{A}
[\Delta(F[A], f'|_{A})\leq 0.01]
\geq \poly(\delta),
\end{equation}
where $\Delta(F[A], f'|_{A}) = \frac{1}{|A|}\#\{i\in A~|~F[A](i)\neq f'(i)\}$ is the relative Hamming distance between $F[A]$ and $f'|_{A}$.

The main drawback of the Johnson direct product encoding is its size blow up: the size
of the encoding of $F$ is roughly the size
of $f$ to the power $k$. This is the same 
behaviour that occurs in parallel repetition, which raises a combinatorial analog of the derandomized parallel repetition problem: 
is there a sparse collection of $k$-sets $\mathcal{S}_k\subseteq \binom{[n]}{k}$, ideally $|\mathcal{S}_k| = O_k(n)$, such that the encoding of $f\colon [n]\to \Sigma$ given by $F\colon \mathcal{S}_k\to\Sigma^k$ defined 
as $F(\{a_1,\ldots,a_k\}) = (f(a_1),\ldots,f(a_k))$, admits a sound $2$-query consistency test?

\subsubsection{Graph Powering}
Underlying the proof of Dinur~\cite{Dinur07} is a derandomized direct product tester in the 99\% soundness regime. Dinur proved that the set system formed by constant size neighbourhoods of vertices in a spectral expander supports a natural $2$-query test with soundness bounded away from $1$.\footnote{Although her original result is not stated in terms of direct product testing, a later work of~\cite{DiksteinD19} observes that her proof indeed implies such a tester.} She used these ideas to show a gap amplification result for PCPs achieved via graph powering. To present it, we expand the definition of label cover to graphs that are not necessarily bipartite.
\begin{definition}
A constraint satisfaction problem 
(2-CSP in short) $\Psi = (G=(V,E), \Sigma, \Phi)$ is composed of a graph $G$, an alphabet $\Sigma$, and a collection of constraints $\Phi = \{\Phi_e\}_{e\in E}$, one for each edge. Each constraint is $\Phi_e\subseteq \Sigma\times \Sigma$, describing the tuples of labels to the endpoints of $e$ that are considered satisfying.
\end{definition}

Dinur started with an instance of a 2-CSP $\Psi$ over a $d$-regular expander graph $G$ with soundness $1-1/\pl n$ and constant alphabet. Given a parameter $t\in\mathbb{N}$, she considered the 2-CSP 
$\Psi'$ over the graph $G'$ whose vertex set is $V$ and $u,v$ are adjacent if they have a path of length $t$ between them in 
$G$. The alphabet of a vertex $u$ is $\Sigma^{d^{t/2}}$, 
which is thought of as an assignment to the neighbourhood
of $u$ of radius $t/2$. The allowed symbols on $u$ are assignments to the neighbourhood of $u$ satisfying all of the constraints of $\Psi$ inside it. Finally, the constraint between 
$u$ and $v$ is that the assignment they give to their neighbourhoods agree on any common vertex to them.

She showed that if most constraints of $\Psi'$ are satisfied (equivalently, the direct product test passes), the assignment to the neighborhoods must be consistent with a global assignment to the vertices, which in turn must satisfy a large fraction of the edges of $\Psi$. This implies that if ${\sf val}(\Psi)\leq 1-\delta$, then 
${\sf val}(\Psi')\leq 1-\min(2\delta, c)$ where $c>0$
is an absolute constant. Thus, each invocation of graph powering improves the soundness of the 2-CSP,
and after $\Theta(\log(1/\delta))$ iterations the resulting 2-CSP will have value at most $1-\Omega(1)$ and quasi-linear size. Each iteration of graph powering blows up the alphabet size though, which Dinur resolved via an alternating step of alphabet reduction.

\subsubsection{The Soundness Limit of Graph Powering}
Following Dinur's result~\cite{Dinur07}, Bogdanov~\cite{bogdanov2005gap} observed 
that graph powering on an arbitrary expander fails to decrease the soundness below $1/2$. Towards this end,  he considers any locally-tree-like expander graph $G$ with large girth, such as the construction in~\cite{lubotzky1988ramanujan}, and defines a CSP $\Psi$ over $G$ whose alphabet is $\Sigma = \{0,1\}$, and constraints are inequality constraints.
The graph $G$ has $n$ vertices, girth $g = \Theta(\log n)$, and the largest cut in it has fractional size $1/2 + o(1)$. The latter fact implies that ${\sf val}(\Psi)=1/2+o(1)$. On the other hand, as long as $t< g/2$, for each vertex $v$ the $t/2$-neighborhood of $v$ has two possible assignments. If $u$ and $v$ are within distance $t$, then there is a $1$-to-$1$ correspondence between the assignments of $u$ and the assignments of $v$. Thus, randomly choosing one of these possible assignments for each $u$ leads to an assignment that satisfies $1/2$ of the constraints in $G^{t}$ in expectation, and in particular ${\sf val}(G^t)\geq 1/2$. This means that graph powering fails to decrease the soundness below $1/2$.

\subsubsection{Subspace-based Direct Product Testers}
The work of Impagliazzo, Kabanets and Wigderson~\cite{ImpagliazzoKW09} made progress on derandomized direct product testing in the 1\% regime by analyzing a more efficient, albeit still polynomial size version of the Johnson direct product tester based on subspaces, which we refer to as the Grassmann direct product tester. 
In this context, we identify the universe $[n]$ with $\mathbb{F}_q^d$ 
(where $\mathbb{F}_q$ is a field), so that an assignment $f\colon [n]\to \Sigma$ is interpreted as a function 
$f\colon \mathbb{F}_q^d\to \Sigma$. The Grassmann direct product encoding of $f$ is a table $F$ that specifies the restriction of $f$ to each 
$d'$-dimensional subspace of $\mathbb{F}_q^d$, 
where $d'<d$ is a parameter chosen appropriately.
The Grassmann direct product test is the natural consistency check, wherein one chooses 
$B$ of prescribed dimension, then two
$d'$-dimensional subspaces $A$ and $A'$ containing $B$, and checks that $F[A]$
and $F[A']$ agree on $B$.
Their work proves that this test has small soundness in a sense similar to~\eqref{eq1_intro}.

The work~\cite{ImpagliazzoKW09} also strengthened the connection between gap amplification and direct product testing. Namely, 
they showed an amplification procedure for PCPs using the Johnson direct product tester. Just like parallel repetition though, this procedure
incurs a polynomial blow-up in the instance
size. One could hope to use the
more efficient Grassmann direct product encoding to improve upon this size blow-up. This turns out to be harder though, since using the Grassmann direct product tester effectively  requires a $2$-CSP whose constraint graph is  compatible with the structure of subspaces. At the very least, a $d'$-dimensional subspace must contain multiple edges of the initial 2-CSP.

\subsubsection{Derandomized Parallel Repetition using Subspace-based Direct Product Testers}\label{sec:subspace_derandomize}
Dinur and Meir~\cite{DinurMeir} obtained this
compatibility by reducing an arbitrary 2-CSP with soundness $1-\Omega(1)$ to a 2-CSP on the De-Bruijn graph with soundness $1-\Omega(1/\log n)$. The benefit of switching to a De-Bruijn graph $G$ is that, if $V(G)$ is denoted by $\F_q^d$, then the set of edges $E(G)$ forms a linear subspace of $\F_q^{2d}$. This made $G$ compatible with the Grassmann direct product tester, which then allowed the authors to amplify the soundness of any 2-CSP on $G$ to soundness close to $0$.

The idea of embedding CSPs on De-Bruijn graphs first appeared in a similar context in~\cite{BFLS,polishchuk1994nearly}. Towards this end, these works used the property that De-Bruijn graphs have efficient routing protocols. The work of~\cite{DinurMeir} demonstrates that, in some cases, derandomized direct product testing theorems can be used to construct PCPs. Their construction however falls short of getting size efficient PCPs since the Grassmann direct product encoding itself has a polynomial size blow-up.

\subsubsection{Linear Sized Direct Product Testers from HDX}

The work~\cite{DinurK17} identified 
high dimensional expanders as a potential
avenue towards constructing size efficient
direct product testers. Roughly speaking, high-dimensional expansion is a generalization of the usual notion of spectral expansion in graphs. An HDX is an expander graph $G$ which in addition has a lot of well-connected cliques of constant size. In this work we restrict our attention to clique complexes, and  denote by $X(k)$ the collection of cliques in $G$ of size $k$ and by $X=\{X(k)\}_{k=1}^d$ the $d$-dimensional clique complex of $G$. The work~\cite{DinurK17} showed that for a sufficiently good HDX, the natural direct product encoding using the set system $X(k)$ admits a natural 2-query tester with soundness 
$1-\Omega(1)$. Thus, as there are known constructions of constant degree high-dimensional expanders, their result gave a derandomized direct-product testing result in the $99\%$ regime based on HDX.
They conjectured that there are 
sparse high-dimensional expanders that support a direct product tester in the more challenging, 1\% regime of soundness.

Motivated by this problem and potential applications to PCPs, the works~\cite{BafnaMinzer,BLM24,DiksteinD-agreement,dikstein2023swap,DDL24} studied
the soundness of the natural direct product tester on $X(k)$. First, the works~\cite{BafnaMinzer,DiksteinD-agreement} identified that for the test to have
a small soundness, it is necessary that the HDX is a ``coboundary expander'', which is a topological notion of expansion. In particular any graph which is locally tree like, such as the one used by Bogdanov~\cite{bogdanov2005gap}, is automatically not a coboundary expander, giving a more comprehensive explanation to his counterexample. Second, the works~\cite{BafnaMinzer,DiksteinD-agreement} 
established that sufficiently strong high-dimensional
expansion and coboundary expansion imply 
that $X(k)$ admits a direct product tester with 
small soundness. 
Following this, the
works~\cite{BLM24,dikstein2023swap,DDL24} constructed sparse HDX, that are also  coboundary expanders. Their construction is a variant of the Chapman-Lubotzky complex~\cite{ChapmanL} with an appropriate choice of parameters. In particular, 
they established that:
\begin{theorem}\label{thm:dp_intro}
For all $\delta>0$, there is $\Delta\in \N$ such that for all large enough $k,d\in \N$ the following holds. There is an infinite sequence of $d$-dimensional clique complexes $\{X_n\}_{n\in N}$ with degree $\Delta$ such that the set system $X_n(k)$ admits a 2-query direct product test with soundness $\delta$.\footnote{We remark that in the works~\cite{BafnaMinzer,BLM24}, 
only the case of Boolean alphabets $\Sigma = \{0,1\}$ was considered. Their proof however is easy to 
adapt to any alphabet, and in Section~\ref{sec:dp-large-alph} we explain these adaptations for the sake of completeness.}
\end{theorem}

With Theorem~\ref{thm:dp_intro} in hand, one may hope to obtain a derandomized parallel repetition result, just like Dinur and Meir~\cite{DinurMeir} turned the tester of~\cite{ImpagliazzoKW09} into a small soundness PCP. Thus, one again wishes to be able to convert a given CSP $\Psi$ into a CSP $\Psi'$ whose underlying graph is compatible with the HDXs from Theorem~\ref{thm:dp_intro}. 
As we elaborate below, this task entails several challenges, and the main contribution
of the current work is to perform such a conversion via a novel connection to the almost-everywhere reliable transmission problem, and the construction of \emph{fault-tolerant} routing protocols on HDX.

\subsection{Almost-Everywhere Reliable Transmission}\label{sec:ae-intro}
In the almost everywhere reliable transmission problem from~\cite{Dwork},\footnote{For ease of application, our formalization slightly deviates from the standard one used in the literature.} the goal is to design a sparse graph $G = (V,E)$, that allows transference of messages in a fault tolerant way.  More precisely, for any permutation $\pi\colon V\to V$ and an alphabet $\Sigma$, the goal is to design a $L$-round protocol; at every round of the protocol each node can send symbols from $\Sigma$ to its neighbours in $G$, and at the end of the protocol the message of $v$ should be transmitted to $\pi(v)$ for most $v$. This guarantee should hold even if a small fraction of vertices or edges in the graph behave maliciously. The parameters of interest in this problem are as follows:
\begin{enumerate}
\item \textbf{Work Complexity:} the work complexity of a protocol is defined as the maximum computational complexity any node in the graph $G$ incurs throughout the protocol.
\item \textbf{Degree:} this is the degree of $G$, which we aim to minimize.
\item \textbf{Tolerance:} a protocol is considered $(\eps(n), \nu(n))$-vertex (edge) tolerant if, when an adversary corrupts up to $\eps(n)$-fraction of vertices (edges) of the graph (allowing them to deviate arbitrarily from the protocol), at most $\nu(n)$-fraction of the transmissions from $u\rightarrow \pi(u)$ are disrupted. 
\end{enumerate}

With this in mind, the simplest protocol one can design on is a ``pebble-routing protocol'' by showing that for any permutation $\pi:V\rightarrow V$, there exists a set of $L$-length paths in $G$ from $u \rightarrow \pi(u)$ such that at any time step every vertex is used exactly once across all the paths. This protocol has work complexity $O(L)$ and is $(\eps,\eps L)$-vertex-tolerant for all $\eps >0$. The work of Dwork et al.~\cite{Dwork} shows that the constant-degree ``butterfly network'' has pebble-routing protocols with $L=O(\log n)$, thus giving a network that can tolerate $\eps \lll \frac{1}{\log n}$-fraction of vertex corruptions. All protocols which improve upon this result do so by considering more complicated protocols that use error correction. The first such protocol which could tolerate a linear number of vertex corruptions was given by~\cite{Upfal}, but it has $\exp(n)$ work complexity. In~\cite{Chandran} a routing network with ${\sf poly}(\log n)$ degree and $O(n)$ work complexity is constructed, and in~\cite{JRV} a routing network with $O(\log n)$ degree and ${\sf poly}(\log n)$ work complexity. As per tolerance, the last two works show that for all $\eps$ less than some universal constant $c>0$, their network is $(\eps, \eps + O(\frac{\eps}{\log n}))$-vertex tolerant.

\subsection{Our Techniques}\label{sec:techniques}
We begin by describing the connection between routing protocols and PCPs. We then discuss our construction of routing protocols over HDX and elaborate on the rest of the proof of \Cref{thm:main}.
\subsubsection{The Connection between Fault-Tolerant Routing Protocols and PCPs}\label{sec:overview:conn}
The works of~\cite{polishchuk1994nearly,DinurMeir} both used the fact that the De-Bruijn graph $G$ on $n$ vertices has a pebble-routing protocol of length $O(\log n)$ to transform the constraint graph of any 2-CSP to the De-Bruijn graph.  Their argument proceeds as follows. Suppose that we have a CSP instance $\Psi$ over a $d$-regular graph $H$ on $n$ vertices. First, break the graph $H$ into a union of $d$ disjoint perfect matchings $\pi_1,\ldots,\pi_d\colon [n]\to [n]$,\footnote{Strictly speaking, we can only guarantee that this is possible to do in the case that $H$ is bipartite, and we show an easy reduction to this case.} thought of as permutations on $V(H)$. For each permutation $\pi_i$, we have a pebble-routing protocol $\cP_i$ on the De-Bruijn graph, where a vertex $v$ transmits its supposed label to $\pi_i(v)$ along the path from $v\rightarrow\pi_i(v)$. In the new CSP $\Psi'$, the alphabet of 
each vertex $v$ consists of the messages 
that it sent and received throughout each one of the $d$ routing protocols. The constraints of 
$\Psi'$ over $(u,v)\in E(G)$ are that for all $t$, the label sent by $u$ at step $t$ matches the label received by $v$ from $u$ at step $t+1$. Thus, for each $i$ the vertex $\pi_i(v)$ will also
have as part of their assignment the symbol 
they supposedly got from $v$ at the 
end of the routing protocol, and we only
allow labels that satisfy the constraint 
$(v,M_i(v))$ in $\Psi$. It is easy to see that if $O(\eps/d'\log n)$-fraction of the edges of $\Psi'$ are violated, where $d'$ is the degree of $G$, then at most $\eps$-fraction of the paths are unsuccessful. As $d' = O(1)$, this implies that
\begin{enumerate}
    \item If ${\sf val}(\Psi) = 1$, then ${\sf val}(\Psi') = 1$.
    \item If ${\sf val}(\Psi)\leq 1-\eps$, then
    ${\sf val}(\Psi')\leq 1-\Omega\left(\frac{\eps}{\log n}\right)$.
    \item If the alphabet of $\Psi$ is of constant size, then the
    alphabet of $\Psi'$ is of polynomial size.
\end{enumerate}
Note that any constant degree graph (including the graph underlying $X_n$) can at best have a pebble-routing protocol of length $\Theta(\log n)$, which is the source of the $\log n$-factor loss in soundness above. This loss is unaffordable to us since the subsequent gap amplification via the direct-product tester in \Cref{thm:dp_intro} would lead to a huge size blow-up. 

Our main observation is that one can use protocols outside the scope of pebble-routing protocols and still get a connection with PCPs. By that, we mean that we can do a similar transformation given any routing protocol on $G$ as long as it is efficient. More precisely, we establish the following general connection:
\begin{lemma}\label{lem:routing-to-pcp-general}
Suppose $G$ is a regular graph on $2n$ vertices such that for all permutations $\pi:V(G)\rightarrow V(G)$, $G$ has a $(\eps,\nu)$-edge-tolerant protocol with work complexity $W$ that can be constructed in time $\poly(n)$. Then there is a polynomial time reduction that given any 2-CSP $\Psi'$ on a $k$-regular graph $H$ with $|V(H)|\leq n$ produces a 2-CSP $\Psi$ on $G$ such that,
\begin{itemize}
\item If $\val(\Psi')=1$ then $\val(\Psi)=1$.
\item If $\val(\Psi')\leq 1-O(\nu)$ then $\val(\Psi)\leq 1-\eps$.
\item If the alphabet of $\Psi'$ is $\Sigma$ then the alphabet of $\Psi$ is $\Sigma^{kW}$. 
\end{itemize}
\end{lemma}
We can now hope to instantiate the above lemma with routing protocols for HDX that are tolerant to a constant fraction of edge-corruptions to embed a 2-CSP onto an HDX while maintaing the soundness.

\subsubsection{Challenges in Using Existing Routing Protocols for PCPs}
With Lemma~\ref{lem:routing-to-pcp-general} in mind, the result of~\cite{JRV} achieves the type of 
parameters we are after. There are two significant differences between their setting and ours, though:
\begin{enumerate}
\item \textbf{Picking the graph:} in the setting of the almost everywhere reliable transmission problem, one is free to design the routing graph $G$ as long as it has the desired parameters. In our case, we need to use the graphs $G$ underlying complexes that supports a direct product tester with low soundness, for which we currently know of only one example -- the complexes from Theorem~\ref{thm:dp_intro}.
\item \textbf{The corruption model:} The prior works discussed in \Cref{sec:ae-intro} have mostly considered the vertex corruption model, whereas to apply \Cref{lem:routing-to-pcp-general} one must design protocols resilient in the more challenging \emph{edge corruption} model. This model was studied by~\cite{chandran2012edge} who gave a protocol with work complexity $n^{\Theta(1)}$, which is again unaffordable to us. 
\end{enumerate}

The main new tool in the current paper is  efficient routing
protocols on high dimensional
expanders in the edge corruption model. Below, we elaborate on the two protocols that we 
give in this paper.

\subsubsection{Clique Based Routing Network on HDX} 
A key idea introduced by~\cite{Chandran} (also used by~\cite{JRV}) is that to construct tolerant routing protocols on a graph $G$, it is beneficial for it to contain many well inter-connected large cliques. In these cliques, messages
can be transmitted very quickly, and they can be used for error correction. As the complexes
$X_n$ from Theorem~\ref{thm:dp_intro} naturally have many well-connected cliques, it stands
to reason that they too admit routing protocols with similar performance to the protocols in~\cite{JRV}.

Indeed, this turns out to be true; see Section~\ref{sec:clique_to_clique} for a formal statement of our clique-based routing scheme. Our analysis of this scheme though requires the complex $X_n$ to have dimension $d = \Omega((\log\log n)^2)$ and to achieve that, the complexes from Theorem~\ref{thm:dp_intro} have
size $n d^{C\cdot d^2} = n2^{{\sf poly}(\log \log n)}$, which falls short of being quasi-linear.
We therefore resort to an alternative routing scheme, based on yet another highly connected structure found in HDX -- links.

\vspace{-1pt}

\subsubsection{Link-based Routing Network on HDX}\label{sec:overview_ltol}
Our link-based routing scheme requires milder expansion properties from $X$, and these properties can be achieved by the complexes from Theorem~\ref{thm:dp_intro} with quasi-linear size. More concretely, take a complex $X = X_n$ as in Theorem~\ref{thm:dp_intro} with spectral expansion $\gamma = (\log n)^{-C}$ and dimension $d = C$ for a large constant $C>0$.\footnote{In Theorem~\ref{thm:dp_intro} the degree is thought of as a constant, giving an HDX with arbitrarily small (but constant) expansion. In Section~\ref{app:cl}, it is shown how to modify this construction to achieve better spectral expansion at the expense of having poly-logarithmic degree.} Let $G=(X(1),X(2))$ be the underlying graph with $n$ vertices and $n\pl n$ edges. We show how to use the 
graph $G$ for routing, and towards this end we use the following properties of $X$:
\begin{enumerate}
    \item 
    {\bf Expansion:} the complex $X$ is a $d$-partite one-sided $\gamma$-local spectral expander. We defer the formal definition of high-dimensional expansion to Section~\ref{sec:prelim}, but for the
    purposes of this overview it is sufficient
    to think of this as saying that many graphs
    associated with $X$, such as the underlying 
    graph $G$ as well as the underlying graph of each link $L_u$, are partite graphs with second largest eigenvalue at most $\gamma$. Here and throughout, the link of $u$ refers to the complex consisting of all faces of $X$ containing $u$, where we omit the vertex $u$ from them.
    \item 
    {\bf Very well connected links:} the links of $X$ are highly connected. By that, we mean that inside every vertex-link $L$ we can set up a collection of constant length paths inside it satisfying: (a) the collection includes a path between almost all pairs of vertices 
    in $L$, and (b) no edge is overly used.
\end{enumerate} 
At a high level, the first property implies that links sample the vertices of $G$ very well and the second property allows us to pretend that links behave very closely to cliques. With this in mind, we now present an overview of our
link-based routing network. 
We first consider the simplified case that 
the graph $G$ is regular; even though our graph $G$ is not regular, it is a helpful case to consider and we later explain how to remove this assumption. 

{\bf Formalizing the setup:} 
Let $V=X(1)$ and fix a permutation $\pi\colon V\to V$. In our setup, for each vertex $v$ in $X$, the vertices in its link $L_v$ hold a common message $m_v$. Denoting by $\Delta$ the degree of $G$, we note that every vertex holds $\Delta$ symbols, and for each one of these it knows which link the message is associated with. By the end of the protocol on $G$, we want that for at least $1-1/\pl n$-fraction of $v$, the majority of vertices in $L_{\pi(v)}$ hold $m_v$. 

{\bf Setting up external paths:}
Recall the 
graph $G = (X(1),X(2))$. 
By the spectral expansion of 
$X$ it follows that the graph $G$ is an expander. At this point we can use the 
results of~\cite{AlonCG94,Nenadov23}, which show $O(\log n)$-length pebble-routing protocols for regular, sufficiently good spectral expanders\footnote{In reality we use a simple deterministic polynomial time algorithm to find paths with much weaker guarantees than a pebble-routing protocol (see Theorem~\ref{thm:expander-routing}). This is because while it is likely that the constructions of~\cite{AlonCG94,Nenadov23} can be made algorithmic and deterministic, the conclusion of \Cref{thm:expander-routing} suffices for our purposes.}. 
At each step in the routing protocol each vertex holds multiple symbols (being part of several links), and we stress that it always knows which symbol it is supposed to
transmit at each point of the protocol.

With this set-up we describe the protocol. At a high level, the goal in the $t$-th step is to transmit the message from a link $L_w$ to the link $L_{w'}$ that occur consecutively in a path above. Note that it must be the case that the vertices $w$ and $w'$ are neighbours in $G$, and hence we may consider the link $L_{w,w'}$. 

{\bf The analysis in an idealized ``clique'' setting:} to 
get a sense for the protocol, we first 
pretend that each one of the links $L_w$, 
$L_{w'}$ and $L_{w,w'}$ in fact forms a 
clique. In this case, the $t^{th}$ step of our protocol begins by the vertices in 
$L_{w}$ sending their message to the 
vertices in $L_{w,w'}$. Each vertex in $L_{w,w'}$ computes the majority value among the symbols they received in this step, and then forwards this symbol to all of the vertices 
in $L_{w'}$. Each vertex in $L_{w'}$ also 
computes the majority value among the symbols it receives, at which point the $t^{th}$ step of the protocol is over. 

The key idea behind the analysis of this protocol is that vertices compute the correct value they were supposed to so long as they are not over-saturated with corrupted edges. More precisely, let $\mathcal{E}\subseteq X(2)$ be the set of corrupted edges, which therefore is at most $\eps$ fraction of the edges. We say a $1$-link $L_w$ is good if at most $\sqrt{\eps}$ fraction of the edges in it are corrupted. Let $V_w\subseteq L_w$ be the set of vertices for which at most $\eps^{1/4}$-fraction of the edges adjacent to them in $L_w$ are corrupted; note that if $L_w$ is good then by an averaging argument, $|V_w|\geq (1-\eps^{1/4})|L_w|$. We refer to the vertices outside $V_w$ as the doomed vertices of $L_w$.
For a $2$-link $L_{w,w'}$, we say it is good if both $L_w,L_{w'}$ are good and at most $\eps^{1/8}$-fraction of the
vertices in it are doomed with respect to $L_w$ or $L_{w'}$. Using spectral arguments it is easy to show that a $\left(1-\frac{1}{\pl n}\right)$-fraction of $1$-links and $2$-links are good.

One can now show that if $L_w$, $L_{w'}$ and $L_{w,w'}$ are all good then the majority value is transmitted from $L_w$ to $L_{w'}$. We can conclude the argument by a simple union bound -- since at most $1/\pl n$-fraction of 1-links and 2-links are bad and the external paths $L_v \rightarrow L_{\pi(v)}$ are of length $O(\log n)$, at most $1/\pl n$-fraction paths contain at least one bad link, therefore all but $1/\pl n$-fraction of the transmissions are successful.\footnote{Technically, the number of 2-links is asymptotically greater than the number of paths, therefore this union bound does not work. We handle this issue by moving to the zig-zag product of $G$ where the number of paths is equal to the number of 2-links.}

{\bf Back to the real ``link'' setting:} 
Links are not as well connected as 
cliques, but they turn out to be connected enough
to almost make the above protocol go through. Indeed, the ``clique'' assumption in 
the above paragraph was mainly used to argue
that the transmission from $L_w$ to $L_{w,w'}$
and from $L_{w,w'}$ to $L_{w}$ can each be done cleanly in a single round. 
While this property is no longer true for 
links, we use the ``very well connected links'' property above to circumvent it. At a high level, we simply replace the immediate transmission between a pair of vertices $u,v$ in the above idealized setting, with a short path between $u \rightarrow v$ inside $L_w$. Using the fact that edges in $L_w$ are used uniformly across the internal paths in $L_w$, we can show that very few of these short paths can be corrupted and a similar error analysis for the link to link transfer goes through.

{\bf Lifting the regularity assumption:} Finally we remark that in our actual protocol we use the zig-zag product of~\cite{ReingoldVW00}, a powerful tool from derandomization. It allows us to define a closely related graph $Z$ which is regular and has constant degree. This move is compatible with routing: paths in $Z$ have natural correspondence to paths in $G$, and the uniform distribution over $Z$ corresponds to the stationary distribution over $G$. Although the routing still takes place on $G$, moving to $Z$ allows us to deal with the irregularity of $G$ and with other technical issues that come up in implementing the union bound above.

\begin{remark}\label{remark:intro2}
We note here that the link-to-link routing protocol works for any simplicial complex with the properties stated above -- sufficiently good local spectral expansion of $X$ and the well-connectedness of links of $X$. In particular, one can pick any other complex $X$ with these properties such as the $d$-partite LSV complexes~\cite{LSV1,LSV2} (in fact, we do so in the proof of~\Cref{thm:edge-routing}) and the complexes of Kaufman and Oppenheim~\cite{kaufman2020local}, and still get the same tolerance guarantees as above.
\end{remark}

\subsubsection{Composing the Links Based Routing Protocol with Known PCP Constructions}
Using our link based routing protocol along with \Cref{lem:routing-to-pcp-general} on the existing PCP results gives a quasi-linear size PCP on the graphs underlying the complexes from Theorem~\ref{thm:dp_intro} with soundness bounded away from $1$. 

\begin{theorem}\label{thm:CSP_on_CL_large_soundness2}
There exists $\eps>0$ such that the following holds. Let $\{X_{n'}\}_{n'\in N}$ be the infinite family of clique complexes from Theorem~\ref{thm:dp_intro}. Then for sufficiently large $d\in\mathbb{N}$, there is $C>0$ and a polynomial time reduction mapping a $3$SAT instance $\phi$ of size $n$ to a CSP $\Psi$ over $(X_N(1),X_N(2))$, for some $d$-dimensional $X_N$, such that:
\begin{enumerate}
    \item If $\phi$ is satisfiable, then $\Psi$ is satisfiable.
    \item If $\phi$ is unsatisfiable, then ${\sf val}(\Psi)\leq 1-\eps$.
    \item The size of the graph underlying $\Psi$ is at most
    $n (\log n)^{C}$, the alphabet $\Sigma$ of 
    it satisfies that $\log(|\Sigma|)\leq (\log n)^C$, and the decision complexity of the
    constraints (i.e., the circuit complexity of a circuit
    checking if a given pair of labels satisfies a constraint)
    is at most $(\log n)^C$.
\end{enumerate}
\end{theorem}
We stress that the key point of Theorem~\ref{thm:CSP_on_CL_large_soundness2} is that
the CSP is over the graph 
underlying the complex from Theorem~\ref{thm:dp_intro},
making it potentially compatible with derandomized parallel repetition. 

\begin{remark}\label{remark:intro1}
One can take the size efficient PCP construction 
of Dinur~\cite{Dinur07} as the starting point towards the proof of \Cref{thm:CSP_on_CL_large_soundness2}. Interestingly though, we can strengthen the connection in~\Cref{lem:routing-to-pcp-general} to get an improved version of \Cref{lem:routing-to-pcp-general} as an unexpected, positive side-effect. In \Cref{lem:improved-routing-to-PCP}, we show that one can embed an arbitrary 2-CSP with soundness $1-\frac{1}{\pl n}$ to a 2-CSP on an HDX (supporting a link-to-link protocol) with soundness $1-\Omega(1)$. This provides a 1-shot amplification procedure giving an alternative to the step by step  approach of Dinur~\cite{Dinur07}. The CSP $\Psi$ however has a large alphabet,
so (just like in Dinur's proof) we require
an alphabet reduction procedure to bring
the alphabet back to constant, while not affecting the soundness or size by much. Thus to prove \Cref{thm:CSP_on_CL_large_soundness2} one can also start with a size efficient PCP construction with weaker soundness guarantees, such as that of Ben-Sasson and Sudan~\cite{BensasonSudan}, and apply \Cref{lem:improved-routing-to-PCP}.
\end{remark}

\subsubsection{Derandomized Parallel Repetition for 2-CSPs on HDX}
Given the 2-CSP in \Cref{thm:CSP_on_CL_large_soundness2} on the graphs $G=(X_n(1),X_n(2))$ from Theorem~\ref{thm:dp_intro}, we can now use the direct-product testing theorem to amplify the soundness of the 2-CSP to any constant close to $0$. Specifically, given a 2-CSP $\Psi$ on $G$, one can naturally define a label cover instance, $\Psi'$ on the bipartite inclusion graph between $X_n(k)$ and $X_n(\sqrt{k})$ as follows. The PCP verifier reads symbols from the tables $F\colon X(k)\to \Sigma^k$ and $G\colon X(\sqrt{k})\to \Sigma^{\sqrt{k}}$, and performs the
following test:
\begin{enumerate}
    \item Sample $A\sim \pi_k$ and $B\subseteq A$ of size $\sqrt{k}$ uniformly.
    \item Read $F[A]$ and $G[B]$. Check that all of the constraints of $\Psi$ inside $A$ are satisfied by the 
    local labeling $F[A]$, and that $F[A]|_{B} = G[B]$.
\end{enumerate}
The instance $\Psi'$ has asymptotically the same size and alphabet. We show that if the direct product test on $X(k)$ has soundness $\delta$, then given that ${\sf val}(\Psi')\geq \delta$, one can get an assignment to $X(1)$ that is consistent with $X(k)$ and in particular satisfies a large fraction of the edges of $\Psi$, implying that $\val(\Psi)\geq 1-O(\delta)$. Using \Cref{thm:CSP_on_CL_large_soundness2} along with \Cref{thm:dp_intro} thus gives us the following conclusion:
\begin{theorem}\label{thm:CSP_on_small_soundness_large_alphabet}
For all $\delta>0$ there exists $C>0$ such that the following 
holds. There is a polynomial time reduction mapping a $3$SAT instance $\phi$ of size $n$ to a label cover instance 
$\Psi$ with the following properties:
\begin{enumerate}
    \item If $\phi$ is satisfiable, then $\Psi$ is satisfiable.
    \item If $\phi$ is unsatisfiable, then ${\sf val}(\Psi)\leq \delta$.
    \item The size of the graph underlying $\Psi$ is at most
    $n (\log n)^{C}$, the alphabet $\Sigma$ of 
    it satisfies that $\log(|\Sigma|)\leq (\log n)^C$,
    and the decision complexity of the
    constraints (i.e., the circuit complexity of a circuit
    checking if a given pair of labels satisfies a constraint)
    is at most $(\log n)^C$.
\end{enumerate}
\end{theorem}
\subsubsection{Applying Alphabet Reduction}
Using Theorem~\ref{thm:CSP_on_small_soundness_large_alphabet} 
and alphabet reduction we conclude the main result of this paper, Theorem~\ref{thm:main}. More specifically, we use the $2$-query PCP composition technique of Moshkovitz and Raz~\cite{MoshkovitzR08} and its abstraction by Dinur and Harsha~\cite{DinurH13}. To apply the latter result, we use two folklore constructions of decodable PCPs, one based on the Reed-Muller code~\cite[Section 6]{DinurH13} and a similar construction based on the Hadamard code. For the sake of completeness we give the details of these constructions in Section~\ref{sec:app-dpcp}.

\section{Preliminaries}\label{sec:prelim}
In this section we give a few basic preliminaries
that will be used throughout the paper.
\vspace{-3ex}
\paragraph{Notations:} 
We use standard big-$O$ notations: we denote $A = O(B)$ 
or $A\lll B$ if $A\leq C\cdot B$ for some absolute constant
$C>0$. Similarly, we denote $A = \Omega(B)$ or $A\ggg B$ 
if $A\geq c B$ for some absolute constant $c>0$. We also 
denote $k\ll d$ to denote the fact that $d$ is taken to 
be sufficiently large compared to any function of $k$. 
If $A$ is a finite set and 
$i \leq |A|$, the notation
$B\subseteq_{i} A$ means that 
we sample a subset of size $i$
of $A$ uniformly.
Given a domain $X$ and a measure $\mu$ over 
it, we denote by $L_2(X;\mu)$ the space
of real-valued functions over $X$ endowed with
the expectation inner product.

\subsection{Properties of Expanders}
We need the following well known version of the expander
mixing lemma for bipartite graphs.
\begin{lemma}\label{lem:bip-eml}
Let $G = (U,V,E)$ be a bipartite graph in 
which the second singular value of the normalized adjacency matrix is at most $\lambda$, and let $\mu$ be the stationary distribution over $G$. Then for all $A \subseteq U$ and $B \subseteq V$ we have that
\[
\left|\Pr_{(u, v) \in E}[u \in A, v \in B] -\mu(A)\mu(B)\right| \leq \lambda\sqrt{\mu(A)(1-\mu(A))\mu(B)(1-\mu(B))}.\]    
\end{lemma}

We also use the following standard sampling 
property of bipartite expanders.
\begin{lemma}\label{lem:sampling}
Let $G = (U,V,E)$ be a weighted bipartite graph with second singular value at most $\lambda$. Let $B \subseteq U$ be a subset with $\mu(B)=\delta$ and set 
\[
T = \left\{v \in V \mid 
\left|\Pr_{u\text{ neighbour of }v}[u \in B]-\delta\right| > \eps \right\}.
\]
Then $\Pr[T]\leq \lambda^2\delta/\eps^2$.
\end{lemma}

\subsection{High-Dimensional Expanders}
    A $d$-dimensional 
    simplicial complex 
    $X = (X(0),\ldots,X(d))$ 
    with vertex set $X(1) = [n]$
    is a downwards closed 
    collection of subsets of $[n]$.
    We follow the convention that 
    $X(0) = \{\emptyset\}$, 
    and for each $i>1$ the set of 
    $i$-faces $X(i)$ is a collection
    of subsets of $X(1)$ of size $i$. The size of $X$ is the total
    number of faces in $X$. The 
    degree of a vertex $v\in X(1)$
    is the number of faces containing
    it, and the degree of $X$ is 
    the maximum degree over all $v\in X(1)$.
    \begin{definition}
For a $d$-dimensional simplicial complex $X = (X(0),X(1),\ldots,X(d))$, $0\leq i\leq d-2$
and $I\in X(i)$, the link of $I$ is the $(d-i)$-dimensional complex $X_I$ whose faces are given as
\[
X_I(j-i) = \{J\setminus I~|~J\in X(j), J\supseteq I\}.
\]
\end{definition}
For a $d$-dimensional complex $X=(X(0),X(1),\ldots,X(d))$ 
and $I\in X$ of size at most $d-2$, the 
graph underlying the link of $I$ is 
the graph whose vertices are $X_I(1)$ 
and whose edges are $X_I(2)$. 
We associate with $X$ a collection 
of distributions over faces. The distribution $\mu_d$ is the uniform 
distribution over $X(d)$, and 
for each $i<d$ the distribution 
$\mu_i$ is a distribution over 
$X(i)$ which results by picking 
$D\sim \mu_d$, and then taking 
$I\subseteq D$ of size $i$ 
uniformly. The distribution $\mu_{I,j-i}$ associated to the link of $I \in X(i)$ is the conditional distribution, $\mu_{j}~|~J \supseteq I$.

\paragraph{Density, average and average on links:}
for a set $S \subseteq X(i)$, let $\mu_i(S)$
denote its density with respect to $\mu_i$.
For $j>i$, $I\in X(i)$ and $S \subseteq X_I(j-i)$ let $\mu_{I,j-i}(S)$ denote the density of $S$ with respect to $\mu_{I,j-i}$. We often omit the subscript and simply write $\mu(S)$ and $\mu_{I}(S)$ in these cases when there is no risk of confusion. We extend the definition for functions, and for $F: X(j) \rightarrow \R$ we define $\mu(F) = \E_{J \sim \mu_j}[F(J)]$ as well as $\mu_{I}(F) = \E_{J \sim \mu_{I,j-i}}[F(J)]$. For a set $K \in X(k)$ for $k \geq j$, 
we define the restricted function $F|_{K}\colon \{J\in X(j)~|~J\subseteq K\}\to\mathbb{R}$ by
$F|_{K}(J) = F(J)$; we think of $\{J\in X(j)~|~J\subseteq K\}$ as being endowed with 
the natural condition measure of $\mu_j$ on this set, and hence define 
\[
\mu(F|_K)
=\E_{\substack{J\sim\mu_j\\ J \subseteq_j K}}[F(J)].
\]

\begin{definition}
We say a $d$-dimensional simplicial complex $X$ is 
a $\gamma$ one-sided local spectral expander if for every $I\in X$ of size at most $d-2$, the second eigenvalue of the normalized adjacency matrix of the graph $(X_I(1), X_I(2))$ is at most $\gamma$.
\end{definition}

\subsection{Properties of Local Spectral Expanders}
Recall that we associated with each $d$-dimensional simplicial 
complex $X$ a sequence of measures $\{\mu_k\}_{1\leq k\leq d}$, where 
$\mu_k$ is a probability measure over $X(k)$. Note that for all $0 \leq t \leq r \leq d$, a sample according to $\mu_t$ can be drawn by first sampling $R \sim \mu_r$, and then sampling $T\subseteq_{t} R$ uniformly. The converse is also true: a sample from $\mu_r$ can be drawn by first sampling $T \sim \mu_t$, and then sampling $R$ from $\mu_r$ conditioned on containing $T$. These observations 
give rise to the standard ``up'' and ``down'' operators, which we present next. %

\begin{definition}%
The operator $U_i^{i+1}$ is a map from $L_2(X(i); \mu_i)$ to $L_2(X(i+1); \mu_{i+1})$ defined as
\[
U_i^{i+1}f(u) 
= 
\E_{v \subseteq_i u}\big[f(v)\big]
\]
for all $u\in X(i+1)$. For $j\geq k+1$, we define $U_k^j$ via composition of up operators: $U_k^j = U_{j-1}^j \circ \ldots \circ U_k^{k+1}$.
\end{definition}

\begin{definition}%
The operator $D_i^{i+1}$ is a map from $L_2(X(i+1); \mu_{i+1})$ to $L_2(X(i); \mu_i)$ defined as
\[
D_i^{i+1}f(u) = \E_{v \supseteq_{i+1} u}\big[f(v)\big]
\]
for all $u \in X(i)$.
For $j\geq k+1$, we define $D_k^j$ via composition of down operators: $D_k^j = D_{k}^{k+1} \circ \ldots \circ D^j_{j-1}$.
\end{definition}

Abusing notations, we use the notations $U^j_k, D^j_k$ 
to denote the operators, as well as the real valued 
matrices associated with them. A key property of
the down and up operators is that they are adjoint:

\begin{claim}%
For all $k \leq j \leq d$, $U_k^{j}$ and $D^{j}_k$ are adjoint operators: for all functions $f\colon X(k)\to\mathbb{R}$ and $g\colon X(j)\to\mathbb{R}$ it holds that $\ip{U_k^{j}f,g} = \ip{f,D^{j}_kg}$. \end{claim}

We need the following lemma regarding the second eigenvalue of the down-up walks $U^j_{k}D^j_{k}$ on $X(j)$ ($j \geq k$), that can be found in~\cite{AlevL20}. Roughly speaking, the lemma asserts that for a
one-sided spectral expander $X$, the singular
values of the operators $U_{\alpha i}^{i}$ 
and $D_{\alpha i}^{i}$ for $\alpha\in (0,1)$
are upper bounded by the eigenvalues
of the corresponding operators in the 
complete complex, up to an additive factor if 
${\sf poly}(i)\gamma$:
\begin{lemma}\label{lem:spectral_gap_of_graphs_from_HDX}
Let $(X, \mu)$ be a $d$-dimensional $\gamma$ one-sided local spectral expander. For all $i \leq d$ and $\alpha \in (1/i, 1)$, the largest singular value of $U^i_{\alpha i}$ and $D^i_{\alpha i}$ is at most $\sqrt{\alpha}+\poly(i)\gamma$. Thus the down-up random walk $U^i_{\alpha i}D^i_{\alpha i}$ on $X(i)$ has second largest singular value at most $\alpha + \poly(i)\gamma$. 
\end{lemma}

\subsection{Partite Complexes}
In this section we define partite complexes 
and state some of their properties.
\begin{definition}
A $d$-dimensional complex $X$ is said to be $d$-partite if $X(1)$ can be partitioned into $d$ disjoint parts called ``colors'', $X(1)=\cup_{i\in [d]} X_i(1)$, such that every $d$-face in $X(d)$ contains one vertex from each color. For a face $I \in X$ we 
let $\text{col}(I)$ denote the set of colors of the vertices it contains. For an $i$-sized set $S \subseteq [d]$, we will use $X_{S}(i) \subseteq X(i)$ to denote $i$-faces $I$ for which $\text{col}(I)=S$.
\end{definition}

The following result is a trickle-down theorem  for partite complexes from~\cite[Lemma 7.5]{DiksteinD19}. 
Let $X$ be a $d$-partite complex, and let $L,R\subseteq[d]$ be two disjoint color classes. We define the bipartite weighted graph $(X_L(|L|), X_R(|R|))$ to be the graph where the weight of an edge $(u,v)$ is equal to
\[
w(u,v)=\Pr_{x\sim \mu_d}\left[x_L = u, x_R = v\right].
\]
\begin{lemma}\label{lem:trickling-partite}
Let $X$ be a 
$d$-partite simplicial complex, and suppose that for all 
$v \in X(1)$ the graph underlying $X_v$ is a 
$\lambda$-one sided 
$(d-1)$-partite expander, for 
$\lambda < \frac{1}{2}$. Suppose that the underlying graph of \(X\) is connected. Then for every $i\neq j$, the bipartite graph between 
$X_{\{i\}}(1)$ 
and 
$X_{\{j\}}(1)$ is a 
$\frac{\lambda}{1-\lambda}$-bipartite expander. 
\end{lemma}

For color sets $L$ and $R$ that are larger than $1$, the following lemma bounds the eigenvalues the 
graph $(X_L(|L|), X_R(|R|))$ provided bounds on the 
second eigenvalue of associated graphs 
on links of $X$.
\begin{lemma}\label{claim:gll-sing-val}
Let $X$ be a $d$-partite complex. 
Suppose that for each link $L$ of $X$ 
and for any two colors $i,j\not\in \text{col}(L)$, we have that the bipartite graph $(L_i(1),L_j(1))$ has second largest eigenvalue at most $\lambda$. Then, for any pair of disjoint sets $L,R \subseteq [d]$, the bipartite graph $(X_L(|L|), X_R(|R|))$ has second largest eigenvalue at most $\poly(d)\lambda$.
\end{lemma}

\subsection{Variants of the Chapman-Lubotzky Complex}
In this section we discuss variants of the 
Chapman and Lubotzky complex~\cite{ChapmanL} 
and some of their properties. In their paper, Chapman and Lubotzky~\cite{ChapmanL} construct an infinite family of complexes that are $2$-dimensional coboundary expanders over $\F_2$. In the works~\cite{BLM24,DDL24} the authors 
extend this construction, and show that
for all $m\in\mathbb{N}$, one can construct variants of the Chapman and Lubotzky complexes that are  $2$-dimensional coboundary expanders over $S_m$. They used this fact to prove that the natural $2$-query
direct product tester has small soundness $\delta$ (that can be taken to be arbitrarily close to $0$ so long as one takes $m$ large enough). Here, the natural $2$-query direct product tester is defined in the following way:
\begin{definition}\label{def:dp-test}
Let $X$ be a $d$-dimensional high-dimensional expander, and let $k<d$. Given a supposed encoding $F\colon X(k)\to \Sigma^k$, the $(k,\sqrt{k})$-direct product tester associated with $X$ proceeds as follows:
\begin{enumerate}
    \item Sample $D\sim \pi_d$.
    \item Sample $B\subseteq D$ of size $\sqrt{k}$ 
    uniformly.
    \item Independently and uniformly sample $k$-faces $A, A'$ satisfying $B\subseteq A,A'\subseteq D$.
    \item Check that $F[A]|_B = F[A']|_B$.
\end{enumerate}
We say that the $(k,\sqrt{k})$-direct product test on $X$ has soundness $\delta$ if the following holds. Let $F:X(k)\rightarrow \Sigma^k$ be any function that passes the $(k,\sqrt{k})$-direct-product test above with probability at least $\delta$. Then there exists a function $f:X(1)\rightarrow\Sigma$ such that,
\[\Pr_{A\sim X(k)}[\Delta(F[A],f|_A)\leq \delta] \geq \poly(\delta).\]
\end{definition}

In the lemma below, we state the properties that we need from the complexes of~\cite{BLM24,DDL24}. These properties include the aforementioned fact about direct product testing, as well as the fact that these complexes are polynomial-time constructible. The polynomial time constructability of these complexes was established by Dikstein, Dinur and Lubotzky~\cite{DDL24}, however as our argument requires a different setting of parameters, some modifications to the argument are necessary. We also need some other features from these complexes, but these are much easier to establish.

\begin{theorem}\label{thm:cl}
For all $\delta\in (0,1)$ there 
exists $\ell \in\mathbb{N}$ such that the following holds for all $C>0$. For large enough $k,d\in \N$, for large enough $n \in \N$ and any prime $q =\Theta(\log^C n)$, one can construct in time ${\sf poly}(n)$ a $d$-dimensional complex $X$ for which $n \leq |X(1)|\leq O_{\ell,d}(n)$ such that for  the following holds:
\begin{enumerate}
\item The complex $X$ is $d$-partite.
\item Every vertex participates in at most $q^{O(d^2)}$ $d$-cliques.
\item For each vertex link $L$ of $X$ there is a group $\sym(L)$ that acts transitively on the $d$-faces $L(d)$.
\item For each vertex link $L$ of $X$ 
and any colors $i\neq j$, the bipartite graph $(L_i(1),L_j(1))$ is uniformly weighted (over edges) and has diameter $O\left(\frac{d}{|i-j|}\right)$.
\item For all links $L$ of $X$, every bipartite graph $(L_i(1),L_j(1))$ for $i\neq j \in [d] \setminus \text{col}(L)$ has second largest eigenvalue at most 
$\frac{2}{\sqrt{q}}$. 
\item The second largest singular value of $G=(X(1),X(2))$ is at most $\frac{1}{d-1}+\frac{2}{\sqrt{q}}\leq \frac{2}{d}$.
\item The $(k,\sqrt{k})$-direct product test on $X$ has soundness $\delta$.
\end{enumerate}
\end{theorem}
\begin{proof}
First fix the parameters $\delta,C, d, n$ and a prime $q=\Theta(\log^C n)$. To construct a complex with $\Theta(n)$ vertices and parameter $q$, we follow the presentation in~\cite[Section 5]{DDL24}. 

Therein the authors pick a pro-$\ell$ subgroup $H_0 \subseteq SU_g(H^{\ell}(\mathbb{Q}_\ell))$ and define the lattice $\Gamma_0$ in Section 5.1.3. Then using Fact 5.12 and Claim 2.31 therein, they show that for some finite $i$, one can pick a normal subgroup $H_i$ of $H_0$ with $[H_0:H_i]=\ell^i$, for which the action of the corresponding lattice $\Gamma_i$ does not merge vertices at distance at most $4$. That is, for all $1\neq \gamma \in \Gamma_i, \gamma$ and $v$ in $\widetilde{C}_d$, $\text{dist}(\gamma,\gamma v) \geq 4$. Then their infinite family of complexes is defined by $\{\Gamma_j \setminus \widetilde{C}_d\}_{j\geq i}$.

In our case, $q$ depends on $n$, so apriori the size of $\Gamma_i \setminus \widetilde{C}_d$ could already be much larger than $n$. To circumvent this issue we use an explicit choice of subgroups as constructed in \Cref{app:cl}, with $H(1)$ replacing $H_0$, and $H(i)$ in \Cref{cor:Hi} replacing $H_i$ from above. Here we used the fact that $H(i)$'s form a filtration of $H(1)$, such that for all $k\geq 1$, $H(k)$ is a pro-$\ell$ subgroup and for all $k\geq 2$, $H(k)$ is an open compact normal subgroup of $H(k-1)$. Let $X_i$ denote the complex $\Gamma_i \setminus \widetilde{C}_d$ where $\Gamma_i$ is the lattice corresponding to $H(i)$.
\Cref{cor:Hi} gives us that 
there is an $i^{\star}$ for which the compact open subgroup $H_{i^{\star}}\subset G(\mathbb{Q}_\ell)$ satisfies that for any element $1\neq \gamma \in \Gamma_{i^{\star}}$, and any vertex $v\in \widetilde{C}_d(0)$, $\text{dist}(v,\gamma v)\ge 4$. And moreover, $|X_{i^{\star}}(d)|\le C_{\ell, d}q^{O(d^{2})}$ for some constant $C_{\ell, d}$ depending only on $\ell$ and $d$, implying that $|X_i(0)|\le \pl n$. 

To construct a complex $X$ with $\Theta(n)$ vertices we start with $i=i^{\star}$, compute $X=X_i$, and as long as it has less than $n$ vertices, increase $i$ by $1$. Note that as $[H(i) : H(i+1)] = \ell^{\Theta(d^2)}$ for all $i\geq 1$, the number of vertices in $X_{i+1}$ is larger by at most factor
$\ell^{O(d^2)}$ compared to the number of vertices in $X_i$. Repeating the argument in Corollary~\ref{cor:Hi}, item 1 one we get that if $\ell^{\lceil i/2\rceil}\geq 4d q^n$, then $\text{dist}(v,\gamma v)\geq n$ for any $v\in \tilde{C}_d(0)$ and $1\neq \gamma\in \Gamma_i$, 
in which case the number of vertices in $X_i$ is certainly at least $n$. Thus, the process will terminate at $i=i_0$, for $i_0\leq O_{d,\ell}(n \log q)$, and at that point $X_{i_0}$ will have between $n$ and $n\cdot \ell^{O(d^2)}$ vertices (in fact, as we shortly see, the process ends at some $i_0 = O_{d,\ell}(\log n)$ due to diameter considerations). As per the runtime of the construction, note that by Lemma~\ref{lem:time_bd_construct}, it takes at most $q^{O_{\ell,d}({\sf diam}(X_{i}))}$ time to construct  $X_{i}$, for every $i \geq i^{\star}$, where ${\sf diam}(X_{i})$ is the diameter of the graph $(X_{i}(1),X_{i}(2))$. After establishing several helpful eigenvalue bounds we will show that ${\sf diam}(X_{i})\leq O(\log_q(|X_{i}(1)|))$ for all $i$, from which it follows that the construction of any $X_{i}$ with at most $n$ vertices can be done in $n^{O_{d,\ell}(1)}$ time. Thus, the overall algorithm above takes $O_d(n\log q)\cdot n^{O_{d,\ell}(1)}=n^{O_{d,\ell}(1)}$ time.

\vspace{-2ex}
\paragraph{$X$ is $d$-partite:} For the first item, we note that the complex
$X$ is a quotient of the affine spherical building $\tilde{C}_{d}$ with some group $\Gamma$. 
The affine building $\tilde{C}_{d}$ 
is a $d$-partite complex and the symplectic group over $\Q_q$, ${\sf SP}(2d, \Q_q)$, acts
transitively on the top faces of $\tilde{C}_{d}$. The complexes constructed in~\cite{BLM24,DDL24} and \Cref{app:cl}
are quotients of $\tilde{C}_{d}$ by subgroups of ${\sf SP}(2d, \Q_q)$, and in particular it follows that all of these quotients are also $d$-partite. 

\vspace{-2ex}
\paragraph{Properties of vertex-links:} 
The link of each vertex in $v\in X(1)$ is a product of two spherical buildings of dimension at most $d$ over $\mathbb{F}_q$ (see~\cite[Fact 3.14]{DDL24} and~\cite[Lemma 6.8]{BLM24}). It is well-known that they satisfy the third and fourth item of the lemma, and this fact can be found in both the papers.

\vspace{-2ex}
\paragraph{Local spectral expansion of $X$:}
For the fifth item, we first note that the statement holds for every link $L\neq L_{\emptyset}$ of $X$ since these are tensor products of spherical buildings and the  eigenvalue computations for them are well-known (see~\cite[Claims 3.2, Lemma 3.10]{DDL24} or ~\cite[Lemmas 2.15, 2.20]{BLM24}). Using this we get that the underlying graph of each vertex link $L$ of $X$ is a $1/\sqrt{q}$-one-sided $d$-partite expander. The conclusion for $L=L_{\emptyset}$ now follows by an application of the trickling down result for partite complexes, \Cref{lem:trickling-partite}.

\vspace{-2ex}
\paragraph{The diameter of $X_{i}$:} fix parts $1\leq j<j'\leq d$, and consider the bipartite graph between them, denoted by $(L_j,L_{j'}, E_{j,j'})$. By the fifth item, this graph has second eigenvalue at most $O(1/\sqrt{q})$. Thus, letting $A_{j,j'}\colon L_2(L_j)\to L_2(L_{j'})$ be the averaging operator defined as 
$A_{j,j'}f(v_{j'}) = \E_{v_j\sim N(v_{j'})}[f(v_j)]$, we have that 
$\langle A_{j,j'} f, g\rangle\lll \frac{1}{\sqrt{q}}\|f\|_2\|g\|_2$
for $f\colon L_j\to \mathbb{R}$, $g\colon L_{j'}\to\mathbb{R}$ such that $\E[f]=0$. Let $A^{\dagger}_{j,j'}$ be the adjoint operator of $A_{j,j'}$. Plugging in $g=A_{j,j'}A_{j,j'}^{\dagger}A_{j,j'}f$ we conclude that
$\|A_{j,j'}^{\dagger}A_{j,j'} f\|_2\lll\frac{1}{\sqrt{q}}\|f\|_2$ for all $f\colon L_j\to \mathbb{R}$ 
such that $\E[f] = 0$. It follows that the weighted graph $G_j$ on $L_j$ corresponding to the operator $A_{j,j'}^{\dagger}A_{j,j'}$ has second eigenvalue at most $O(1/\sqrt{q})$ in absolute value. Thus, it follows (e.g. from~\cite[Theorem 2]{chung1989diameters}) that the diameter of $G_j$ is $O(\log_q |L_j|)\leq O(\log_q |X_i(1)|)$, and therefore the diameter of $(L_j,L_{j'}, E_{j,j'})$ is also $O(\log_q |X_i(1)|)$. Since this is true for all $j,j'$, the diameter of $X_{i}$ is $O(\log_q |X_i(1)|)$.

\vspace{-2ex}
\paragraph{Second largest singular value of $G$:}
For the sixth item, write $X(0) = V_1\cup\ldots\cup V_d$, let $M$ be the normalized adjacency operator of $X$, and 
let $f\colon X(0)\to\mathbb{R}$ be a function
with $\E[f] = 0$ and $\E[f^2] = 1$. As $X$ is
$d$-partite, we have that $\mu_1(V_i) = \frac{1}{d}$ for all $i$, and we denote $a_i = \E_{v\sim \mu_1}[f(v)~|~v\in V_i]$.
\begin{equation}\label{eq:constructability1}
\langle f, Mf \rangle
=\E_{\substack{i,j\in [d]\\ i\neq j}}
\left[\E_{(u,v)\sim M_{i,j}}\left[f(u)f(v)\right]\right]
=
\E_{\substack{i,j\in [d]\\ i\neq j}}
\left[\E_{(u,v)\sim M_{i,j}}\left[(f(u)-a_i)(f(v)-a_j)\right]\right]
+\E_{\substack{i,j\in [d]\\ i\neq j}}
\left[a_ia_j\right],
\end{equation}
where $M_{i,j}$ is the normalized adjacency 
operator of the bipartite graph $G_{i,j}$. 
Using the fifth item we have
\begin{align}
\left|
\E_{\substack{i,j\in [d]\\ i\neq j}}
\left[\E_{(u,v)\sim M_{i,j}}\left[(f(u)-a_i)(f(v)-a_j)\right]\right]
\right|
&\leq 
\frac{2}{\sqrt{q}}\E_{\substack{i,j\in [d]\\ i\neq j}}
 \sqrt{\E_{u\in V_i}(f(u)-a_i)^2}\sqrt{\E_{v\in V_j}(f(v)-a_j)^2} \notag\\
&\leq \frac{2}{\sqrt{q}}\E_{i}\E_{u\in V_i}(f(u)-a_i)^2\notag\\
&\leq \frac{2}{\sqrt{q}}\E_{i}\E_{u\in V_i}\left[f(u)^2\right],
\label{eq:constructability2}
\end{align}
which is at most $2/\sqrt{q}$. In the second transition we used Cauchy Schwarz. Next, we have
\[
0 
= \E_{i,j}[a_ia_j]
= \frac{1}{d}\E_i[a_i^2]
+\left(1-\frac{1}{d}\right)\E_{i,j\in [d], i\neq j}
\left[a_ia_j\right],
\]
so $\E_{i,j\in [d], i\neq j}
\left[a_ia_j\right] = -\frac{1}{d-1}\E_{i}[a_i^2]$. Finally, 
\[
\E_{i}[a_i^2]
=\sum\limits_{i}\mu(V_i)\E_{v}[f(v)~|~v\in V_i]^2
\leq 
\sum\limits_{i}\mu(V_i)\E_{v}[f(v)^2~|~v\in V_i]
=\E[f^2]=1,
\]
so overall $\left|\E_{i,j\in [d], i\neq j}
\left[a_ia_j\right]\right|\leq \frac{1}{d-1}$. 
Plugging this and~\eqref{eq:constructability2} 
into~\eqref{eq:constructability1} finishes
the proof of the sixth item.

\vspace{-2ex}
\paragraph{Soundness of the Direct-Product Test:}
The seventh item follows from~\cite[Theorem 1.3]{BLM24} (or more precisely, from the combination of~\cite[Theorem 6.12]{BLM24} and~\cite[Theorem 1.8]{BLM24}) as well as from \cite[Theorem 1.1]{DDL24}. The works~\cite{DiksteinD-agreement,dikstein2023swap,DDL24} 
established the statement above as is, 
whereas the proof of~\cite{BafnaMinzer,BLM24}
established the statement for 
the alphabet $\Sigma = \{0,1\}$. For 
the sake of completeness, in
Section~\ref{sec:dp-large-alph} we explain how to adapt the argument from the latter papers to the 
case of general alphabets $\Sigma$.
\end{proof}

\begin{remark}
    We remark that in~\cite{DDL24} the authors suggest a different way of efficiently constructing a complex with around $n$ vertices when starting with a complex with less than $n$ vertices. Their method proceeds by computing an $\ell$-lift of the complex, which can be shown to exist as the $H(i)$'s are pro-$\ell$ groups. Also, an $\ell$-lift can be efficiently computed, as this problem could be phrased as a linear algebra problem over $\mathbb{F}_{\ell}$. In the above argument, we too could have used the $\ell$-lift method when starting with $X_{i}$ for $i=i^{\star}$ and increased $i$ until $X_i$ has the correct number of vertices. We chose the more direct approach though, since we have an explicit description of all the subgroups $H(i)$.
\end{remark}

\section{Background on Routing Protocols}
In this section we discuss routing protocols, 
the pebble routing problem and a relaxation of
it that is sufficient for us.
\subsection{Pebble Routing Protocols}
We start the discussion by formally defining routing/communication protocols on a graph $G$, 
as well as defining other related notions.
\begin{definition}\label{def:comm-protocol}
Given a graph $G$, an $r$-round routing protocol $\cR$ on $G$ is a set of rules where at each round vertices can send and receive messages to and from their neighbours in $G$. To decide what messages to send forward, each vertex is allowed to perform an arbitrary computation on all of its received messages. 

The work complexity\footnote{We note that the notion of work complexity used in prior work of~\cite{JRV} is slightly different from ours.} of $\cR$ is the total computation that any vertex performs throughout the protocol, as measured by the circuit size for the equivalent Boolean functions that the vertex computes.

At round $0$, the protocol starts out with an arbitrary function $f:V(G)\rightarrow \Sigma$ on the vertices of $G$, and after $\cR$ is implemented, at the final round $r$, we end up with a function $g:V(G)\rightarrow \Sigma$ on the vertices, also referred to as the function computed by the protocol $\cR$.
\end{definition}

We now formally define the pebble routing problem, first studied in~\cite{AlonCG94}. 
\begin{definition}\label{def:routing}
We say that a graph $G$ has an $r$-round pebble-routing protocol, denoted by $\rt(G) = r$, if the following holds. For all permutations $\pi: V(G) \rightarrow V(G)$ there is an $r$-round communication protocol on $G$ such that in each round every vertex receives exactly one message symbol in $\Sigma$, and then sends this message forward to exactly one of its neighbors. If the protocol starts with $f:V(G)\rightarrow\Sigma$, then at the $r^{th}$-round we have the function $g:V(G)\rightarrow\Sigma$ on the vertices satisfying $g(\pi(u))=f(u)$.
\end{definition}
Note that this protocol can also be thought of as a set of $|V(G)|$ paths, each one transmitting a message from $u \rightarrow \pi(u)$, where at each round any vertex is involved in exactly one path. 
We encourage
the reader to think of pebble routing in this way. We remark that any pebble-routing protocol has work complexity at most $r\log|\Sigma|$.\footnote{We are accounting for the input length at a vertex in its computation cost.} In a sense, pebble routing is the simplest possible protocol, 
where no computation is performed by any vertex as it only forwards its received message. The works~\cite{AlonCG94,Nenadov23} study the pebble routing problem and prove that a sufficiently good expander graph admits an $O(\log n)$-length protocol.

For our purposes, it suffices to consider a relaxation of the pebble routing problem, 
in which the protocol is required to satisfy that $g(\pi(u))=f(u)$ for all but $o(1)$ 
fraction of the vertices. Also, we relax 
the condition that at each round, each
vertex is used once. Instead, we only require
that throughout the protocol, each vertex 
is a part of at most ${\sf poly}(\log n)$ paths. Formally, we need the following result:
\begin{theorem}\label{thm:expander-routing}
There exists a universal constant 
$\alpha>0$ such that the following holds.
Let $G = (V, E)$ be a regular, expander graph with second largest singular value $\sigma_2(G) \leq \alpha$. Let $c \geq 0$ be a fixed constant and $\pi: V \to V$ be a permutation. Then there is a $\poly(|E|)$-time algorithm to construct a protocol with $O(\log n)$-length paths $\cP$ from $u$ to $\pi(u)$ for all but $O\left(\frac{1}{\log^c(n)}\right)$-fraction of $u$. Furthermore, every vertex in $V$ is used in at most $t = O(\log^{c+1} n)$ paths in $\cP$.
\end{theorem}
\begin{proof}
We note that~\cite[Theorem 2]{Upfal} is
the case of $c=0$ in the above theorem. The proof for the case of $c>0$ is an easy generalization of the argument therein, and we give the formal argument in Section~\ref{app:path-expander} for the sake of completeness.
\end{proof}

\subsection{Routing Protocols under Adversarial Corruption}
We now formally introduce the notion of almost-everywhere (a.e.) reliable transmission in the presence of adversarial corruptions. %
The goal in this problem is to design sparse networks and routing protocols for them, which allow a large fraction of honest nodes to communicate reliably even under adversarial corruptions of the communication network. 
Our setting is slightly different from  most prior works in two ways. First, we need to consider a more general model of adversarial corruptions, where an arbitrary constant fraction of the edges may behave maliciously (as opposed to vertices being corrupted). 
Second,  we are given a permutation $\pi$ on $V$, and we need to design a routing protocol that succeeds in transmitting all but a small fraction of messages from $u \rightarrow \pi(u)$ correctly. %
Formally:
\begin{definition}
\label{def:adv-comm-protocol}
We say an edge $(u,v) \in G$ is uncorrupted if whenever $u$ transfers a message $\sigma$ across $(u,v)$, then $v$ receives $\sigma$; otherwise, we say the edge $(u,v)$ is corrupted. We say that a graph $G$ has a $(\eps(n),\nu(n))$-edge-tolerant routing protocol with work complexity $w(n)$ and round complexity $r(n)$ if the following holds. Let $\pi: [n] \rightarrow [n]$ be any permutation of $V(G)$. Then there is an $r$-round communication protocol on $G$ such that for all functions  $f:V(G)\rightarrow \Sigma$, after running the protocol for $r$-rounds each vertex
$v$ computes a value $g(v)\in \Sigma$ with the following guarantee: any adversary that corrupts at most $\eps(n)$-fraction of the edges,
\[\Pr_{u \in V(G)}[f(u) \neq g(\pi(u))] \leq \nu(n).\]
As in \Cref{def:comm-protocol}, the work complexity $w(n)$ of the protocol is the maximum computation that any node performs throughout the protocol, as measured by circuit size.
\end{definition}

Recall that if $G$ is a constant degree graph which admits a pebble routing protocol of length at most $r$, then $G$ admits a simple routing protocol that is $(\eps,r\eps)$-tolerant under edge corruptions. 
This simple connection though is insufficient for us: on a constant degree graph, the best we can hope for is $\rt(G) = \Theta(\log n)$, which necessitates taking $\eps \leq \Theta\left(\frac{1}{\log n}\right)$. Using such a protocol in conjunction with Lemma~\ref{lem:routing-to-pcp-general} only leads to a 2-CSP on $G$ with soundness $1-\Theta\left(\frac{1}{\log n}\right)$, which is too weak for our purposes. 

In the next section we will construct more involved protocols on the graphs underlying sufficiently good HDX, that are tolerant to a constant fraction of edge corruptions and have poly-logarithmic work complexity. Our main idea is to perform a message transfer on highly connected dense subgraphs in $G$, which naturally makes the protocol error-resilient, similar to the protocols of~\cite{Chandran,JRV}.

\section{Link to Link Routing on CL-complexes}\label{sec:link-to-link}
The goal of this section is to present a routing 
protocol over graphs underlying high-dimensional
expanders. We begin by giving a high-level
overview of the idea in Section~\ref{sec:regularity-issues}, followed 
by a detailed description of the protocol and
its analysis.

\subsection{High-level Overview}\label{sec:regularity-issues}
Throughout this section, 
we denote by $\maj_{\nu}(\sigma_1,\ldots,\sigma_k)$ 
the value $\sigma$ such that $\sigma_i = \sigma$ or at least $\nu$
fraction of $i$'s if such value exists, and $\perp$ otherwise.
Typically $\nu$ will be close to $1$ (say, $\nu = 0.99$).

\par
Recall that ultimately, we want to
use our routing protocol to embed a regular $2$-CSP instance $\Psi'$ on the graph $G$. Since the instance $\Psi'$ is over a regular graph and the graph $G$ is not, there is some incompatibility between the two. In particular, it doesn't make much sense to identify vertices of $\Psi'$ with vertices of $G$ in an arbitrary way.

To circumvent this issue we use the zig-zag product: we choose an appropriate family of expander graphs $\mathcal{H}$ and take the zig-zag product graph $Z = G\zz \mathcal{H}$. Thus, we get a graph $Z$ which is a regular expander graph, and additionally there is a natural correspondence between vertices in $Z$ and in $G$. Indeed, a vertex $v$ 
in $Z$ is naturally composed of a 
cloud-name, which we denote by 
$v_1$ and is actually a name of a vertex from $G$, and additionally an inner index that we denote by 
$v_2$ (corresponding to an edge of $v_1\in G$). We associate a vertex of 
the $2$-CSP $\Psi'$ with a vertex of $v\in V(Z)$ (this makes sense now as
these two graphs are regular), and this in turn is associated to the vertex $v_1$ in $G$. As discussed in \Cref{sec:overview:conn}, we break $\Psi$ into $k$ permutations $\pi_1,\ldots, \pi_k$ each one thought of as a permutation on $V(Z)$, and our goal is to transfer the value of $v$ to $\pi(v)$ for some fixed permutation $\pi$ from $\{\pi_1,\ldots,\pi_k\}$. We will instead think of the value of $v$, as held by all the vertices in the link of $v_1$, denoted by $L_{v_1}$, and this value needs to be transferred to the link $L_{\pi(v)_1}$. Once we have this kind of transfer, in \Cref{sec:routing-to-pcp} we show how to reduce $\Psi$ to a 2-CSP on $G$.

Formally, the initial function $A_0$ that has to be routed will be on the links associated to each $v\in V(Z)$. First define the set $\cS=\{(v,u):v\in V(Z),u\in X_{v_1}(1)\}$. Then $A_0:\cS \rightarrow \Sigma$ should be thought of as a function that for every $v\in V(Z)$, assigns the link $X_{v_1}(1)$ values that are mostly the same across vertices of the link, in the sense that there exists $\sigma_v \in \Sigma$ such that $\maj_{0.99}(A_0(v,u)|u\in X_{v_1}(1))=\sigma_v$. For the overview, it is helpful to think of all the vertices in $L_{v_1}$ holding the same value, $\sigma_{v}$\footnote{We need to start with the weaker condition of almost all vertices holding $\sigma_v$ to make the embedding result in \Cref{sec:link-routing-pcp} go through. Our proof easily ports over to this more general setting, so it is useful to think of all vertices in $L_{v_1}$ holding $\sigma_v$ for the overview.}. Note that every value $A_0(v,u)$ is held at the vertex $u\in G$ ($v$ is a label for this value) and every vertex in $G$ holds multiple such values. Even though the actual routing takes place on $G$, it is convenient to keep the graph $Z$ in mind for the analysis, think of $A_0$ as a function on the links $X_{v_1}$ (for all $v\in Z$) and the protocol as a link to link transfer.

After the routing protocol ends, we will have a function $A_T:\cS\rightarrow \Sigma$. Our main result in this section, \Cref{lem:link-routing}, shows that the majority value on the link $X_{v_1}$ gets transferred to the link $X_{\pi(v)_1}$ for most $v\in Z$.

With this setup in mind, we first find a collection of paths $\cP$ on $Z$, each path being from $v \rightarrow \pi(v)$, using the relaxed pebble routing protocol in \Cref{thm:expander-routing}. This is possible since $Z$ is a regular expander graph. We now use these paths to implement the link to link transfer.

{\bf Transmitting on Links:} 
Each path $P= u_1 \rightarrow \ldots \rightarrow u_T$ in $\cP$, for $T = O(\log n)$ and $u_T=\pi(u_1)$, can equivalently be thought of as a path over vertex-links: $L_{u_1}\rightarrow \ldots\rightarrow L_{(u_T)_1}$. 
Pretending for a moment that each link is in fact a clique, the protocol proceeds as follows: 
vertices of $L_{u_1}$ send their
message to $L_{u_1,u_2}$; the vertices in $L_{u_1,u_2}$ each
compute a majority value and pass 
it on to the vertices in $L_{u_2}$, 
and so on. We show that for any adversarial strategy, this protocol 
succeeds in transmitting the correct message on almost all of the paths in $\cP$. The key here is that vertices are allowed to take majority values, so as long as they are not over-saturated with corrupted edges, they will compute the correct value.

Returning to our actual scenario, 
the vertex links in $X$ do not
actually form cliques and so we cannot use the protocol as described above. To remedy this situation, 
we show that for each vertex link $L$ in $X$ we can set up a collection of short paths $\cP_L$ such that for almost all vertex pairs $u,v\in L$, the collection $\cP_L$ contains a short path between $u$ and $v$. Furthermore, no edge is used in the paths $\cP_L$ too often. The collection of paths $\cP_L$ allows us to pretend that the link $L$ almost forms a clique, in the sense that we transmit a message between pairs $u,v\in L$ using the path from $\cP_L$ between them (as opposed to directly as in the case of cliques).

{\bf Gap amplification:} Finally we remark that one additional  benefit of having $A_0$ on links is that only $o(1)$-fraction of the message transfers from $L_{v_1} \rightarrow L_{\pi(v)_1}$ are unsuccessful. We show that our protocol is $(\eps,1/\pl n)$-tolerant, a guarantee which is impossible if the initial function was on $V(G)$. Therefore associating the vertices of $\Psi'$ to links, as we show in Section~\ref{sec:routing-to-pcp}, translates to a gap amplification result for PCPs. More precisely, if we start with a $2$-CSP $\Psi'$ such that $\val(\Psi')\leq 1-1/\pl n$, then we get a 2-CSP $\Psi$ on $G$ with $\val(\Psi)\leq 1-\Omega(1)$, where $\Psi$ is obtained using our link to link routing protocol. This gives the same amplification as achieved in the gap amplification procedure of 
Dinur~\cite{Dinur07}, but the alphabet size is $\exp(\pl n))$. This can be brought down to constant-sized alphabet by incurring a size blow-up by a factor $\pl n$ using the alphabet reduction technique (which we anyway have to do, see Section~\ref{sec:final}).

\skipi
Throughout this section we fix $X$ a complex
as in Theorem~\ref{thm:cl} with the parameters $\delta\in (0,1)$ chosen arbitrarily,\footnote{The parameter $\delta$ dictates the soundness of our final PCP. The results in this section and the next hold for all $\delta$.}  $q=(\log n)^C$
for sufficiently large constant $C>0$, $d$ a large constant, and $|X(1)|=n$. We also
fix the graph $G = (X(1),X(2))$.

\subsection{Routing Inside a Vertex Link}\label{sec:inside-link}
In this section we describe a routing procedure
inside individual links. This routing procedure
will help us to facilitate the intuition that
links in $X$ are almost as well connected as 
cliques, supporting the approach presented in 
Section~\ref{sec:regularity-issues}.

We now show that inside each vertex link $L$, 
we can construct (in polynomial time) a set of short paths $\cP$ between almost all pairs of vertices in $L$. More precisely, we construct a collection of paths $\cP$ between all pairs of vertices $U \in L_i(1), V \in L_j(1)$ for two indices $i,j$ that are far apart. Our algorithm will give up on a small fraction of such pairs; we refer to the path between them as ``invalid'' and denote $P(u,v)=\perp$.
\begin{lemma}\label{lem:path-ij}
Let $\eps>0$ and fix a pair of indices $i,j \in [d-1]$ with $|i-j| \geq (d-1)\eps$. Then there is a $\poly(|L_{ij}(2)|)$-time algorithm that constructs a set of paths $\cP_{ij}=\{P(u,v)\}_{u\in L_i(1),v\in L_j(1)}$, where each valid path is of length $O(1/\eps)$ and at most $O(\eps)$ fraction of paths are invalid. Furthermore, each edge in $L_{ij}(2)$ is used in at most $O\left(\frac{|L_i(1)||L_j(1)|}{\eps^3|L_{ij}(2)|}\right)$ paths in $\cP_{ij}$.
\end{lemma}
\begin{proof}
Consider any pair of indices $i,j \in [d]$ with $|j-i| \geq (d-1)\eps$. Fix $\ell=\Theta(1/\eps)$ and $t=\Theta\left(\frac{|L_i(1)||L_j(1)|}{\eps^3|L_{ij}(2)|}\right)$. 
By the fourth item in Theorem~\ref{thm:cl} the diameter of $(L_i(1),L_j(1))$ is at most $O(1/\eps)$. The algorithm to construct the paths simply picks the shortest paths between a pair of vertices iteratively, and deletes any edges that has been used at least $t$-times. Note that since paths are of length $\leq \ell$, in an ideal scenario 
where every edge occurrs equally often, each edge would belong to  $O\left(\frac{|L_i(1)||L_j(1)|}{\eps|L_{ij}(2)|}\right)$ paths. 
We take $t$ to be larger so as to allow some slack, which still gives us the uniformity over edges as in the statement of the lemma. 
We now proceed to the formal argument.

For a set of edges $\cE \subseteq L_{ij}(2)$, let $L_{ij}(\cE)$ denote the bipartite graph with vertices $(L_i(1),L_j(1))$ and the set of edges $\cE$. Formally our algorithm is as follows: 
\begin{mdframed}
\begin{itemize}
\item Instantiate $\cP_{ij}=\emptyset$, $\cE=L_{ij}(2)$.
\item For every $u\in L_i(1), v\in L_j(1)$ do the following:
\begin{enumerate}
\item Find the shortest path $P(u,v)$ between $u$ and $v$ in the graph $L_{ij}(\cE)$. If the length of $P(u,v)$ is at most $\ell$ then add it to the set $\cP_{ij}$, else set $P(u,v)=\perp$.
\item If any edge $e$ in $L_{ij}(2)$ has been used in at least $t$ paths in $\cP_{ij}$ then remove it from $\cE$.
\end{enumerate}
\end{itemize}
\end{mdframed}

It is easy to see that the above algorithm runs in polynomial time in $|L_{ij}(2)|$ and that every edge in $L_{ij}(2)$ is used in at most $t$ paths in $\cP_{ij}$. It remains to show that the algorithm finds a path of length at most $\ell$ between almost all pairs of vertices. 

Let $\cE_f$ denote the set $\cE$ when the algorithm terminates and $\cE_0=L_{ij}(2)$ denote the set at the start. 
It is easy to check that the number of edges that the algorithm removes from $\cE_0$ is at most $|L_i||L_j|\ell\cdot 1/t = \Theta(\eps^2|L_{ij}(2)|)$, implying that $\Pr_{e\sim L_{ij}(2)}[e\notin \cE_f]\leq \eps^2$.
Since the diameter of $(L_i(1),L_j(1))$ is at most $\ell$, for every $u \in L_i, v\in L_j$ we may fix an arbitrary shortest path $P_{u,v}$ (length $\leq \ell$) between them. By the third item in Theorem~\ref{thm:cl} there is a group of symmetries $\sym(L)$ that acts transitively on the top faces of $L$. Thus, we may consider the path $g \circ P_{u,v}$ for any $g \in \sym(L)$, with goes between $g\circ u$ and $g\circ v$. 
We consider the distribution $\cD$ over paths, 
obtained as sample $u \sim L_i, v \sim L_j$, $g \sim \sym(L)$, and output the path $P= g \circ P_{u,v}$. Note that a random edge drawn from a random path $P \sim \cD$, is a uniformly random edge in $L_{ij}(2)$ due to the transitive symmetry of the group $\sym(L)$. Therefore,
\[
\E_{\substack{P \sim \cD\\ e\sim P}}
[\mathbb{1}_{e\not\in \cE_f}] \lll \eps^2.\]
Since the marginal of the starting and ending point of $P \sim \cD$ is the uniform distribution over $u \sim L_i,v \sim L_j$ and 
the length of the path is $O(1/\eps)$, we can rearrange the left hand side above and  apply a union bound to get
\[\E_{u\sim L_i,v\sim L_j}[\E_{\substack{P \sim \cD|u,v}}[\mathbb{1}_{\exists e \in P, e \notin \cE_f}]]\lll\ell \eps^2 \lll \eps.\]
Thus, by an averaging argument for at least a $(1-O(\eps))$-fraction of the pairs $u,v$, \[\E_{\substack{P \sim \cD|u,v}}[\mathbb{1}_{\exists e \in P, e \notin \cE_f}] \leq 1/2.\]
This means that when the algorithm terminates, 
for at least $(1-O(\eps))$ fraction of vertex pairs $u,v$ there is at least one path of length at most $\ell$ between them. We argue
that for each such pair $u,v$, the 
collection $\cP_{i,j}$ already contains a path between $u$ and $v$ upon the termination of the algorithm, since we have only deleted edges.
\end{proof}

We now use the short paths $\cP$ between all pairs of vertices in $L$ to argue that an 
adversary that corrupts a small fraction of the edges can only corrupt a small fraction of the paths in $\cP$. Recall that some of the paths in $\cP$ might be invalid, and to simplify notation
we account these as paths that are corrupted by default. Thus, we say that a path in $\cP$ is corrupted if it equals $\perp$ or if at least one of the edges in it is corrupted. 

\begin{lemma}\label{lem:inside-link}
Fix $\eps >0$ and a vertex link $L$ of $X$. Then there is a $\poly(|L|)$-time algorithm to construct a set of paths $\cP=\{P_{U,V}\}_{U\neq V \in L}$ in which each valid path has length at most $O(1/\eps^{1/8})$. 
Furthermore, any adversary that corrupts at most $\eps$-fraction of the edges in $L$, corrupts at most $O(\eps^{1/8})$-fraction of the paths in 
$\cP$.
\end{lemma}
\begin{proof}
For every pair of indices $i,j \in [d-1]$ with $|j-i| \geq \eps^{1/8}(d-1)$, run the polynomial time algorithm in \Cref{lem:path-ij} with the parameter $\eps^{1/8}$, to get a set of paths $\cP_{ij}$ of length $O(1/\eps^{1/8})$ each, between all pairs $u\in L_i(1), v\in L_j(1)$ such that at most $O(\eps^{1/8})$-fraction of the paths are invalid and  every edge in $L_{ij}(2)$ is used in at most $O\left(\frac{|L_i(1)||L_j(1)|}{\eps^{3/8}|L_{ij}(2)|}\right)$ paths. For every pair $i,j\in [d-1]$ with $|j-i|<\eps^{1/8}d$, set $\cP_{ij}=\{P(u,v)\}_{u\in L_i(1),v\in L_j(1)}$ with $P(u,v)=\perp$.

Fix any adversary that corrupts at most $\eps$-fraction of the edges in $L$, and let $\cE$ denote the set of corrupted edges. Let the fraction of corrupted edges inside $L_{ij}(2)$ be denoted by $\mu_{ij}(\cE)$. Let $\good$ denote the set of indices $i,j \in [d-1]$ that satisfy, $|j-i| \geq \eps^{1/8}(d-1)$ and $\mu_{ij}(\cE) \leq \sqrt{\eps}$. By 
Markov's inequality it follows that the 
a random pair $i,j\sim [d]$ is in $\good$ 
with probability at least $1-\eps^{1/8}-\sqrt{\eps}\geq 1-O(\eps^{1/8})$.

Fix any pair $(i,j) \in \good$ and consider the set of paths $\cP_{ij}$. Recall that say that a path $P(u,v)$ is corrupted if  either $P(u,v)=\perp$ or any of its edges is corrupted. Since at most $\sqrt{\eps}|L_{ij}(2)|$ edges are corrupted and each such edge is used in at most $O\left(\frac{|L_i(1)||L_j(1)|}{\eps^{3/8}|L_{ij}(2)|}\right)$ paths, we get that the number of corrupted paths in $\cP_{ij}$ is at most 
\[
\sqrt{\eps}|L_{ij}(2)| \cdot \Theta\left(\frac{|L_i(1)||L_j(1)|}{\eps^{3/8}|L_{ij}(2)|}\right)\lll \eps^{1/8}|L_i(1)||L_j(1)|.
\]
It follows from the union bound that 
the fraction of corrupted paths is at most
\[
\Pr_{i,j\sim [d]}[(i,j)\notin \good]+\Pr_{\substack{i,j\sim [d]\\ u\sim L_i(1)\\v\sim L_j(1)}}[P(u,v) \text{ is corrupted}\mid (i,j)\in \good]
\lll 
\eps^{1/8}.\qedhere
\]
\end{proof}

\subsection{Moving to the Zig-Zag Product of \texorpdfstring{$G$}{G}}\label{sec:zz}
As discussed earlier in Section~\ref{sec:regularity-issues}, the fact that graph $G$ is not regular introduces several technical difficulties; this prohibits us 
from using Theorem~\ref{thm:expander-routing} directly, and it also poses difficulties later on when we use routing schemes to embed a regular 2-CSP instance onto $G$. In this section we circumvent these issues by considering the zig-zag product of $G$ with appropriate expanders. The resulting graph $Z$ will be
a virtual, regular graph; by virtual, we mean that the graph $Z$ will help
us in choosing appropriate paths
in $G$ (for the routing) and in the analysis. The routing itself is still done on the graph $G$.

For a set $S \in X(s)$, let $\Delta_S=|X_S(d-s)|$. Henceforth we will think of $G$ as an undirected graph without weights, and for that we replace each edge $(u,v)$ with $\Delta_{\{uv\}}$ parallel edges. Since the probability of drawing $(u,v)$ from $\mu_2$ is $\frac{\Delta_{\{u,v\}}}{\binom{d}{2}|X(d)|}$, this gives us the same random walk matrix over $X(1)$.

We use the zig-zag product to get a regular graph $Z$ from $G$. The additional benefit of this operation is that $Z$ has constant degree graph while also being an expander. The zig-zag product was first defined in~\cite{ReingoldVW00} and is typically stated for regular graphs, but below we state a similar construction for irregular graphs $G$. We follow the exposition from the lecture notes~\cite{lec-zig-zag}, except that we use the extension to irregular graphs.

\subsubsection{The Replacement Product and the Zig-Zag Product}
The following fact is well-known.
\begin{lemma}\label{lem:exp-construction}
For all $\sigma>0$, there exist $k,m_0\in \N$ and a family of $k$-regular graphs $\cH=\{H_m\}_{m \geq m_0}$ which is polynomial-time constructible, where for each $m\geq m_0$ that graph $H_m$ has $m$ vertices and $\sigma_2(H_m)\leq \sigma$.
\end{lemma}
For $\sigma$ to be determined later, fix the family $\cH$ of expander graphs as in Lemma~\ref{lem:exp-construction}, and fix an undirected (possibly irregular) graph $G$ with minimum degree at least $m_0$. 

Our presentation of the zig-zag product follows~\cite{lec-zig-zag} 
almost verbatim, and it is convenient to first define the replacement product of $G$ with $\cH$, denoted by $G \rcirc \cH$. Assume that for each vertex of $G$, there is some ordering on its $D$ neighbors. Then the replacement product $G \rcirc \cH$ is constructed as follows:
\begin{itemize}
  \item Replace a vertex $u$ of $G$ with a copy of $H_{\deg(i)}$, that is the graph from $\cH$ on $\deg(i)$ many vertices (henceforth called a cloud). For $u \in V(G)$, $c \in V(H_{\deg(i)})$, let $(u,c)$ denote the $c^{th}$ vertex in the cloud of $u$.
  \item Let $(u, v) \in E(G)$ be such that $v$ is the $c_1^{th}$ neighbor of $u$ and $u$ is the $c_2^{th}$ neighbor of $v$. Then $((u,c_1), (v,c_2)) \in E(G \rcirc \cH)$. Also $\forall u \in V(G)$, if $(c_1, c_2) \in E(H_{\deg(i)})$, then $((u, c_1), (u, c_2)) \in E(G\rcirc \cH)$.
\end{itemize}
Note that the replacement product constructed as above has $2|E(G)|$ vertices and is $(k+1)$-regular. The zig-zag product $G \zz H$ is  constructed as follows:
\begin{itemize}
  \item The vertex set $V(G \zz \cH)$ is the same as that of $G\rcirc \cH$.
  \item $((u, c_1), (v, c_4)) \in E(G \zz \cH)$ if there exist $c_2$ and $c_3$ such that $((u, c_1), (u, c_2))$, $((u, c_2), (v, c_3))$ and $((v, c_3), (v, c_4))$ are edges in $E(G \rcirc H)$, i.e.\ $(v, c_4)$ can be reached from $(u, c_1)$ by taking a step in the  cloud of $u$ to go to $(u,c_2)$, then a step between the clouds of $u$ and $v$ to go to $(v,c_3)$, and finally a step in the cloud of $v$ to reach $(v,c_4)$.
\end{itemize}
It is easy to see that the zig-zag product is a $k^2$-regular graph on $2|E(G)|$ vertices. Given that $G$ is an expander, the zig-zag product graph $G \zz \cH$ is also an expander. The proof when $G$ is a regular graph can be found in~\cite{lec-zig-zag,ReingoldVW00}. The proof for the irregular case above is exactly the same hence we omit it here. 

\begin{lemma}\label{lem:zz}
If $G$ is a graph on $n$ vertices with $\sigma_2(G)\leq \alpha$, and $\cH=\{H_m\}_{m \geq m_0}$ is a family of $k$-regular graphs with $\sigma_2(H_m) \leq \beta$ for all $m \geq m_0$, then $G \zz \cH$ is a $k^2$-regular graph with second largest singular value at most $\alpha + \beta + \beta^2$.
\end{lemma}

Along with the above statement, we will use the following obvious but important fact,
\begin{fact}\label{fact:distribution-zz}
For $G=(X(1),X(2))$ and $Z=G \zz \cH$, the distribution that samples a uniformly random vertex $(v,c)$ of $Z$ and outputs $v$, is equal to $v \sim X(1)$. Similarly the distribution that samples a uniformly random edge $((u,c_1),(v,c_2))$ of $Z$ and outputs $(u,v)$ is equal to the distribution over edges $(u,v)\sim X(2)$.   
\end{fact}

\subsection{The Routing Protocol using Links}\label{sec:link-protocol}
Fix $\alpha>0$ from Theorem~\ref{thm:expander-routing}. Consider the graph $G=(X(1),X(2))$, viewed as an undirected and unweighted graph with multi-edges. Let $\cH$ be the family of expanders from \Cref{lem:exp-construction} with $\sigma_2(\cH)\leq \frac{\alpha}{2}$ and $Z=G \zz \cH$. For every vertex $v\in Z$ we will let  $v_1$ denote its component in $X(1)$ and $v_2$ denote its position in the cloud of $v_1$, that is, $v=(v_1,v_2)$. We have the following link-to-link transfer lemma:

\begin{lemma}\label{lem:link-routing}
There exist $\eps_0,\alpha,d_0>0$ such that for all $C \in \N$ there exists a constant $C'\in \N$ such that for large enough $n \in \N$ the following holds. Let $X$ be a complex from \Cref{thm:cl} with dimension $d \geq d_0$, $|X(1)|=n$ and the parameter $q=\log^{C'} n$. Let $Z=G\zz \cH$ be the zig-zag product of $G$ with the family of expander graphs with $\sigma_2(\cH)\leq \frac{\alpha}{2}$ as in \Cref{lem:exp-construction}. 
Let $\pi:V(Z)\rightarrow V(Z)$ be any permutation. Then there is a routing protocol on $G=(X(1),X(2))$ with round complexity $T=O(\log n)$ and work complexity $q^{O(d^2)}\log |\Sigma|$ such that for all initial functions $A_0:\cS \rightarrow \Sigma$ satisfying
\[\Pr_{v\sim V(Z)}[\maj_{0.99}(A_0(v,u)\mid u\sim X_{v_1}(1))\neq \perp] \geq 1-\eta,\]
and for all possible adversaries that corrupt at most $\eps$-fraction of edges with $\eps\leq \eps_0$, the protocol computes the function $A_T:\cS\rightarrow \Sigma$ satisfying
\[\Pr_{v\sim V(Z)}[\maj_{0.99}(A_T(\pi(v),w)\mid w\sim X_{\pi(v)_1}(1))=\maj_{0.99}(A_0(v,u)\mid u\in X_{v_1}(1))]\geq 1-\eta-\frac{1}{\log^C n}.\]

\end{lemma}
\begin{proof}
Fix a permutation $\pi$ on $V(Z)$ and an initial  function $A_0:\cS\rightarrow\Sigma$. 
\vspace{-2ex}
\paragraph{Setting up paths over $1$-links:} Using \Cref{lem:zz} we know that the zig-zag product $Z=G\zz \cH$ is a $k$-regular graph, for $k=\Theta(1)$, with $\sigma_2(Z)\leq \frac{2}{d}+\frac{\alpha}{2}+\frac{\alpha^2}{4}\leq \alpha$. By Theorem~\ref{thm:expander-routing}, we may fix a relaxed-pebble-routing protocol $\cR$ that has $|V(Z)|$ paths each of length $T' \leq O(\log N)$ where every vertex and edge in $Z$ is used in at most $\log^{C_1}n$ paths, for some constant $C_1$ that depends on $C$, and at most $\frac{1}{3\log^C n}$ paths are invalid.
It will be convenient for us notationally to have all of the paths in $\cR$ have the same length $T'$, and we do so by repeating the final vertex of the paths the amount of times necessary.

Each path in $\cR$ is of the form $P= u_1\rightarrow \ldots \rightarrow
u_{T'}$, with $u_i\in V(Z)$, $u_{T'}=\pi(u_1)$ and for all $i$, $(u_{i},u_{i+1})$ is an edge in the zig-zag product $Z=G\zz \cH$. This implies that $((u_{i})_1,(u_{i+1})_1)$ must be an edge in $G$. Therefore given $P$, we will think of the following  link to link message transfer, 
\[X_{(u_1)_1}(1) \rightarrow X_{(u_{1})_1,(u_{2})_1}(1) \rightarrow X_{(u_{2})_1}(1) \rightarrow \ldots \rightarrow X_{(u_{T'})_1}(1),\]
to implement the required transfer from $X_{(u_{1})_1}(1)$ to $X_{(u_{T'})_1}(1)$.

Let $T:=2T'$. Expand each path in $\cR$ as shown above into $1$-links connected by $2$-links. For even $t \in [0,T]$ let $L_{j,t}$ denote the $1$-link that occurs in the $(t/2)^{th}$-time-step of the $j^{th}$ path in $\cR$, and for odd $t\in [T]$ let $L_{j,t}$ denote the intermediate $2$-link that is the  intersection of the 1-links $L_{j,t-1}$ and $L_{j,t+1}$. 

\vspace{-2ex}
\paragraph{Setting up paths inside $1$-links:} For every $1$-link $X_u$ for $u\in X(1)$, we use the algorithm in \Cref{lem:inside-link} with the parameter $\sqrt{\eps}$, to construct a collection of short paths $\cP_{u} = \{P_u(v,w)\}_{v,w\in X_u(1)}$ between all pairs of vertices $v,w$ inside the link $X_u$. We will refer to these as the internal paths in $X_u$.

\vspace{-2ex}
\paragraph{The description of the routing protocol:}
for each path $j$ and time-step $t > 0$, each vertex $u\in L_{j,t}$ takes the majority of the values it receives from $v \in L_{j,t-1}$; this occurs through the path $P_{L_{j,t-1}}(v,u)$ if $t$ is odd or the path $P_{L_{j,t}}(v,u)$ if $t$ is even (some paths may be invalid, in which case we interpret the received value on them as ``$\perp$''). Then $u$ passes on this value to $w \in L_{j,t+1}$ through the path $P_{L_{j,t}}(u,w)$ if $t$ is even and the path $P_{L_{j,t+1}}(u,w)$ if $t$ is odd. To formalize this, let the ``outgoing message'' from $u$ at time-step $0$, path $j$ be $\out(u,j,0)=A_0(v_j,u)$ where $v_j \in V(Z)$ is the start vertex on the $j^{th}$ path. For a vertex $u \in X(1)$, for every path $j\in [N]$ and time-step $t\in [T]$ where $u\in L_{j,t}$, $u$ maintains a set of ``incoming messages'' $\In(u,v,j,t)$ that it receives from vertices $v\in L_{j,t-1}$. The vertex $u$ then sets its outgoing message as $\out(u,j,t)=\maj_{v\sim L_{j,t-1}}(\In(u,v,j,t))$ if more than $1/2$-fraction (computed according to the distribution on $L_{j,t-1}$) of the list has the same value, else sets it to $\perp$.
This outgoing message is then sent to every vertex in $L_{j,t+1}$ through the  paths $\cP_{L_{j,t}}$ if $t$ is even and $\cP_{L_{j,t+1}}$ if $t$ is odd.

\vspace{-2ex}
\paragraph{Bad links and bounding them:}
We now begin the analysis of the protocol, and for that we need to introduce a few notions. First, at most $\eps$-fraction of the edges are corrupted, and we denote the set of corrupted edges by $\cE \subseteq X(2)$. A $1$-link $X_u$ is called bad if it contains too many corrupted edges, more precisely, if $\mu_u(\cE) \geq \sqrt{\eps}$ and good otherwise. 
\begin{claim}\label{claim:bad-1-link}
$\Pr_{u \sim X(1)}[X_u \text{ is bad}] \leq \frac{\poly(d)}{q}$.
\end{claim}
\begin{proof}
Deferred to \Cref{sec:omitted}.    
\end{proof}

For any $1$-link $X_u$ and $v,w\in L(1)$, we say an internal path $P_{u}(v,w)$ is corrupted if it equals $\perp$ or any of the edges on it are corrupted. We define the set $\cD_u \subseteq X_u(1)$ of doomed vertices of $X_u$ as those vertices $v\in X_u(1)$ for which at least $\eps^{1/32}$-fraction of the paths $P_{u}(v,w)$ for $w \sim X_u(1)$ are corrupted, i.e.\ 
\[
\cD_u = 
\left\{v \in X_u(1): \Pr_{w\sim X_u(1)}[P_{u}(v,w) \text{ is corrupted}] \geq \eps^{1/32}\right\}.
\]
The following claim asserts that a good $1$-link cannot have too many doomed vertices.
\begin{claim}\label{claim:doomed}
If $X_u$ is a good link then,
$\Pr_{v \sim X_u(1)}[v \in \cD_u] \lll \eps^{1/32}$.
\end{claim}
\begin{proof}
Deferred to \Cref{sec:omitted}.    
\end{proof}

A $2$-link $X_{u,v}$ is said to be bad if either $X_u$ or $X_v$ is a bad $1$-link, or one of $\mu_{uv}(\cD_u)$ or $\mu_{uv}(\cD_v)$ is at least $\eps^{1/64}$. 
\begin{claim}\label{claim:bad-2-link}
$\Pr_{(u,v) \sim X(2)}[X_{u,v} \text{ is bad}] \leq \frac{\poly(d)}{q}$.
\end{claim}
\begin{proof}
Deferred to \Cref{sec:omitted}.    
\end{proof}

\vspace{-2ex}
\paragraph{Link to Link Transfer on Good Paths:} We now use Claims~\ref{claim:bad-1-link},~\ref{claim:doomed},~\ref{claim:bad-2-link} to finish the analysis of the protocol. Recall that every path contains a $1$-link at an even time-step and a $2$-link at an odd time-step.
Consider any path $j$ where (1) $\maj_{0.99}(A_0(v_j,u))\neq \perp$, (2) it is a valid path in $\cR$, and (3) for all times steps $t$, $L_{j,t}$ is a  good link. On such a path we will show that $\maj_{0.99}(\out(u,j,t)\mid u\in L_{j,t})=\maj_{0.99}(\out(v,j,t+1))$ for all $t$. After that, we argue that almost all paths satisfy these properties. The proof of the former fact is broken into two steps: the message transfer from $L_{j,t}$ to $L_{j,t+1}$ and then the message transfer from $L_{j,t+1}$ to $L_{j,t+2}$, and we argue about each step separately.

The argument proceeds by induction on $t$. Fix some even $t \in [T-2]$ and let $\maj_{0.99}(\out(v,j,t))=\sigma \neq \perp$. The vertices $u \in L_{j,t+1}$ that are not doomed with respect to $L_{j,t}$ ($u \notin \cD_{L_{j,t}}$) will receive the value $\sigma$ on at least $1-O(\eps^{1/32})-0.01 \geq 1/2$-fraction of the paths $P_{L_{j,t}}(v,u)$ and therefore will compute the correct majority, setting $\out(u,j,t+1)=\sigma$. Since $L_{j,t+1}$ is a good $2$-link we know that at most $O(\eps^{1/64})\leq 0.01$-fraction of its vertices are doomed, which gives that $\maj_{0.99}(\out(u,j,t+1)|u\in L_{j,t+1})=\sigma$.

The argument for odd $t$ is exactly the same. Fix some odd $t \in [T-2]$ and let $\maj_{0.99}(\out(v,j,t))=\sigma \neq \perp$. The vertices $u \in L_{j,t+1}$ that are not doomed with respect to $L_{j,t}$ ($u \notin \cD_{L_{j,t}}$) will receive the value $\sigma$ on at least $1-O(\eps^{1/32})-0.01 \geq 1/2$-fraction of the paths $P_{L_{j,t}}(v,u)$ and therefore will compute the correct majority, setting $\out(u,j,t+1)=\sigma$. Since $L_{j,t+1}$ is a good $1$-link, \Cref{claim:doomed} implies that at most $O(\eps^{1/32})
\leq 0.01$-fraction of its vertices are doomed, which immediately implies that $\maj_{0.99}(\out(u,j,t+1)|u\in L_{j,t+1})=\sigma$. 

Setting $A_T(\pi(v_j),v)=\out(v,j,T)$ 
we conclude that for the $j^{th}$-path, $\maj_{0.99}(A_T(\pi(v_j),v))=\maj_{0.99}(A_0(v_j,u))$. 

\vspace{-2ex}
\paragraph{Bounding the number of Good Paths:}
We now finish off the proof by calculating the number of paths $j$ satisfying conditions (1), (2) and (3) above. By the assumption of the lemma we know that there are at most $\eta$-fraction paths violating (1) and by construction of $\cR$ at most $\frac{1}{3\log^C n}$-fraction paths violate (2).

To account for condition (3) it will be useful to switch back to the equivalent view of $\cR$ as a set of paths over $Z$. We call a vertex $v\in Z$ bad if the corresponding $1$-link $X_{v_1}$ is bad and an edge $(v,w)\in Z$ bad if the corresponding $2$-link $X_{v_1,w_1}$ is bad. Using  Claims~\ref{claim:bad-1-link} and~\ref{claim:bad-2-link} we get
\[\Pr_{v \sim Z}[v \text{ is bad}] \leq \frac{\poly(d)}{q},
\qquad\qquad\qquad\Pr_{(v,w) \sim Z}[(v,w) \text{ is bad}] \leq \frac{\poly(d)}{q},\]
since by \Cref{fact:distribution-zz} we have equality of the distributions in question.

Now note that condition (3) is equivalent to saying that path $j$ contains only good vertices and good edges (from $Z$) in it. The protocol $\cR$ uses every vertex at most $\log^{C_1}n$ times which implies that at most $\frac{\poly(d)}{q}\cdot |V(Z)|\cdot \log^{C_1}n\leq \frac{|V(Z)|}{3\log^C n}$-paths contain bad vertices on them (by setting $q$ to be a large enough polynomial of $\log n$). Similarly since $\cR$ uses an edge in $Z$ at most $\log^{C_1}n$ times, we get that at most $\frac{\poly(d)}{q}\cdot |E(Z)|\cdot \log^{C_1}n\leq \frac{|V(Z)|}{3\log^C n}$ where we used that $|E(Z)|=\Theta(|V(Z)|)$ and $q$ is large enough. Therefore by a union bound we get that at most $\eta+\frac{1}{\log^C N}$-fraction paths fail in transmitting the majority symbol to the link $L_{\pi(v_j)}$ correctly, giving us the conclusion in the lemma.
\end{proof}

\subsubsection{Proofs of Omitted Claims}\label{sec:omitted}
In this section we give the proofs of several 
claims used throughout the proof of Lemma~\ref{lem:link-routing}.
\begin{claim}[\Cref{claim:bad-1-link} restated]
$\Pr_{u \sim X(1)}[X_u \text{ is bad}] \leq \frac{\poly(d)}{q}$.
\end{claim}
\begin{proof}
Recall that $\cE$ is a set of corrupted edges, and let $\cE_{jk}$ denote the set of edges $\cE \cap X_{jk}(2)$ with $\mu_{jk}(\cE)$ denoting its measure in $X_{jk}(2)$. Note that $\E_{j\neq k\in [d]}[\mu_{jk}(\cE)]=\mu(\cE)\leq \eps$. 

Fix some $i \in [d]$ and consider the bipartite graph, $B_i=(X_i(1), \cup_{j\neq k \in [d]\setminus i} X_{jk}(2))$. We say that a vertex in the right side of $B_i$ is corrupted if the corresponding edge belongs to $\cE$. 

First note that for all $u\in X_i(1)$, the link $X_u$ is bad if the fraction of $u$'s neighbors in $B_i$ that are corrupted is at least $\sqrt{\eps}$. Therefore let us bound the probability that this event occurs. By item 5 of \Cref{thm:cl} and \Cref{claim:gll-sing-val}, the bipartite graphs $(X_i(1),X_{jk}(2))$ have second largest singular value at most $\poly(d)/\sqrt{q}$ for all $j \neq k$. Therefore the second largest singular value of $B_i$ is also at most $\poly(d)/\sqrt{q}$. Applying \Cref{lem:sampling} we get
\[\Pr_{u \sim X_i(1)}[X_u \text{ is bad}] \leq \frac{\poly(d)}{q},\]
where we used that the expected fraction of bad neighbors is $\E_{j \neq k \in [d]\setminus i}[\mu_{jk}(\cE)]\leq 2\eps$.

Since the above bound holds for all $i$, we get the conclusion in the lemma.
\end{proof}

\begin{claim}[\Cref{claim:doomed} restated]
If $X_u$ is a good link then,
$\Pr_{v \sim X_u(1)}[v \in \cD_u] \lll \eps^{1/32}$.
\end{claim}
\begin{proof}
Since $X_u$ is a good link we know that $\mu_u(\cE)\leq \sqrt{\eps}$. By the construction of the internal paths $\cP_u$ from \Cref{lem:inside-link} we know that at most $O(\eps^{1/16})$-fraction of the paths $P_u(v,w)$ for $v,w\sim X_u(1)$ are corrupted. Therefore by Markov's inequality we get that for at most $O(\eps^{1/32})$-fraction of $v\sim X_u(1)$,  $\Pr_{w\sim X_u(1)}[P_u(v,w)\text{ is corrupted}]\geq \eps^{1/32}$. Therefore $v\in \cD_u$ with probability at most $O(\eps^{1/32})$.
\end{proof}

\begin{claim}[\Cref{claim:bad-2-link} restated]
$\Pr_{(u,v) \sim X(2)}[X_{u,v} \text{ is bad}] \leq \frac{\poly(d)}{q}$.
\end{claim}
\begin{proof}
Recall that the link $X_{u,v}$ is bad if either $X_u$ or $X_v$ is a bad $1$-link or one of $\mu_{uv}(\cD_u)$ or $\mu_{uv}(\cD_v)$ is at least $\eps^{1/64}$. Using \Cref{claim:bad-1-link} we can bound the probability that one of $X_u$ or $X_v$ is bad, so let us now bound the probability of the latter events.

Fix a good $1$-link $X_u$ henceforth and let us bound the fraction of $v\sim X_u(1)$ for which $\mu_{uv}(\cD_u)$ is large. Using \Cref{claim:doomed} we get that $\mu_u(\cD_u)\lll \eps^{1/32}$. The link $X_u$ is a $(d-1)$-partite complex with colors $[d-1]$ and the property that for all $i\neq j \in [d-1]$ the bipartite graph $(X_{u,i}(1),X_{u,j}(1))$ has second largest eigenvalue at most $O\left(\frac{1}{\sqrt{q}}\right)$. Let $\mu_j(\cD_u)$ denote the measure of $\cD_u \cap X_{u,j}(1)$ inside $X_{u,j}(1)$. One can check that $\E_{j\sim [d-1]}[\mu_j(\cD_u)]=\mu_u(\cD_u)$.

Fix a color $i \in [d-1]$ of $X_u$ and consider the bipartite graph, $B_i=(X_{u,i}(1), \cup_{j \in [d-1]\setminus i} X_{u,j}(1))$. This graph has second largest singular value at most $O\left(\frac{1}{\sqrt{q}}\right)$.
We say that a vertex in the right side of $B_i$ is doomed if it belongs to $\cD_u$. Noting that for all $v\in X_i(1)$, $\mu_{uv}(\cD_u)$ is the fraction of doomed neighbors of $v$ in $B_i$ and applying \Cref{lem:sampling} we get
\[\Pr_{v \sim X_{u,i}(1)}[\mu_{uv}(\cD_u) \geq \eps^{1/64}] \leq \frac{\poly(d)}{q},\]
where we used that the expected fraction of bad neighbors is $\E_{j \sim [d-1]\setminus\{i\}}[\mu_{j}(\cD_u)]\leq 2\mu_u(\cD_u)\lll \eps^{1/32}$. Since the above bound holds for all $i$, we get that,
\[\Pr_{v \sim X_{u}(1)}[\mu_v(\cD_u) \geq \eps^{1/64}] \leq \frac{\poly(d)}{q}.\]
We can now bound the fraction of $2$-links that are bad by a union bound:
\[\Pr_{(u,v)\sim X(2)}[X_{u,v}\text{ is bad}]\leq 2\Pr_{u\sim X(1)}[X_u \text{ is bad}]+2\Pr_{v \sim X_u(1)}[\mu_{uv}(\cD_u) \geq \eps^{1/64}\mid X_u \text{ is good}]\leq \frac{\poly(d)}{q},\]
where we used \Cref{claim:bad-1-link} to bound the first term.
\end{proof}

\subsection{A Routing Protocol based on Cliques}\label{sec:clique_to_clique}
In this section, we 
formally state the performance of 
the clique-to-clique 
routing protocol. 
This protocol is similar
in spirit to the link-to-link protocol and 
it is where most of our
intuition comes from. 
\begin{lemma}\label{lem:clique-routing}
There is $\eps_0 \in (0,1)$ such that for all $C \in \N$, for large enough $n \in \N$ the following holds. Let $X$ be a $d$-dimensional complex with $|X(1)|=n$, $|X(d)|=N$, $d=\Theta(\log\log^2 n)$ and $N \leq n^2$, that is a $\gamma$-one-sided local spectral expander with $\gamma < 1/\poly(d)$. Let $\pi:X(d)\rightarrow X(d)$ be any permutation. Then there is a routing protocol on $G=(X(1),X(2))$ with round complexity $T=O(\log N)$ and work complexity $\max_u O(Td|X_u(d-1)|)$ such that for all initial functions $A_0:X(d)\times [d]\rightarrow \Sigma$ satisfying
\[\Pr_{D\sim X(d)}[\maj_{0.99}(A_0(D,u)\mid u\in D)\neq \perp] \geq 1-\eta,\]
and for all possible adversaries that corrupt at most $\eps$-fraction of edges with $\eps\leq \eps_0$, the protocol computes the function $A_T:X(d)\times [d]\rightarrow \Sigma$ satisfying
\[\Pr_{D\sim X(d)}[\maj_{0.99}(A_T(\pi(D),v)\mid v\in \pi(D))=\maj_{0.99}(A_0(D,u)\mid u\in D)]\geq 1-\eta-\frac{1}{\log^C N}.\]
\end{lemma}
\begin{proof}
    We omit the formal
    proof and instead provide a brief sketch, as it is 
    very similar to the proof of Lemma~\ref{lem:link-routing} (also, we do not use this statement later on
    in the paper). 
    Analogously to the 
    proof therein, one 
    sets up a collection of paths, this time in the clique to clique graph, whose vertices are $X(d)$ and edges correspond to cliques that intersect in size $\alpha d$ (the zig-zag product
    trick is not necessary in this case). One defines the notion 
    of ``bad vertices'', 
    which are vertices that have many 
    corrupted edges adjacent to them, and 
    subsequently defines the notion of bad cliques, which are cliques 
    that contain many bad vertices. In contrast to the proof of 
    Lemma~\ref{lem:link-routing}, an upper bound of the fraction of bad cliques is not established by spectral bounds this time; instead one appeals to the Chernoff-type bound of~\cite{dikstein2024chernoff}, which gives bounds of the form $2^{-\Theta_{\eps}(\sqrt{d})}$. This
    is ultimately the reason that the argument requires the dimension of the complex to be 
    somewhat super constant.
\end{proof}

\section{Embedding a PCP on an HDX}\label{sec:routing-to-pcp}
The works of~\cite{polishchuk1994nearly,DinurMeir} used routing networks to transform the graph underlying 2-CSPs to an explicit graph. Towards this end, they used a pebble routing protocol on De-Bruijn graphs.
In this section we  generalize their argument and show that any tolerant routing protocol on $G$ gives rise to a PCP embedding result on $G$. 
We then apply this connection to our specific link-to-link routing protocol. We show that in this case, 
this connection is in fact stronger and gives a gap amplification statement.

\subsection{Connection between Routing Protocols and PCPs}
The transformation we describe for modifying the underlying graph of a $2$-CSP has one significant downside: it increases the alphabet size considerably. 
Since the routing length is always $\Omega(\log n)$ on constant degree graphs, the alphabet size always increases to be at least polynomial. In fact, since in our case the work complexity is poly-logarithmic, the alphabet size will increase further and will be $2^{\poly(\log n)}$. 
Therefore, to 
facilitate alphabet reduction 
steps later on, 
we need a more 
refined notion related to 
the alphabet size of CSPs, 
called the \emph{decision complexity}. 
Additionally, to simplify the presentation, we also generalize the notion of $2$-CSPs, and allow for varying alphabets $\Sigma(u) \subseteq \Sigma$ for the vertices $u \in G$.\footnote{This helps us to restrict the provers' strategy to be one that satisfies additional constraints. This technique is usually referred to as folding, and it makes the soundness analysis slightly cleaner.}

\begin{definition}\label{def:gen-2-csp}
An instance $\Psi = (G=(V,E), \Sigma, \{\Sigma(u)\}_{u\in V}, \{\Psi_{e}\}_{e\in E})$ of a generalized $2$-CSP consists of a weighted graph $G$, alphabets $\Sigma,\{\Sigma(u)\}$ with $\Sigma(u)\subseteq \Sigma$ for all $u\in V$, and constraints $\Psi_e\subseteq \Sigma\times\Sigma$, one for each edge. The decision complexity of $\Psi$ is defined as the maximum, over all edges $e = (u,v)$, of the sum of the following circuit complexities: the circuit complexity of checking membership in $\Psi_{(u,v)}$ i.e. 
the circuit complexity of deciding if $(\sigma,\sigma') \in \Psi_{(u,v)}$, the circuit complexity of checking membership in $\Sigma(u)$, and the circuit complexity of checking membership in $\Sigma(v)$.

\end{definition}

Informally, the decision complexity
of a CSP is the complexity of the 
constraints of it. In many PCP reductions, the alphabet size is a constant, and as the decision complexity is always upper bounded by ${\sf poly}(|\Sigma|)$, it is often omitted from the discussion. In our case the decision complexity will 
be ${\sf poly}(\log|\Sigma|)$, 
and it is closely related to
the notion of work complexity
in the context of the almost everywhere reliable transmission problem (as we exhibit next).

The following lemma translates routing protocols for a graph $G$ to PCPs on the graph $G$, generalizing~\cite[Lemma 3.4]{DinurMeir}. 
\begin{lemma}[\Cref{lem:routing-to-pcp-general} restated]\label{lem:routing-to-pcp-general-restated}
Suppose $G$ is a regular graph on $2n$ vertices that has an $(\eps,\nu)$-edge-tolerant routing protocol on the alphabet $\Sigma$ with work complexity $W$, that can be constructed in time $\poly(n)$. Then there is a $\poly(n)$ time reduction that, given a 2-CSP instance $\Psi'$ on a $k$-regular graph $H$ and alphabet $\Sigma$, with $|V(H)|\leq n$, produces a 2-CSP instance $\Psi$ on $G$ such that:
\begin{itemize}
\item If $\val(\Psi')=1$ then $\val(\Psi)=1$.
\item If $\val(\Psi')\leq 1-8\nu$ then $\val(\Psi)\leq 1-\eps$.
\item The alphabet size of $\Psi$ is at most $|\Sigma|^{kW}$. 
\item The decision complexity of $\Psi$ is $O(W)+O(k|\Sigma|)$.
\end{itemize}
\end{lemma}

\begin{proof}
We begin by reducing the 2-CSP $\Psi'$ to a 2-CSP $\Phi$ on a bipartite $2k$-regular graph $G'$ on $|V(G)|$ vertices such that if $\Psi'$ is satisfiable then so is $\Phi$ and if $\val(\Psi')\leq 1-8\nu$ then $\val(\Phi)\leq 1-\nu$. 
The goal here is to obtain a regular bipartite graph (which can hence be decomposed 
into perfect matchings), as well as align 
the number of vertices 
of $\Phi$ to match 
the number of vertices in $G$.
\vspace{-2ex}
\paragraph{Reducing to $\Phi$:}
Let $n= a\cdot |V(H)|+r$ for some $a,r\in \N$ with $r\leq |V(H)|$. Define a graph $G''$ on the vertex set $[n]$ which is formed by taking $a$-many disjoint copies of the $k$-regular graph $H$ and an arbitrary $k$-regular graph (possibly with multi-edges) on the remaining $r$ vertices. Define a 2-CSP $\Psi''$ on $G''$ which is equal to $\Psi'$ on each of the disjoint copies of $H$ and has the ``equality'' constraint on the edges of the remaining $r$ vertices. It is easy to see that if $\val(\Psi')=1$ then $\val(\Psi'')=1$ and if $\val(\Psi')\leq 1-8\nu$ then $\val(\Psi'')\leq 1-4\nu$. 

We will now reduce the 2-CSP $\Psi''$ to a 2-CSP $\Phi$ whose constraint graph $G'=(L\cup R, E')$ is bipartite. Let $L,R$ be equal to $V(G'')$. The left vertex set is defined as $L = V(G'')\times \{1\}$ and the right vertex set is defined as $R = V(G'')\times \{2\}$.  Let $E'$ be the set of edges $((u,1),(v,2))$ where $(u,v)$ is an edge in $G''$ and additionally add $k$ edges between each pair $((u,1),(u,2))$. For the former edges put the constraints corresponding to $\Psi''$ and for the latter put in equality constraints. The completeness is clear, so let us prove the soundness of the reduction. Assume that $\val(\Phi)\geq 1-\nu$ via the assignment $A$. Define the assignment $B$ on $G''$ as $B(u)=A((u,1))$. For at least $1-2\nu$-fraction of $u$, $A((u,1))=A((u,2))$, denoted by the set $\good$. 
Sampling a constraint $((u,1),(v,2))$, with probability at least $1-4\nu$ we have that $v$ is in $\good$ and the constraint is satisfied in $\Phi$, in which case $B$ satisfies the constraint $(u,v)$ in $\Psi''$. In particular, 
it follows that ${\sf val}(\Psi'')\geq 1-4\nu$, 
finishing the proof of
the soundness.
\vspace{-2ex}
\paragraph{Routing Setup:}
Let $[2n]$ denote the vertex set of both $G, G'$; abusing notation we will use the same letters to denote a vertex in $G$ and in $G'$. Since the graph $G'$ is a bipartite  $2k$-regular graph, we may partition the edge set of $G'$ into $2k$ perfect matchings, which we denote by $\pi_1,\ldots, \pi_{2k}$. We think of $\pi_1,\ldots,\pi_{2k}$ both as matchings and also as the permutations on $V(G)$. Note that since each $\pi_i$ is a perfect matching we have that $\pi_i^2=\text{id}$. 

Let $\cR_i$ denote the protocol that routes the matching $\pi_i$ on $V(G)$. We know that $\cR_i$ has round complexity $T$ and every vertex $u\in X(1)$ receives at most $T_u$ messages (over all rounds), and has work complexity $W$. We now describe the CSP $\Psi$. 

\vspace{-2ex}
\paragraph{The graph:} 
The graph underlying the CSP $\Psi'$ is $G$.
\vspace{-2ex}
\paragraph{The alphabet:} 
Let $\Sigma$ be the alphabet of $\Phi$. The alphabet of each vertex $u$ is a subset of $\Sigma^{O(kT_u)}$ and 
we think of a label to $u$ as describing the messages received by $u$ throughout all the protocols $\cR_i$. The messages that $u$ sends to its neighbor $v$ at every round is a deterministic function of the messages it has received so far. With this in mind, an assignment to the CSP can be thought of as 
a transcript of the messages that 
were routed in the protocol, 
and our goal in designing the 
constraints is to check that 
the message that is sent by a vertex $u$ to its neighbor $v$ at round $t$, is received correctly by $v$ at round $t+1$.

Formally, we think of an assignment to the CSP as a maps $A_0\colon V(G) \to \Sigma$, and maps $\In_{i,t}\colon 2E(G)\to \Sigma^*$ for each $i\in [2k]$ and round $t\in [T]$; here $2E(G)$ denoted the set of ordered tuples $[u,v]$ for each $(u,v)\in E(G)$. As part of its alphabet, every vertex $u$ holds the value $A_0(u)\in\Sigma$ and the symbols $\In_{i,t}[u,v]$ that specify the messages that $u$ receives from its neighbors $v$ in the protocol $\cR_{i}$ at round $t$. 

\vspace{-2ex}
\paragraph{Viewing the rules of each protocol as maps:} We now define certain maps that are useful for specifying the constraints of $\Psi$. First, the rules of each protocol $\cR_i$ can be described using the following maps $\out_{i,t}\colon 2E(G)\times \Sigma^* \rightarrow \Sigma^*\cup \{\perp\}$, for $i\in [2k], t\in [T-1]$; the symbol $\out_{i,t}[u,v,\sigma]$ is the message that $u$ sends to $v$ in round $t$ in the protocol $\cR_i$ if the messages $u$ received in previous rounds are given by $\sigma$, for a valid  $\sigma$ of the correct length, and otherwise $\out_{i,t}[u,v,\sigma] = \perp$. Given the transcript $\In_{i,t'}[u,\cdot]$ for all $t'\leq t$ (which is part of the assignment to $u$), the message that $u$ sends to $v$ is  $\out_{i,t}[u,v,A_0[\cdot,u]\circ \In_{i,1}[u,\cdot]\ldots \circ \In_{i,t}[u,\cdot]]$ which we will denote by $\out_{i,t}[u,v]$ for brevity. We stress that this value only depends on 
the assignment given to $u$.

Finally define the output of the protocol $\cR_i$ as the map $A_{i,T}$. Formally, for all $i\in [2k]$, consider the map $A_{i,T}\colon V(G) \times \Sigma^* \to \Sigma$ where the symbol $A_{i,T}[u,\sigma]$ specifies the output of the routing algorithm at a vertex $u$ in $\cR_i$ when given as input the transcript $\sigma$, which in our case equals $A_0[\cdot,u]\circ \In_{i,1}[u,\cdot]\circ \ldots \In_{i,T-1}[u,\cdot]$. Again for brevity, we omit the dependence on $\sigma$ as it is clear from context, and use $A_{i,T}(u)$ to denote the corresponding output symbol. 

We emphasize that the maps $A_{i,T}$ and $\out_{i,t}$ are not part of the alphabet, but since they are a deterministic function of the messages received at a vertex, we use them while defining the constraints.

\vspace{-2ex}
\paragraph{Intuition towards defining constraints:}
In an ideal proof we want $A_0$ to be a satisfying assignment to $\Phi$, and the maps $\In_{i,t},\out_{i,t}$ to be the transcript when the protocol $\cR_i$ is executed on $G$.

It is therefore convenient to view the maps $A_0, A_{i,T}$ and $\In_{i,t},\out_{i,t}$ from the point of view of what happens in the routing protocol. We want to ensure that the message transmission across every edge behaves as it is supposed to -- for the edge $(u,v)$ the outgoing message that $v$ sends to $u$ at any round should equal the message that $u$ receives at the next round and vice versa. Note that this check only depends on the alphabet of $u$ and $v$. Secondly, suppose that the routing protocol was successful and that $A_0$ was indeed a satisfying assignment to $\Phi$. Then for every $u$, the protocol $\cR_i$ successfully transmitted the symbol $A_0(\pi_i(u))$ from the vertex $\pi_i(u)$ to the vertex $\pi_i(\pi_i(u))=u$, that is, $A_{i,T}(u)$ equals $A_0(\pi_i(u))$. In particular, we would have that $(A_0(u),A_{i,T}(u))$ would satisfy the constraint $(u,\pi_i(u))$ in $\Phi$. Since this only depends on $u$, we enforce this as a hard constraint on the alphabet of $u$ 
via folding.

\vspace{-2ex}
\paragraph{Folding:} We constrain the label set of 
$u$ to only be tuples where $(A_0(u),A_{i,T}(u))$ satisfies the constraint $\Phi(u,\pi_i(u))$ in $\Phi$, for all $i\in [2k]$. By that, we mean that only labels that satisfy this condition are allowed in an assignment to $\Psi$. 

\vspace{-2ex}
\paragraph{The constraints of $\Psi$:} For an edge  $(v,u)\sim E(G)$, read the labels of $u,v$ and 
for each $i\in [2k]$ and $t\in [T-1]$ check that $\In_{i,t+1}(u,v) = \out_{i,t}(v,u)$ and $\In_{i,t+1}(v,u) = \out_{i,t}(v,u)$.

In words, the constraint on $v,u$ checks
that the message that $u$ receives from $v$ at round $t+1$ is the message that $v$ sent to it at the prior round and vice versa. The decision complexity of the constraints is the sum of the circuit complexity of (1) computing $\out_{i,t}(v,u)$, $\out_{i,t}(u,v)$ and $A_{i,T}(u), A_{i,T}(v)$ over $i$ and $t$, (2) checking $(A_0(u),A_{i,T}(u))\in \Phi(u,\pi_i(u))$ and $(A_0(v),A_{i,T}(v))\in \Phi(v,\pi_i(v))$ over all $i$, and (3) checking if $\In_{i,t+1}(u,v)=\out_{i,t}(v,u)$ and vice versa for all $i,t$. This in total amounts to $O(W)+O(k|\Sigma|)$.

\skipi
This completes the description of $\Psi$, and we now analyze the completeness and the soundness of the reduction.

\vspace{-2ex}

\paragraph{Completeness:} Suppose that ${\sf val}(\Phi)=1$ and that $A:V(G') \rightarrow \Sigma$ 
is satisfying assignment for $\Phi$. We take $A_0(u)=A(u)$ for all $u\in V(G)$ and define the maps $\In_{i,t},\out_{i,t}$ and $A_{i,T}$ according to the execution of the routing protocols $\cR_i$ for each $i\in [2k]$ when instantiated with $A_0$. To argue that this is a valid assignment we must check that it satisfies the folding constraints; to check that its value is $1$, we must verify that it satisfies all the constraints of $\Psi$. The latter condition is clear since the assignments $\In_{i,t}$ satisfy all of the routing constraints of $\Psi$ by definition. To check the folding constraint, fix a vertex $u$ and $i\in [2k]$. Since $A_{i,T}$ is the output of $\cR_i$ when executed on a graph with no corrupted edges we get that $A_{i,T}(u)=A_0(\pi_i(u))=A(\pi_i(u))$ for all $u$. Since $A$ is a satisfying assignment, $(A_0(u),A_{i,T}(u))$ satisfies the constraint $\Phi(u,\pi_i(u))$ as required.

\vspace{-2ex}
\paragraph{Soundness:} 
Suppose that $\val(\Psi) \geq 1-\eps$. Let $(A_0$, $\{\In_{i,t}\}_{i\in [2k], t\in [T]})$ be the assignment achieving this value. Let $\{\out_{i,t}\}, A_{i,T}$ be the deduced maps with $A_0,\In_{i,t}$ as input. 

We will show that this implies that $\val(\Phi) \geq 1-\nu$ by exhibiting that the assignment $B$ defined as $B(u)=A_0(u)$ has  high value for $\Phi$. Let $\mathcal{E} \subseteq E(G)$ be the set of edges violated by $(A_0, \{A_{i,t}\}_{i\in [2k], t\in [T]})$; we know that $\mu(\mathcal{E}) \leq \eps$. 

Fix any $i\in [2k]$. We know that for all edges $(u,v)\not\in \mathcal{E}$, any message that $u$ sends to $v$ is received correctly by $v$, in which case the tables $\In_{i,t}$ describe a correct simulation of the routing protocol initiated with the assignment $A_0$ and the set of corrupted edges $\mathcal{E}$. For every $i\in [2k]$, the tolerance guarantee of the routing protocol $\cR_i$ gives that
\[
\Pr_{u\sim [2n]}[A_0(u) \neq A_{i,T}(\pi_i(u))]\leq \nu.\]
By folding, for all $i$ and $u\in [2n]$, $(A_0(u),A_{i,T}(u))\in\Phi(u,\pi_i(u))$. Therefore letting $\viol(B)$ denote the fraction of edges violated by $B$, we get that, 
\begin{align*}
{\sf viol}(B)
&=\Pr_{i\sim [2k],u\in [2n]}[(B(u),B(\pi_i(u))) \notin \Phi(u,\pi_i(u))]\\ 
&\leq \Pr_{i,u}[(A_0(u),A_{i,T}(u))\notin \Phi(u,\pi_i(u))]+\Pr_{i,u}[A_{i,T}(u) \neq A_0(\pi_i(u))]\\
&\leq \nu,
\end{align*}
where $\Pr_{i,u}[(A_0(u),A_{i,T}(u))\notin \Phi(u,\pi_i(u))] = 0$ by folding. The conclusion now follows by the fact that $\val(\Phi)\geq 1-\nu$ implies that $\val(\Psi')\geq 1-8\nu$.
\end{proof}

\subsection{Embedding PCPs on an HDX, with amplification}\label{sec:link-routing-pcp}
In this section we show how to use the link-to-link routing protocol from \Cref{sec:link-to-link} to convert a 2-CSP $\Psi$ to a 2-CSP $\Psi'$ on a graph underlying an HDX.
The idea is similar to the idea
in the proof of Lemma~\ref{lem:routing-to-pcp-general-restated}, but since our graph $G=(X(1),X(2))$ may not be regular we cannot directly apply the lemma. As remarked earlier, by associating a vertex of the CSP to a link of $X$ though, we can handle these regularity issues, as well as get an amplification result. More precisely, we have:

\begin{lemma}\label{lem:improved-routing-to-PCP}
There exists $\eps>0$ such that for all $C>0$ and $\delta\in (0,1)$ there are constants $C',d_0,n_0\in \N$ such that the following holds for all $n\geq n_0$ and $d\geq d_0$. Let $X$ be the $d$-dimensional complex from \Cref{thm:cl} with parameters $q=\log^{C'} n$ and $\delta$, and $2n\leq |X(1)|\leq O_{\delta}(1) n$
(which can be constructed in time $\poly(n)$). Then there is a $\poly(n)$ time procedure mapping any 2-CSP $\Psi'$ with $n$ vertices on a $k$-regular graph $H$ and alphabet $\Sigma$, to a 2-CSP $\Psi$ on the graph $G = (X(1),X(2))$ with alphabet size at most $|\Sigma|^{kq^{O(d^2)}}$, satisfying the following properties:
\begin{enumerate}
\item Completeness: If $\val(\Psi')=1$ then $\val(\Psi)=1$.
\item Soundness: If $\val(\Psi')\leq 1-\frac{1}{\log^C n}$, then $\val(\Psi) \leq 1-\eps$.
\item The decision complexity of $\Psi$ 
is at most $kq^{O(d^2)}|\Sigma|$.
\end{enumerate}
\end{lemma}

\begin{proof}
We first use \Cref{thm:cl} to construct a complex $X$ in $\poly(n)$-time which is a $d$-dimensional complex with $2n\leq |X(1)| \leq O_{\delta}(1)n$ and $q=\Theta(\log^{C'}n)$ for $C'$ chosen to be large enough. As before, for $S \in X(i)$ let $\Delta_S$ be the number of $d$-faces containing $S$, i.e. $\Delta_S = |X_S(d-i)|$. Let $G=(X(1), X(2))$ and let $Z=G\zz \cH$ be the zig-zag product of $G$ (thought of as an undirected graph with multi-edges) with the family of expanders $\cH$ from \Cref{lem:exp-construction}, as described in \Cref{sec:zz}. The number of vertices of $Z$ is equal to the number of multi-edges in $G$, which is $\sum_{u\in X(1)}\Delta_{uv}=\binom{d}{2}|X(d)|$ and is denoted by $N$ throughout. 

We start by reducing to a 2-CSP $\Phi$ on a bipartite $2k$-regular graph $G'$ on $|V(Z)|$ vertices (which is at least $2n$) such that if $\Psi'$ is satisfiable then so is $\Phi$ and if $\val(\Psi')\leq 1-8\nu$ then $\val(\Phi)\leq 1-\nu$. This reduction is the same as that in \Cref{lem:routing-to-pcp-general}, and we omit the details.

Since the graph $G'$ is a bipartite $2k$-regular graph, we may partition the edge set of $G'$ into $2k$ perfect matchings, which we denote by $\pi_1,\ldots, \pi_{2k}$. Abusing notation, we think of $\pi_1,\ldots,\pi_{2k}$ also as the permutations on $V(Z)$ corresponding to the matchings. Note that since each $\pi_i$ is a perfect matching, we have that $\pi_i^2=\text{id}$. 

Let $\cR_i$ denote the protocol from \Cref{lem:link-routing} that routes the matching $\pi_i$ on $V(Z)$ with the parameter $2C$ in place of $C$. We know that $\cR_i$ has round complexity $T=O(\log n)$ and work complexity $q^{O(d^2)}\log|\Sigma|$. We now describe the CSP $\Psi$.

\vspace{-2ex}
\paragraph{The graph:} 
the graph underlying the CSP $\Psi$ is $G = (X(1),X(2))$.
\vspace{-2ex}
\paragraph{The alphabet:} 
For every $u\in X(1)$ let $T_u$ denote the number of messages that $u$ receives over all rounds. The alphabet of each vertex $u$ is a subset of $\Sigma^{O(kT_u)}$.

First recall that a vertex $j\in V(Z)$ is a tuple $(j_1,j_2)$ with $j_1\in X(1)$ and $j_2$ in the cloud of $j_1$. We think of an assignment to the CSP as a collection of maps, $A_0\colon \cS \to \Sigma$, for $\cS=\{(j,u):j\in V(Z),u\in X_{j_1}(1)\}$ and $\In_{i,t}\colon 2E(G)\to \Sigma^*$ for each $i\in [2k]$ and round $t\in [T]$. 

As a part of its alphabet, every vertex $u\in X(1)$ holds the symbols $\{A_0(j,u)\mid j\in V(Z), u\in X_{j_1}(1)\}$. Additionally it holds the symbols $\In_{i,t}[u,v]$, which denotes the message that $u$ receives from its neighbors $v$ in the protocol $\cR_{i}$ at round $t$. 

\vspace{-2ex}
\paragraph{Viewing the rules of each protocol as maps:} We define the maps $\out_{i,t}\colon 2E(G)\times \Sigma^* \rightarrow \Sigma^*\cup \{\perp\}$, for $i\in [2k], t\in [T-1]$ in the following way. The symbol $\out_{i,t}[u,v,\sigma]$ is the message that $u$ sends to $v$ in round $t$ and protocol $\cR_i$ if the messages received by it in previous rounds are given by $\sigma$, for a valid  $\sigma$ of the correct length; otherwise $\out_{i,t}[u,v,\sigma] = \perp$. We let $\out_{i,t}[u,v]$ denote the message that $u$ sends to $v$ at round $t$, on the transcript $\In_{i,\cdot}[u,\cdot]$.

Similarly define the output of the protocol $\cR_i$ as the map $A_{i,T}$ like in \Cref{lem:link-routing}. Formally, for all $i\in [2k]$, consider the map $A_{i,T}\colon \cS \times \Sigma^* \to \Sigma$ where the symbol $A_{i,T}[(j,u),\sigma]$ specifies the output of the routing algorithm at a vertex $u$ in $\cR_i$ with respect to $j\in V(Z)$ (for some $u\in X_{j_1}(1)$) when given as input the transcript $\sigma$. On the transcript $\In_{i,\cdot}[u,\cdot]$ we use $A_{i,T}(j,u)$ to denote the corresponding output symbol.

\vspace{-2ex}
\paragraph{Intuition towards defining constraints:}
In an ideal proof we want $A_0(j,u)$ to be the same on all $u\in X_{j_1}(1)$ and $A_0(j,\cdot)$ to be a satisfying assignment for $\Phi$. The maps $\In_{i,t},\out_{i,t}$ should be the transcript when the protocol $\cR_i$ is executed on $G$.

We want to ensure that the message transmission across every edge behaves as it is supposed to -- for the edge $(u,v)$ the outgoing message that $v$ sends to $u$ at any round should equal the message that $u$ receives at the next round and vice versa. 

Secondly, suppose that $A_0(j,\cdot)$ was a satisfying assignment to $\Phi$ and that no edge in the protocol is corrupted. Then, for every $j$ for which $u\in X_{j_1}(1)$, the protocol $\cR_i$ successfully transmitted the symbol $A_0(\pi_i(j),\cdot)$ from the link $X_{\pi_i(j)_1}$ to the link $X_{j_1}$, that is, $A_{i,T}(j,\cdot)$ equals $A_0(\pi_i(j),\cdot)$. In particular, we would have that $(A_0(j,u),A_{i,T}(j,u))$ satisfies the constraint $(j,\pi_i(j))$ in $\Phi$. Since this only depends on $u$, we enforce this as a hard constraint on the alphabet of $u$.

\vspace{-2ex}
\paragraph{Folding:} We constrain the label set of 
$u$ to only be tuples where $(A_0(j,u),A_{i,T}(j,u))$ satisfies the constraint $\Phi(j,\pi_i(j))$ in $\Phi$, for all $i\in [2k]$ and $j\in V(Z)$ where $u\in X_{j_1}(1)$. By that, we mean that only labels that satisfy this condition are allowed in an assignment to $\Phi$. 

\vspace{-2ex}
\paragraph{The constraints of $\Psi$:} For an edge  $(v,u)\sim X(2)$, read the labels of $u,v$ and check that,
\begin{enumerate}
\item For each $j\in V(Z)$ for which $u,v$ are both in $X_{j_1}(1)$: $A_0(j,u) = A_0(j,v)$, and for all $i\in [2k]$ it holds that $A_{i,T}(j,u) = A_{i,T}(j,v)$. 
\item For each $i\in [2k]$ and $t\in [T-1]$: $\In_{i,t+1}(u,v) = \out_{i,t}(v,u)$ and $\In_{i,t+1}(v,u) = \out_{i,t}(v,u)$.
\end{enumerate}
In words, the constraint on $v,u$ check
that they hold the same value when they are 
inside the same link at the beginning and end of the protocols, and that the message that $u$ receives from $v$ at round $t+1$ is the message that $v$ sent to it at the prior round and vice versa. 

One can check that the decision complexity of the constraints is the sum of the circuit complexity of 1) computing $\out_{i,t}(v,u)$ and $\out_{i,t}(u,v)$ over $i$ and $t$, 2) checking the routing constraints hold, 3) $(A_0(j,u),A_{i,T}(j,u))\in \Phi(j,\pi_i(j))$ and $(A_0(j,v),A_{i,T}(j,v))\in \Phi(j,\pi_i(j))$ over all $i,j$. Since the work complexity of $\cR_i$ is $q^{O(d^2)}\log|\Sigma|$ and the complexity of computing (2) is at most $kq^{O(d^2)}|\Sigma|$ the decision complexity amounts to $kq^{O(d^2)}|\Sigma|$.
This completes the description of $\Psi$, and we now analyze the completeness and the soundness of the reduction.

\vspace{-2ex}

\paragraph{Completeness:} Suppose 
that ${\sf val}(\Phi)=1$, and let $A:V(G') \rightarrow \Sigma$ be a satisfying assignment. We take $A_0(j,u)=A(j)$ for all $u\in X_{j_1}(1)$ and define the maps $\In_{i,t},\out_{i,t}$ and $A_{i,T}$ according to the execution of the routing protocols $\cR_i$ for each $i\in [2k]$ when instantiated with $A_0$. To argue that this is a valid assignment we must check that it satisfies the folding constraints; 
to check that its value is $1$, we must verify that it satisfies all the constraints of $\Psi$. 
The latter condition is clear since the assignment $A_0$ is equal on all the vertices in the link $X_{j_1}$ for all $j$, and the assignments $\In_{i,t}$ satisfy all of the routing constraints
of $\Psi$ by definition. To check the folding constraints, fix a vertex $u$, take $j\in V(Z)$ with $u\in X_{j_1}(1)$, and take $i\in [2k]$. Since $A_{i,T}$ is the output of $\cR_i$ when executed on a graph with no corrupted edges we get that $A_{i,T}(j,u)=A_0(\pi_i(j),u')=A(\pi_i(j))$ for all $u'\in X_{\pi_i(j)_1}$. Since $A$ is a satisfying assignment, $(A_0(j,u),A_{i,T}(j,u))$ satisfies the constraint $\Phi(j,\pi_i(j))$ as required.

\vspace{-2ex}
\paragraph{Soundness:} Suppose that $\val(\Psi) \geq 1-\eps$, where $\eps$ is less than $\eps_0$, the absolute constant in \Cref{lem:link-routing}. Let $(A_0$, $\{\In_{i,t}\}_{i\in [2k], t\in [T]})$ be the assignment achieving this value. Let $\{\out_{i,t}\}, A_{i,T}$ be the deduced maps with $A_0,\In_{i,t}$ as input. We will show that this implies that $\val(\Phi) \geq 1-\frac{1}{8\log^C n}$ by exhibiting a high-valued assignment $B$ for it. In fact, our assignment for $G'$ will be $B(j) = \maj_{u\in X_{j_1}(1)}(A_0(j,u))$ if a clear majority of at least $99\%$ inside $X_{j_1}$ exists, and $\perp$ otherwise. 

Let us upper bound $\viol(B)$, where we count every edge on the vertices assigned $\perp$ as violated. Let 
$\mathcal{E} \subseteq X(2)$ be the set of edges violated by $A_0, \{A_{i,t}\}_{i\in [2k], t\in [T]}$; we know that $\mu(\mathcal{E}) \leq \eps$. 

We first upper bound the probability that $B(j)=\perp$. Suppose that $j\in V(Z)$ is such that, $\mu_{j_1}(\mathcal{E})\leq 0.05$. Then there exists $\sigma_j\in \Sigma$ such that $A_0(j,v)=\sigma_j$ for at least $0.99$-fraction of $v \sim X_{j_1}$ (since the spectral gap of the graph $(X_u(1),X_u(2))$ is at least $1/2$). This in turn implies that $B(j)\neq \perp$. Therefore, 
\begin{equation*}
\Pr_{j \sim V(Z)}[B(j)=\perp]\leq \Pr_{j\sim V(Z)}[\mu_{j_1}(\mathcal{E}) \geq 0.05]=\Pr_{u\sim X(1)}[\mu_{u}(E)\geq 0.05] \leq \frac{1}{\log^{2C} n},    
\end{equation*}
where we used \Cref{fact:distribution-zz} in the second transition and \Cref{claim:bad-1-link} in the last transition.

For convenience of analysis, we also define $B(i,\pi_i(j),T)$ as $\maj_{u\in X_{j_1}(1)}(A_{i,T}(\pi_i(j),u))$ if at least $99\%$ of the vertices in $X_{j_1}(1)$ have the same value in $A_{i,T}$, and $\perp$ otherwise. As is the case above, for $B(i,j,T)$ to be $\perp$ it holds that at least $0.05$-fraction of the edges in $X_{\pi_i(j)_1}$ are in $\mathcal{E}$ too, so for all $i\in [2k]$,
\[\Pr_{j \sim V(Z)}[B(i,j,T)=\perp] \leq \frac{1}{\log^{2C} n}.\]
Fix any $i\in [2k]$. We know that for all edges $(u,v)\not\in \mathcal{E}$, any message that $u$ sends to $v$ is received correctly by $v$, in which case the tables 
$\In_{i,t}$ describe a correct simulation of the routing protocol in~\Cref{lem:link-routing} initiated with the assignment $A_0$. Provided that $\eps$ is small enough, Lemma~\ref{lem:link-routing} tells us 
that
\[
\Pr_{j\sim V(Z)}[B(j) \neq B(i,\pi_i(j),T)]\leq \Pr_{j\sim V(Z)}[B(j)=\perp]+\frac{1}{\log^{2C} N}\leq \frac{2}{\log^{2C} n}.\]
By folding, for all $i,j$ and $u\in X_{j_1}(1)$, $(A_0(j,u),A_{i,T}(j,u))\in\Phi(j,\pi_i(j))$, therefore we get that
\[\Pr_{i,j}[(B(j),B(i,j,T))\notin \Phi(j,\pi_i(j))] \leq \Pr_{i,j}[B(j)\neq \perp]+\Pr_{i,j}[B(i,j,T)\neq \perp]\leq \frac{2}{\log^{2C} n}.\]
Using the union bound we conclude that
\begin{align*}
{\sf viol}(B)
&=\Pr_{i\sim [2k],j\in [N]}[(B(j),B(\pi_i(j))) \notin \Phi(j,\pi_i(j))]\\ 
&\leq \Pr_{i,j}[(B(j),B(i,j,T))\notin \Phi(j,\pi_i(j))]+\Pr_{i,j}[B(i,j,T) \neq B(\pi_i(j))]\\
&\leq \frac{4}{\log^{2C} n},
\end{align*}
which implies that $\val(\Phi) \geq 1-\frac{1}{8\log^C n}$. This implies that $\val(\Psi')\geq 1-\frac{1}{\log^C n}$ as required.
\end{proof}

\section{Amplification of 2-CSPs on Complexes with Direct Product Testers}\label{sec:amplification}
In this section, we show how to amplify the soundness of a 2-CSP on $G=(X(1),X(2))$ where $X$ is an HDX that supports a direct product test. If $X$ is a sparse complex, then this result is a derandomized parallel repetition procedure for 2-CSPs on $G$.
In the time of writing this paper we only know of one family of sparse complexes with this property: the Chapman-Lubotzky complexes from \Cref{thm:cl}. Thus, in~\Cref{sec:final} we instantiate 
this idea with these complexes.

\subsection{Gap Amplification to Low Soundness}
Fix any complex $X$ for which the $(k,\sqrt{k})$-direct-product tester on $X$ has soundness $\delta$, and consider any 2-CSP $\Psi$ on $G=(X(1),X(2))$. Our reduction will produce a Label Cover instance $\Phi$ on the bipartite inclusion graph $G'=(X(k),X(\sqrt{k}))$ with left alphabet $\Sigma^k$ and right alphabet $\Sigma^{\sqrt{k}}$. The constraint on $(U,V)$ check whether the label on $U \in X(k)$ satisfies all the constraints in $G$ (since for all $u,v\in U$, $(u,v)\in X(2)$) and further if projected to $B$ it equals the label given to $B$. 

As we did in \Cref{def:gen-2-csp}, to simplify the presentation of the proof, we define a generalized version of the Label Cover problem from \Cref{def:label-cover}. In particular we allow for varying alphabets $\Sigma_L(U) \subseteq \Sigma_L$ to the left-side of vertices $U \in L$. This helps us to restrict the prover to provide a label to $U$ that satisfies additional constraints (which in our case would be that the label in $\Sigma^k$ given to $U$ satisfies all the constraints inside $U$) which makes our soundness analysis cleaner to carry out. 
\begin{definition}\label{def:label-cover-generalized}
An instance $\Phi = (G=(L\cup R,E,w), \Sigma_L, \Sigma_R, \{\Sigma_L(U)\}_{U\in L}, \{\Phi_{e}\}_{e\in E})$ of generalized label cover consists of a weighted bipartite 
graph $G$, alphabets $\Sigma_L, \Sigma_R, \{\Sigma_L(U)\}$ with $\Sigma_L(U)\subseteq \Sigma_L$ for all $U\in L$, and constraints 
$\Phi_e\subseteq \Sigma_L\times \Sigma_R$, one for each edge. Each one of the constraints is a projection constraint, meaning
that for every $e=(U,V)\in E$ there is a map $\phi_e\colon \Sigma_L(U)\to\Sigma_R$ such that
\[
\Phi_e = \{(\sigma,\phi_e(\sigma))~|~\sigma\in \Sigma_L\}.
\]
\end{definition} 
We remark that a hardness result for generalized label cover can be  easily converted to a hardness result for the standard definition. When the alphabet size is super-constant though, one needs to be careful so as to preserve the decision complexity of the constraints while performing this translation. Therefore, since our alphabet size is large, in all the intermediate step in our reductions from now on we use the generalized Label Cover problem. After performing alphabet reduction to get a generalized Label Cover instance with constant-sized alphabet in \Cref{sec:final}, we use this simple translation to go back to the standard label cover problem. 

The main result of this section 
is the following statement:
\begin{lemma}\label{lem:soundness-amp}
For all $\delta>0$ there exists $k\in\mathbb{N}$ such that the following holds. Suppose that $X$ is a complex for which the $(k,\sqrt{k})$-direct product tester has soundness $\delta^2$. Then there is a polynomial time procedure such that given a generalized 2-CSP instance $\Psi$ over the weighted graph $G=(X(1), X(2))$ with alphabets $\Sigma,\{\Sigma(u)\}_{u\in G}$ and decision complexity $D$, produces an instance of generalized Label Cover $\Phi$ over the weighted inclusion graph $(X(\sqrt{k}), X(k))$ over left alphabet $\Sigma^k$ and right alphabet $\Sigma^{\sqrt{k}}$ such that:
\begin{enumerate}
\item The projection map $\phi_{(A,B)}$ associated to the edge $(A,B)$ is defined as the restriction of the assignment to $A$ to the coordinates in $B$. That is, $\forall \sigma\in \Sigma_L(A), \phi_{(A,B)}(\sigma)=\sigma|_B$.
\item For all $A \in X(k)$, the circuit complexity of checking membership in $\Sigma_L(A)$ is $O(k^2D)$.
\item If ${\sf val}(\Phi) = 1$, then ${\sf val}(\Psi) = 1$.
\item If ${\sf val}(\Phi)\leq 1-4\delta$, then ${\sf val}(\Psi)\leq \delta$.
\end{enumerate}    
\end{lemma}

\begin{proof}
Our label cover instance $\Phi$ has vertices $L = X(k)$, $R = X(\sqrt{k})$ and edges between them given by inclusion. Letting $\Sigma$ be the alphabet of $\Psi$, we take $\Sigma_L=\Sigma^k$ to be the alphabet of the left side of $\Phi$ and $\Sigma^{\sqrt{k}}$ to be the alphabet of the right side of $\Phi$. For every vertex $A = (a_1,\ldots,a_k) \in U$ let $\Sigma_L(A)$ be the set of assignments $(\sigma_1,\ldots,\sigma_k) \in \Sigma^k$ where for every $i$, $\sigma_i\in \Sigma(a_i)$ and for every $i \neq j$, $(\sigma_i,\sigma_j)$ satisfies the constraint $\Psi_{(a_i,a_j)}$ on the edge $(a_i,a_j)\in G$. The decision complexity of membership in $\Sigma_L(A)$ is easily seen to be $O(k^2 D)$. The constraints $\Phi_e$ are as defined as in the lemma statement.

The completeness of the reduction is clear, and we move on to the soundness analysis. Suppose 
that $\val(\Phi) \geq \delta$, and 
fix assignments $F:X(k) \rightarrow \Sigma^k$ and $G: X(\sqrt{k}) \rightarrow \Sigma^{\sqrt{k}}$ realizing $\val(\Phi)$, where $F(A) \in \Sigma_L(A)$ for all $A \in X(k)$. Thus,
\[\val(\Phi) = \Pr_{\substack{B \sim X(\sqrt{k})\\A \supset_k B}}[F[A]|_B = G[B]] \geq \delta.\]
Using Cauchy-Schwarz we conclude that
\begin{align*}
\Pr_{\substack{D \sim X(d) \\ B \sim X(\sqrt{k})\\ B \subset A,A' \subset D}}\left[F[A]|_B = F[A']|_B\right] &\geq \E_{\substack{D \sim X(d) \\ B \sim X(\sqrt{k})}}\left[\E_{B \subset A,A' \subset D}[\Ind(F[A]|_B=G[B])\Ind(F[A']|_B=G[B])]\right] \\
&=\E_{\substack{D \sim X(d) \\ B \sim X(\sqrt{k})}}\left[\E_{B \subset A \subset D}[\Ind(F[A]|_B=G[B])]^2\right] \\
&\geq \left(\E_{\substack{D \sim X(d) \\ B \sim X(\sqrt{k})}}\left[\E_{B \subset A \subset D}[\Ind(F[A]|_B=G[B])]\right]\right)^2 \\
&= 
\left(\E_{\substack{D \sim X(d) \\ A \subset_k D \\ B \subset_{\sqrt{k}} A}}[\Ind(F[A]|_B=G[B])]]\right)^2 \\
&\geq \delta^2.
\end{align*} 
This implies that $F$ passes the direct product test and therefore using the soundness of the test we get a function $f:X(1) \rightarrow \Sigma$ such that,
\[\Pr_{A\sim X(k)}\left[\Delta(F[A],f|_{A})\leq \delta \right]\geq \poly(\delta). 
\]
Let $\cB \subseteq X(2)$ be the set of constraints that $f$ violates. By construction $F[A]$ satisfies all the constraints inside $A$, therefore wherever it holds that $\Delta(F[A],f|_{A})\leq \delta k$ we get that $f$ satisfies at least $(1-\delta)^2 \geq 1-2\delta$ fraction of the constraints inside $A$.
In particular, we conclude that
\begin{equation}\label{eq:viol-f}
\Pr_{A\sim X(k)}\left[\mu(\cB|_A) \leq 2\delta\right]\geq \poly(\delta).    
\end{equation}
Suppose for the sake of contradiction that $\mu(\cB)>4\delta$. Applying \Cref{lem:sampling} we get,
\[\Pr_{A \sim X(k)}\left[\mu(\cB|_K) \leq 2\delta\right] \lll \frac{1}{k\delta},\]
since the bipartite graph $(X(2), X(k))$ has second largest eigenvalue at most $O(1/\sqrt{k})$ by \Cref{lem:spectral_gap_of_graphs_from_HDX}. Since $k$ is chosen to be large enough as a function of $\delta$, in particular at least $\frac{1}{\poly(\delta)}$, this is a contradiction to \eqref{eq:viol-f}. Thus we get that $\mu(\cB) \leq 4\delta$, which in turn means that $\val(\Psi) \geq 1-4\delta$.
\end{proof}

\section{Alphabet Reduction via Decodable PCPs}
In this section we discuss the construction of PCPs for Circuit-SAT with small alphabet but large size (polynomial, or even exponential). 
The tools presented in the paper so far 
lead to size efficient PCPs with large alphabets, and our 
goal here is to facilitate the use the efficient composition
theorems of~\cite{MoshkovitzR08,DinurH13} to reduce the alphabet size.

To apply the abstract composition theorem of~\cite{DinurH13} we require 
PCP constructions in which one has a ``decodable verifier''. By that, we mean that the PCP verifier not only probabilistically checks whether a proof of satisfiability is correct or not, but it is also able to decode a symbol of the satisfying assignment with high probability. We present the formal definition in \Cref{sec:prelim-dpcp-defn}. These constructions will be used as inner PCPs 
in our composition.

We remark that so far in the paper we discussed 
$2$-query PCPs using the framekwork as label 
cover, and in~\cite{DinurH13} the proof composition is presented in the language of ``robust PCPs''. The language of robust PCPs can be seen to be an equivalent formulation of label cover, but it is easier to
use in the context of composition. Thus, 
for convenience we carry out most of the argument in the language of robust PCPs, formally defined in Section~\ref{sec:prelim-robust-pcp}. The material presented in Sections~\ref{sec:prelim-robust-pcp} and~\ref{sec:prelim-dpcp-defn} is almost verbatim repeat of~\cite{DinurH13}, but we give it here for the sake of completeness.

\subsection{Robust PCPs}\label{sec:prelim-robust-pcp}
We now discuss the notion of robust PCPs, 
which will be the outer PCPs in our composition.
First defined in~\cite{Ben-SassonGHSV06, DinurReingold}, robust PCPs have been implicit in all PCP constructions. The only difference between robust PCPs and standard PCPs is in the soundness condition: while the standard soundness condition measures how often the PCP verifier accepts a false proof, the robust soundness condition measures the average distance between the local view of the verifier and an accepting local view. The definition given below is from~\cite{DinurH13}:
\begin{definition}[Robust PCPs]
For functions $r, q, m, a, s : \mathbb{Z}^+ \rightarrow \mathbb{Z}^+$ and $\delta : \mathbb{Z}^+ \rightarrow [0,1]$, a verifier $V$ is a robust probabilistically checkable proof (robust PCP) system for a language $L$ with randomness complexity $r$, query complexity $q$, proof length $m$, alphabet size $a$, decision complexity $s$ and robust soundness error $\delta$ if $V$ is a probabilistic polynomial-time algorithm that behaves as follows: On input $x$ of length $n$ and oracle access to a proof string $\pi\in \Sigma^{m(n)}$ over the (proof) alphabet $\Sigma$ where $|\Sigma| = a(n)$, $V$ reads the input $x$, tosses at most $r(n)$ random coins, and generates a sequence of locations $I = (i_1, \dots, i_q) \in [m]^{q(n)}$ and a predicate $f : \Sigma^q \rightarrow \{0,1\}$ of decision complexity $s(n)$, which satisfies the following properties:
\textbf{Completeness:} If $x \in L$ then there exists $\pi$ such that
\[
\Pr_{(I,f)}[f(\pi_I) = 1] = 1.
\]
\textbf{Robust Soundness:} If $x \notin L$ then for every $\pi$
\[
\E_{(I,f)}[\agr(\pi_I, f^{-1}(1))] \leq \delta,
\]
where the distribution over $(I,f)$ is determined by $x$ and the random coins of $V$. 
\end{definition}
Next we define the notion of proof degree and regularity for a robust PCP. 
\begin{definition}
Given a robust PCP system, we will refer to the maximum number of local windows any index in the proof participates in, as the \emph{proof degree}, denoted by \(d(n)\). More precisely, for each \(i \in [m(n)]\), if we let
\[
R_i = \{ r \in \{0, 1\}^{r(n)} \mid i \in I(r) \},
\]
then \(d(n) = \max_i |R_i|\). Furthermore, if \(|R_i| = d(n)\) for all \(i\), we will say the PCP system is regular.    
\end{definition}

\paragraph{Equivalence of Label Cover and Robust PCPs:} the notion of robust PCP is in fact equivalent to generalized label cover (\Cref{def:label-cover-generalized}) as shown in~\cite[Lemma 2.5]{DinurH13}, and we now give some intuition for this equivalence. If a language $L$ has a robust PCP, then here is a reduction from $L$ to generalized Label Cover: the set of left vertices is the set of random strings of the robust PCP, the set of right vertices is the set of the proof locations. An edge $(r,i)$ exists if the proof location $i$ is probed on random string $r$. The label to a left vertex $r$ is an accepting local view of the verifier on random string $r$ while a label to the right vertex $i$ is the proof symbol in the corresponding proof location $i$. An edge $(r,i)$ is consistent if the local view is consistent with the proof symbol. Conversely, given a reduction from $L$ to generalized label cover, we can get a robust PCP verifier for $L$ as follows: the verifier expects as proof a labeling of the set of right vertices, the verifier chooses a random left vertex, queries all its neighbors and accepts if there exists a label to the left vertex that satisfies all the corresponding edges.
We summarize this discussion with the following lemma (see~\cite{DinurH13} for a formal proof): 
\begin{lemma}\label{lem:lc-robust-equivalence}
For every $\delta : \mathbb{Z}^+ \rightarrow \mathbb{R}^+$, and $r, q, m, a : \mathbb{Z}^+ \rightarrow \mathbb{Z}^+$, the following two statements are equivalent:
\begin{enumerate}
    \item Gap-Generalized-Label-Cover$[1,\delta]$ is $\mathbf{NP}$-hard for instances with the following parameters:
    \begin{itemize}
        \item left degree at most $q(n)$,
        \item right degree at most $d(n)$
        \item right alphabet $\Sigma(n)$ with $|\Sigma| = a(n)$,
        \item left alphabet $\{\Sigma_L(U)\}_{U\in L}$,
        \item size of right vertex set at most $m(n)$, and
        \item size of left vertex set at most $2^{r(n)}$.
    \end{itemize}
    \item Every $L \in \mathbf{NP}$ has a robust PCP with completeness $1$, robust soundness error $\delta$ and the following parameters:
    \begin{itemize}
        \item query complexity $q(n)$,
        \item proof degree at most $d(n)$
        \item proof alphabet $\Sigma(n)$ with $|\Sigma| = a(n)$,
        \item maximum number of accepting local views $\max_{U\in L}(|\Sigma_L(U)|)$,
        \item proof length $m(n)$, and
        \item randomness complexity $r(n)$
    \end{itemize}
\end{enumerate}
Furthermore, suppose that $\Sigma_L=\Sigma^k$ and $\Sigma_R=\Sigma^t$ for some alphabet $\Sigma$ and $k,t\in \N$, all the constraints $\phi_{(u,v)}$ of the Label Cover instance check if the label of $u$ restricted to $v$ is equal to the label of $v$, and the circuit complexity of checking membership in the language $\Sigma_L(U)$ is at most $D$, then the decision complexity of the robust PCP is $O(D+q(n))$. 
\end{lemma}
It is important to note that this is a syntactic correspondence between the notions of generalized Label-Cover and robust PCPs and there is no loss of parameters in going from one framework to another. In particular, going from label cover to a robust PCP and back, one gets back the original label cover instance. Even though these two notions are syntactically equivalent, some results are easier to state/prove in one framework than the other. 
In Section~\ref{sec:amplification} 
we proved a hardness of generalized label cover with large alphabet, but applying alphabet reduction will be easier to carry out in the robust PCP framework. 

\subsection{Decodable PCP}\label{sec:prelim-dpcp-defn}
We  now describe the notion of a decodable PCP (dPCP) from~\cite{DinurH13}, which will serve as our inner PCP in the composition. It is sufficient to define dPCPs for the
problem $\csat_{\Sigma}$ for our purposes, and as such we focus
the discussion on it. The problem $\csat_{\Sigma}$ is concerned with circuits $C$ whose input is
a string from $\Sigma^n$. It will often be
more convenient for us to think of circuits
over large alphabet as the equivalent Boolean circuit $\tilde{C}\colon \{0,1\}^{n\log |\Sigma|}\to\{0,1\}$ in which each input wire of $C$ is split into $\log |\Sigma|$ wires in $\tilde{C}$ in the obvious way. With this in mind, we define
the circuit size of $C$ to be the size of 
$\tilde{C}$, and define the $\csat_{\Sigma}(N,S)$ problem in the following way:
\begin{definition}[Circuit-SAT]
An instance of $\csat_{\Sigma}(N,S)$ is a circuit $C:\Sigma^{N}\rightarrow \{0,1\}$ of size at most $S$. The goal is to decide whether there exists an
input $x\in\Sigma^N$ such that $C(x) = 1$.
\end{definition}

Given an instance $C$ of $\csat$, a probabilistically checkable proof for $C\in \csat$ often takes a string $y$ such that $C(y) = 1$ and encodes it using a probabilistically checkable proof. We refer to such a $y$ as an NP-witness of the fact that $C \in \csat$. 

A standard PCP verifier for the language $\csat$ would verify that the input circuit is satisfiable, with the help of a PCP, which is typically (but not necessarily) an encoding of the NP-witness $y$. 
A PCP decoder for CircuitSAT is a stronger notion. Just like a PCP verifier, it expects the PCP to be an encoding of the NP witness.
However, in addition to that, after performing its local check, a PCP decoder is expected to decode back a location in the NP witness. 

\begin{definition}[PCP Decoders]\label{def:pcp-decoders}
A PCP decoder for $\csat_{\Sigma}$ over a proof alphabet $\sigma$ is a probabilistic polynomial-time algorithm $D$ that on input a circuit $C: \Sigma^k \rightarrow \{0,1\}$ of size $n$ and an index $j \in [k]$, tosses $r = r(n)$ random coins and generates (1) a sequence of $q = q(n)$ locations $I = (i_1, \dots, i_q)$ in a proof of length $m(n)$ over the alphabet $\sigma$ and (2) a (local decoding) function $f : \sigma^q \rightarrow \Sigma \cup \{\bot\}$ whose corresponding circuit has size at most $s(n)$, referred to henceforth as the decision complexity of the decoder.
\end{definition}

With this in mind we can now define decodable PCPs, where a verifier either rejects a proof, or decodes a symbol that belongs to a small list of satisfying assignments for the CircuitSAT instance.
\begin{definition}\label{dPCPs}[Decodable PCPs] 
For functions $\delta : \mathbb{Z}^+ \rightarrow [0,1]$ and $L: \mathbb{Z}^+ \rightarrow \mathbb{Z}^+$, we say that a PCP decoder $D$ is a decodable probabilistically checkable proof (dPCP) system for $\text{CircuitSAT}_\Sigma$ with soundness error $\delta$ and list size $L$ if the following completeness and soundness properties hold for every circuit $C : \Sigma^k \rightarrow \{0,1\}$:
\begin{itemize}
    \item \textbf{Completeness:} For any $y \in \Sigma^k$ such that $C(y) = 1$ there exists a proof $\pi \in \sigma^m$, also called a decodable PCP, such that
    \[
    \Pr_{j,I,f} [f(\pi_I) = y_j] = 1,
    \]
    where $j \in [k]$ is chosen uniformly at random and $I, f$ are distributed according to $C_j$ and the verifier’s random coins.

    \item \textbf{Soundness:} For any $\pi \in \sigma^m$, there is a list of $0 \leq \ell \leq L$ strings $y^1, \dots, y^\ell$ satisfying $C(y^i) = 1$ for all $i$, and furthermore that
    \[
    \Pr_{j,I,f} [f(\pi_I) \notin \{\bot, y_j^1, \dots, y_j^\ell\}] \leq \delta.
    \]

    \item \textbf{Robust Soundness:} We say that $D$ is a robust dPCP system for $\text{CircuitSAT}_\Sigma$ with robust soundness error $\delta$, if the soundness criterion in can be strengthened to the following robust soundness criterion,
    \[
    \mathbb{E}_{j,I,f} [\operatorname{agr}(\pi_I, \operatorname{BAD}(f))] \leq \delta,
    \]
    where $\operatorname{BAD}(f) := \{ w \in \sigma^q | f(w) \notin \{\bot, y_j^1, \dots, y_j^\ell\} \}$.
\end{itemize}
\end{definition}

\subsection{Constructions of Decodable PCPs from Reed-Muller and Hadamard Codes}\label{sec:dpcp-statement}
In this section we discuss two well-known constructions of decodable PCPs. These 
constructions are based
on classical primitives in PCP literature, 
and we include them in full details for 
the sake of completeness.

First, we have the following construction 
of dPCPs based on Hadamard codes.
\begin{lemma}\label{lem:had-dpcp}
For all $\delta >0$, for $q = 1/\delta^{O(1)}$ and for all alphabets $\Sigma$, the language $\csat_{\Sigma}(N,S)$ has a regular decodable PCP with the following parameters:
\begin{enumerate}
    \item Robust soundness error $\delta$.
    \item Proof alphabet size $q$. %
    \item Proof length $q^{O(S^2)}$.
    \item Randomness complexity $O(S^2\log(q))$.
    \item Query complexity and decision complexity $q^{O(\log|\Sigma|)}$.
    \item List size $1/\delta^{O(1)}$.
\end{enumerate}
\end{lemma}
\begin{proof}
Deferred to~\Cref{sec:had-dpcp}.
\end{proof}

Second, we have the following construction 
of dPCPs based on Reed-Muller codes.
\begin{lemma}\label{lem:rm-dpcp}
For all $\delta >0$ and all alphabets $\Sigma$, $\csat_{\Sigma}(N,S)$ has a regular decodable PCP with the following parameters: 
\begin{enumerate}
    \item Robust soundness error $\delta$.
    \item Proof alphabet size and proof length 
    at most $S^{O(1)}$.
    \item Randomness complexity at most $O(\log S)$.
    \item Query and decision complexity at most
    $(\log(S))^{O(\log|\Sigma|)}$.
    \item List size at most $1/\delta^{O(1)}$.
\end{enumerate} 
\end{lemma}
\begin{proof}
Deferred to~\Cref{sec:rm-dpcp}.
\end{proof}

\section{The Final PCP: Putting it All Together}\label{sec:final}
In this section we combine all the components from the previous sections to get a 2-query PCP of quasi-linear size, constant alphabet and small soundness, thereby proving Theorem~\ref{thm:main}. We begin by presenting a few tools from~\cite{Dinur07,DinurH13} that are necessary for us, namely their regularization and alphabet reduction lemmas and their composition theorem.

\subsection{Regularization Procedures for PCPs}

First we state the following lemma~\cite[Lemma 4.2]{Dinur07} to convert an arbitrary constraint graph to a 2-CSP on a regular graph with constant degree.
\begin{lemma}\label{lem:regular-dinur}
There exist constants $c,k \in \N$ and a polynomial time procedure that when given as input a 2-CSP instance $\Psi$ over a constraint graph $G'$ with $|V(G')|+|E(G')|=n$ over alphabet $\Sigma$, outputs a 2-CSP $\Psi'$ over a constraint graph $G'$ with $|V(G')|\leq 2|E(G)|$ and $|E(G')|=\Theta(kn)$ over alphabet $\Sigma$ such that,
\begin{itemize}
\item $G$ is $k$-regular.
\item If $\val(\Psi)=1$ then $\val(\Psi')=1$.
\item If $\val(\Psi)=1-\rho$ then $\val(\Psi')\leq 1- \rho/c$. 
\end{itemize}
\end{lemma}

Next we state a similar procedure that converts a robust PCP into a robust PCP that is also regular. Additionally it also reduces the alphabet of a robust PCP.
\begin{lemma}\label{lem:regular-pcp}
There exists a constant 
$C > 0$ such that for all 
$\eps: \mathbb{Z}^+ \rightarrow [0,1]$, the following holds. Suppose 
$L$ has a robust PCP verifier $V$ with randomness complexity $r$, query complexity $q$, proof length $m$, average proof degree $d$, robust soundness error $\delta$ over a proof alphabet $\Sigma$. Then $L$ has a regular reduced robust PCP verifier, which we shall denote by 
$\text{regular}_{\eps}(V)$ with:
\begin{itemize}
    \item randomness complexity 
    $\log m + \log d$,
    \item query complexity 
    $Cq \log |\Sigma|^{1/4}$,
    \item proof length 
    $Cq^2 2^r \log |\Sigma|^{1/10}$,
    \item proof degree $C/\eps^4$,
    \item proof alphabet of size at most $C/\eps^6$,
    \item and robust soundness error 
    $\delta + \eps$.
\end{itemize}
\end{lemma}

\subsection{PCP Composition}
We need the following efficient and abstract
composition theorem due to~\cite{DinurH13}:
\begin{theorem}\label{thm:composition}
For all $\eps >0$ the following holds. Suppose $3$SAT has a regular robust PCP verifier $V$ with robust soundness error $\Delta$, proof alphabet $\Sigma$, query complexity $Q$, decision complexity $S(n)$ and suppose $\text{CircuitSAT}_\Sigma(Q,S(n))$ has a robust PCP decoder $\cD$ with proof alphabet $\sigma$, robust soundness error $\delta$ and list size $\ell$. Then, $3$SAT has a robust PCP verifier $V' = V \circledast \cD$, with query complexity $O(q/\eps^4)$, robust soundness error $\Delta \ell + 4\ell\eps + \delta$ and other parameters as stated in \Cref{table-comp}. Furthermore, if the PCP decoder $\cD$ is regular, then so is the composed verifier $V'$.
\end{theorem}

\begin{table}[H]
\centering
\caption{Parameters for Efficient Composition.}\label{table-comp}
\begin{tabular}{|c|c|c|c|}
\hline & $V$ & $\cD$ & $V' = V \circledast \cD$ \\
\hline
Proof Alphabet & $\Sigma$ & $\sigma$ & $\sigma$ \\
Randomness Complexity & $R$ & $r$ & $\log M + r + \log D$ \\
Query Complexity & $Q$ & $q$ & $\frac{4}{\eps^4} \cdot q$ \\
Decision Complexity & $S$ & $s$ & $4s/\eps^4 + D\log\sigma$ \\
Proof Degree & $D$ & $d$ & $d$ \\
Proof Length & $M$ & $m$ & $2^R \cdot m$ \\
Robust Soundness Error & $\Delta$ & $\delta$ & $\Delta \ell + 4\ell\eps + \delta$ \\
List Size & - & $\ell$ & - \\
Input Size & $n$ & $S(n)$ & $n$ \\
\hline
\end{tabular}
\end{table}
Note that all the parameters (for $V$) with capitalized letters are functions of $n$ and the parameters (for $\cD$) with uncapitalized letters are functions of $S(n)$. The parameters of the composed PCP should be read accordingly.

\subsection{Proof of Theorem~\ref{thm:main}}
We start from a known size efficient PCP construction; either 
the construction of~\cite{BensasonSudan} that has soundness $1-1/\pl n$ or or the construction of~\cite{Dinur07} that has soundness $1-\Omega(1)$, will do. 
For a graph $G=(V,E)$, let $\size(G) = |V|+|E|$. Below we state the result of~\cite{BensasonSudan} in its more convenient formulation in terms of hardness of 2-CSPs; this formulation  can be found in~\cite[Lemma 8.3]{Dinur07}.
\begin{theorem}\label{thm:bs}
There exist constants $c_1, c_2 > 0$ such that there is a polynomial time reduction mapping a $3$SAT instance $\varphi$ of size $n$ to a 2-CSP instance $\Psi$ over the graph $G = (V, E)$ and alphabet $\Sigma$ where
\begin{itemize}
    \item We have $\size(G) \leq n(\log n)^{c_1}$ and $|\Sigma| = O(1)$.
    \item If $\varphi$ is satisfiable, then $\val(\Psi) = 1$.
    \item If $\varphi$ is not satisfiable, then $\val(\Psi) \leq 1-\frac{1}{(\log n)^{c_2}}$.
\end{itemize}
\end{theorem}

The work of~\cite{Dinur07} showed how to get to constant soundness while maintaining quasi-linear size. Again we state her result in the more convenient 2-CSP formulation.
\begin{theorem}\label{thm:dinur}
There exist constants $c_1, c_2, c_3 > 0$ such that there is a polynomial time reduction mapping a $3$SAT instance $\varphi$ of size $n$ to a 2-CSP instance $\Psi$ over the graph $G = (V, E)$ and alphabet $\Sigma$ where,
\begin{itemize}
    \item $\size(G) \leq n(\log n)^{c_1}$ and $|\Sigma| = c_2$.
    \item If $\varphi$ is satisfiable, then $\val(\Psi) = 1$.
    \item If $\varphi$ is not satisfiable, then $\val(\Psi) \leq 1-c_3$.
\end{itemize}
\end{theorem}

Using Dinur's PCP in conjunction with \Cref{lem:improved-routing-to-PCP}, we get a 2-CSP instance with soundness $1-\Omega(1)$ whose constraint graph is the base graph of the complex from Theorem~\ref{thm:cl}.

\begin{lemma}[\Cref{thm:CSP_on_CL_large_soundness2} restated]\label{lem:CSP_on_CL}
There exists $\eps>0$ such that for all $\delta\in (0,1)$ there exist constants $C,C'>0$ so that the following holds for all sufficiently large integers $d$ and $n$. Let $\{X_{n'}\}_{n'\in N}$ be the infinite sequence of complexes from \Cref{thm:cl}, where every $X_{n'}$ is a $d$-dimensional complex on $n'$ vertices with parameters $q=\Theta(\log^{C'} n')$ and $\delta$, that is constructible in time $\poly(n')$. Then there is a polynomial time reduction mapping any $3$SAT instance $\varphi$ of size $n$ to a 2-CSP $\Psi$ over the graph $G=(X_{n'}(1),X_{n'}(2))$, for some complex $X_{n'}$ from the family, such that:
\begin{enumerate}
\item We have that $n'\leq n\log^C n$, the alphabet $\Sigma$ satisfies that $\log(|\Sigma|)\leq q^{Cd^2}$, and the decision complexity of the constraints of $\Psi$ is at most $q^{Cd^2}$.
\item If $\varphi$ is satisfiable, then $\Psi$ is satisfiable.
\item If $\varphi$ is unsatisfiable, then ${\sf val}(\Psi)\leq 1-\eps$.
\end{enumerate}
\end{lemma}
\begin{proof}
Applying Dinur's reduction from \Cref{thm:dinur} to $\varphi$ and then applying the regularization procedure in \Cref{lem:regular-dinur}, in $\poly(n)$ time we get a 2-CSP $\Psi'$ whose constraint graph $G'$ is $k$-regular for an absolute constant $k$, with $|V(G')|\leq n\log^{O(1)} n$, and alphabet size $|\Sigma'|=O(1)$. We have that $\val(\Psi')=1$ if $\val(\varphi)=1$ and $\val(\Psi')=1-\eps'$ if $\val(\varphi)<1$, for some universal constants $\eps'\in (0,1)$.

We will now apply the polynomial time reduction in \Cref{lem:improved-routing-to-PCP} to $\Psi'$. This gives us a 2-CSP $\Psi$ on the constraint graph $G=(X_{n'}(1),X_{n'}(2))$, where $X_{n'}$ is a $d$-dimensional complex with $|V(G')|\leq n' \leq O_{\delta}(1)|V(G')|$ and parameters $q=\Theta(\log^{C'}n')$ for some large enough constant $C'$ and $\delta$. The alphabet size of $\Psi$ satisfies $\log|\Sigma|=kq^{O(d^2)}\log|\Sigma'| = q^{O(d^2)}$ and the decision complexity is $kq^{O(d^2)}|\Sigma|=q^{O(d^2)}$. If $\Psi'$ is satisfiable then so is $\Psi$, and if $\val(\Psi')\leq 1-\eps'$ then $\val(\Psi)\leq 1-\eps$ for some absolute constant $\eps >0$, as required.\footnote{One can check that the proof above works even if we apply the result of~\cite{BensasonSudan}, \Cref{thm:bs}, instead of \Cref{thm:dinur}, since \Cref{lem:improved-routing-to-PCP} only requires that $\val(\Psi')\leq 1-\frac{1}{(\log n)^c}$ to get the desired conclusion.}
\end{proof}

Now that we have a constant soundness PCP on the graphs underlying the complexes from \Cref{thm:cl}, we can apply the gap amplification procedure from \Cref{lem:soundness-amp} to get a 2-CSP with small soundness (but large alphabet size). This uses the fact that these complexes support a direct product test with small soundness. 

\begin{lemma}\label{lem:low-soundness-final}
For all $\delta\in (0,1)$ there exist constants $C,C'>0$ so that the following holds for all sufficiently large integers $k, d$ and $n$. Let $\{X_{n'}\}_{n'\in N}$ be the infinite sequence of complexes from \Cref{thm:cl}, where every $X_{n'}$ is a $d$-dimensional complex on $n'$ vertices with parameters $q=\Theta(\log^{C'} n')$ and $\delta$ that is constructible in time $\poly(n')$. There is a polynomial time reduction mapping a $3$SAT instance $\varphi$ of size $n$ to a generalized label cover instance $\Psi$ over the weighted inclusion graph $(X_{n'}(k), X_{n'}(\sqrt{k}))$ for some $n'\leq n\log^C n$, such that,
\begin{enumerate}
\item If $\varphi$ is satisfiable, then $\Psi$ is satisfiable.
\item If $\varphi$ is unsatisfiable, then ${\sf val}(\Psi)\leq \delta$.
\item The left alphabet of $\Psi$ is $\Sigma^k$ and right alphabet is $\Sigma^{\sqrt{k}}$ for some alphabet $\Sigma$ with $\log|\Sigma|\leq q^{Cd^2}$.
\item The projection map $\phi_{(A,B)}$ associated to the edge $(A,B)$ in $\Psi$ is defined as: $\forall \sigma\in \Sigma_L(A), \phi_{(A,B)}(\sigma)=\sigma|_B$. Furthermore, for all $A \in X(k)$, the circuit complexity of checking membership in $\Sigma_L(A)$ is at most $q^{Cd^2}$.
\end{enumerate}
\end{lemma}

\begin{proof}
Fix $\delta$ and then fix $k,d\in \N$ to be sufficiently large constants depending on $\delta$, as dictated by \Cref{thm:cl} and \Cref{lem:CSP_on_CL}. Applying the polynomial time reduction in \Cref{lem:CSP_on_CL} on $\varphi$, we get a 2-CSP $\Psi'$ on the weighted graph $(X_{n'}(1),X_{n'}(2))$, where $X_{n'}$ is a $d$-dimensional complex from \Cref{thm:cl} with $n'\leq n\log^{O(1)}n$ and parameters $q=\Theta(\log^{C'} n)$ and $\delta$. The alphabet size of $\Psi'$ satisfies $\log|\Sigma|\leq q^{O(d^2)}$ and the decision complexity is at most $q^{O(d^2)}$. If $\varphi$ is satisfiable then so is $\Psi'$ and if not then $\val(\Psi')\leq 1-\eps$ for some absolute constant $\eps>0$. 

\Cref{thm:cl} states that the $(k,\sqrt{k})$-direct product test on $X_{n'}$ has soundness $\delta$. Thus applying the polynomial time reduction in \Cref{lem:soundness-amp} on $\Psi'$ we get a generalized label cover instance $\Psi$ with $\val(\Psi)\leq \delta$ if $\val(\Psi')\leq 1-\eps$ which is at most $1-4\delta$ (by lowering $\delta$ if required). The other properties required of $\Psi$ follow immediately from \Cref{lem:soundness-amp}. 
\end{proof}

We now apply alphabet reduction using the standard technique of proof composition of PCPs, and for that 
we switch to the framework of robust PCPs using~\Cref{lem:lc-robust-equivalence}. Alphabet reduction for label cover corresponds to query/decision complexity reduction for the equivalent robust PCP, therefore applying proof composition with the PCP above as an outer PCP and the decodable PCP based on the Reed-Muller code from \Cref{lem:rm-dpcp} as an inner PCP, we can reduce the queries to $\poly(\log\log\log n)$, while maintaining the almost-linear size. 

\begin{lemma}\label{lem:rm-comp}
For all $\delta>0$ there exists $C \in \N$, such that for sufficiently large $n\in \N$, $3$SAT on $n$ variables has a regular robust PCP with proof length $\leq n(\log n)^C$, randomness complexity $\leq \log_2(n) + C\log\log n$, query and decision complexity $\leq (\log\log\log n)^{C}$ and robust soundness error $\delta$.
\end{lemma}

\begin{proof}
Let $\delta'$ be a function of $\delta$ that we will set later. Applying the polynomial time reduction in \Cref{lem:low-soundness-final} on $\varphi$ with  soundness parameter $\delta'$ and parameters $k,d$ chosen to be a large enough constants, we get a generalized label cover instance $\Psi'$ on the weighted graph $(X_{n'}(k),X_{n'}(\sqrt{k}))$ for $n'\leq n\log^{O(1)}n$, and alphabet $\Sigma$ satisfying that both $\log|\Sigma|$ and the decision complexity are at most $\log^{O(1)}n$. Note that the distribution over the left side of vertices in $\Psi'$ equals $\mu_k$, which is not uniform. The randomness complexity of sampling from $\mu_k$ is at most $\log_2(|X(d)|)+\log_2(\binom{d}{k})=\log_2(n\log^{O(1)}n)$, therefore we can replace the left side by putting in a vertex for every random string. Now, using the equivalence between generalized Label Cover and robust PCPs from \Cref{lem:lc-robust-equivalence} this gives us a robust PCP $P_\varphi$ for $3$SAT with the parameters in \Cref{table-rm1}.

Now we can conclude by applying the composition from \Cref{thm:composition} with the Reed-Muller based dPCP from \Cref{lem:rm-dpcp}. Our goal is to reduce the decision complexity of $P_\varphi$ to $\poly(\log\log\log n)$. To get some intuition of the parameters, note that one step of composition with the Reed-Muller dPCP roughly reduces the original decision complexity $D$ to roughly $\poly\log(D)$, while increasing the proof size by a factor of $\poly(D)$. Since the PCP $P_\varphi$ has a decision complexity of $\pl n$, if we apply the composition twice we will reduce the decision complexity to ${\sf{polylogloglog}} n$, while incurring a factor of $\pl n$ blow-up in size.

{\bf Regularizing:} before applying composition we must ensure that we have a regular robust PCP with constant-sized alphabet. The robust PCP $P_\varphi$ produced in Lemma~\ref{lem:low-soundness-final} may not have these properties. To remedy this we first apply \Cref{lem:regular-pcp} with the parameter $\delta'$ that regularizes the PCP and also reduces its alphabet while paying nominally in the proof length, to get the PCP verifier $P'_\varphi$. 

{\bf A first composition step with Reed-Muller based robust dPCP:} we apply composition, namely~\Cref{thm:composition}, with the parameter $\eps_1=\delta'^{c_1}$ for large enough $c_1 > 0$, with the decodable PCP $D_{\delta_1}$ from \Cref{lem:rm-dpcp} with soundness $\delta_1=\delta'^{c_2}$ for small enough $c_2 \in (0,1)$, to get the PCP verifier $V_1 = P'_\varphi \circledast D_{\delta_1}$ with soundness $\delta'\poly(\frac{1}{\delta_1})+\eps_1\poly(\frac{1}{\delta_1})+\delta_1$ which is $\delta'^c$ for some constant $c\in (0,1)$. The parameter evolution of both these operations is summarized below in \Cref{table-rm1}.

\vspace{-1ex}
\begin{table}[H]
\centering
\caption{Parameters after One Step of Composition with Reed-Muller dPCP.\vspace{-1ex}}\label{table-rm1}
\begin{tabular}{|c|c|c|c|}
\hline & $P_\varphi$ & $P'_\varphi$ & $V_1=P'_\varphi \circledast D_{\delta_1}$ \\
\hline
Proof Alphabet & $\exp(\pl n)$ & $\poly(1/\delta')$ & $\pl n$ \\
Randomness Complexity & $\log_2(n\pl n)$ & $\log(n\pl n)$ & $\log(n\pl n)$ \\
Query Complexity & $O(1)$ & $\pl n$ & $\poly(\frac{1}{\delta'})(\log\log n)^{O(\log(1/\delta'))}$ \\
Decision Complexity & $\pl n$ & $\pl n$ & $\poly(\frac{1}{\delta'})(\log\log n)^{O(\log(1/\delta'))}$ \\
Proof Length & $n\pl n$ & $n\pl n$ & $n\pl n$ \\
Robust Soundness Error & $\delta'$ & $2\delta'$ & $\delta'^{c}$ \\
\hline
\end{tabular}
\end{table}

{\bf A second composition step with Reed-Muller based robust dPCP:} we again apply the alphabet reduction procedure in \Cref{lem:regular-pcp} with the parameter $\delta'^c$ to get a regular robust PCP verifier $V'_2$. After that we apply one more step of composition, with the parameter $\eps_2=\delta'^{c_3}$, using the Reed-Muller dPCP $D_{\delta_2}$ with the soundness parameter $\delta_2=\delta'^{c_4}$. This gives us the regular robust PCP verifier $V_2=V'_2\circledast D_{\delta_2}$ with soundness error $\delta'^{c'}$ for some constant $c'\in (0,1)$. Setting $\delta'=\delta^{1/c'}$ finishes the proof. The parameter evolution is summarized below in \Cref{table-rm2}.

\vspace{-1ex}
\begin{table}[H]
\centering
\caption{Parameters after Second Step of Composition with Reed-Muller dPCP.\vspace{-1ex}}\label{table-rm2}
\begin{tabular}{|c|c|c|}
\hline & $V'_2$ & $V_2=V'_2 \circledast D_{\delta_2}$ \\
\hline
Proof Alphabet & $\poly(1/\delta')$ & $\poly(\frac{1}{\delta})(\log\log n)^{O(\log(1/\delta))}$ \\
Randomness Complexity & $\log(n\pl n)$ & $\log(n\pl n)$ \\
Query Complexity & $\poly(\frac{1}{\delta'})(\log\log n)^{O(\log(1/\delta'))}$ & $\poly(\frac{1}{\delta})(\log\log\log n)^{O(\log(1/\delta))}$ \\
Decision Complexity & $\poly(\frac{1}{\delta'})(\log\log n)^{O(\log(1/\delta'))}$ & $\poly(\frac{1}{\delta})(\log\log\log n)^{O(\log(1/\delta))}$ \\
Proof Length & $n\pl n$ & $n\pl n$ \\
Robust Soundness Error & $2\delta'^c$ & $\delta$ \\
\hline
\end{tabular}
\end{table}
\end{proof}

The PCP construction in~\Cref{lem:rm-comp} 
is size efficient and has a moderately small
alphabet size (which is not constant yet). 
Thus, we now apply another step of query reduction and composition, this time using the Hadamard code based dPCP as an inner PCP from~\Cref{lem:had-dpcp}.  The result of this process will be a robust PCP with constant query complexity and small soundness, which in the language of label cover corresponds to constant size alphabet and small soundness, thereby establishing~\Cref{thm:main}.

\begin{theorem}[\Cref{thm:main} restated]\label{thm:main-restated}
For all $\delta>0$, there exists $C>0$ and a 
polynomial time procedure such that given an instance 
$\varphi$ of $3$-SAT of size $n$ produces a label cover 
instance $\Psi$ with the following properties:
\begin{enumerate}
\item The size of $\Psi$ is at most $n\log^C n$ and the alphabet size of $\Psi$ is at most $\poly(1/\delta)$.
\item If $\varphi$ is satisfiable, then ${\sf val}(\Psi) = 1$.
\item If $\varphi$ is unsatisfiable, then ${\sf val}(\Psi)\leq \delta$.
\end{enumerate}
\end{theorem}

\begin{proof}
Let $\delta'>0$ be a function of $\delta$ that we will set later. Given a $3$SAT instance $\varphi$, using \Cref{lem:rm-comp} we get a robust and regular PCP $P_\varphi$ with soundness $\delta'$. We first apply alphabet reduction using \Cref{lem:regular-pcp} with the parameter $\delta'$ to get a robust PCP $P''_\varphi$ with constant-sized alphabet. We then apply composition, with the parameter $\delta'^{c_1}$ for $c_1\in (0,1)$ chosen to be a large enough absolute constant, and then apply composition with the Hadamard-based decodable PCP from \Cref{lem:had-dpcp} with soundness $\delta_1$, denoted by $D_{\delta_1}$. This gives us a regular PCP $V$ with constant query complexity and robust soundness error $\delta$ by setting $\delta'=\delta^{c}$ for some absolute constant $c >0$. The evolution of parameters is summarized below in \Cref{table-had}.

\begin{table}[H]
\centering
\caption{Parameters after Composition with Hadamard dPCP.}\label{table-had}
\begin{tabular}{|c|c|c|}
\hline & $P'_\varphi$ & $V=P'_\varphi \circledast D_{\delta_1}$ \\
\hline
Proof Alphabet & $\poly(1/\delta')$ & $\poly(1/\delta)$ \\
Randomness Complexity & $\log(n\pl n)$ & $\log(n\pl n)$ \\
Query Complexity & $\poly(\frac{1}{\delta})(\log\log\log n)^{O(\log(1/\delta))}$ & $(1/\delta)^{O(\log1/\delta)}$ \\
Decision Complexity & $\poly(\frac{1}{\delta})(\log\log\log n)^{O(\log(1/\delta))}$ & $\poly(1/\delta)$ \\
Proof Length & $n\pl n$ & $n\pl n$ \\
Robust Soundness Error & $2\delta'$ & $\delta$ \\
\hline
\end{tabular}
\end{table}
Using \Cref{lem:lc-robust-equivalence} to view the robust PCP $V$ as a generalized label cover instance,   gives us the instance $\Psi$ as required. Note that this is a generalized label cover instance, with $|\Sigma_L|\leq O_\delta(1)$, but where for every vertex $U \in L$ the alphabet $\Sigma_L(U)$ might be a subset of $\Sigma_L$. To convert this to a usual Label cover instance, we can simply allow the left alphabet to be all of $\Sigma_L$, where for every $U \in L$ we fix a mapping $G_U: \Sigma_L \rightarrow \Sigma_L(U)$
with $G_U(\sigma)=\sigma$ for all $\sigma\in \Sigma_L(U)$, and interpret the prover's assignment $A(U)$ as the assignment $G_U(A(U))$. It is easy to see that the modified instance has the same soundness.
\end{proof}

\section*{Acknowledgements}
We thank Pavel Etingof for 
bringing us into
contact with Zhiwei Yun. We sincerely thank Zhiwei Yun for helpful
communication about 
Theorem~\ref{thm:cl} 
and for kindly agreeing
to write an appendix to
this paper showing 
that variants of the Chapman-Lubotzky complexes can be constructed with $q=\pl n$. We thank 
Shiva Chidambaram for helpful conversations 
about pro-$p$ groups, Karthik C.S. for pointing out the implication to set cover and Ryan O'Donnell for drawing our attention to the question of polynomial-time constructability of the variants of the Chapman-Lubotzky complexes.

Dor Minzer is supported by NSF CCF award 2227876, and NSF CAREER award 2239160. Nikhil Vyas is supported by a Simons Investigator Fellowship, NSF grant DMS-2134157, DARPA grant W911NF2010021, and DOE grant DE-SC0022199.

\bibliographystyle{alpha}
\bibliography{references}

\appendix
\addtocontents{toc}{\protect\setcounter{tocdepth}{2}}

\section{Proof of Theorem~\ref{thm:dp_intro}}\label{sec:dp-large-alph}
This section is devoted to the proof of
Theorem~\ref{thm:dp_intro} in the case of large alphabet $\Sigma$. We first recall how the proof proceeds in the case $\Sigma=\{0,1\}$. In that case, the argument consists of two parts. The first part is the work of~\cite{BafnaMinzer}, which shows that complexes that possess coboundary expansion support a direct product tester over $\Sigma=\{0,1\}$. The second part is the work~\cite{BLM24}, which constructs sufficiently good coboundary expanders using variants of the Chapman-Lubotzky complexes~\cite{ChapmanL}. Combining the two 
results gives Theorem~\ref{sec:dp-large-alph} in the case that $\Sigma = \{0,1\}$.  

The only part of the argument that changes for larger alphabets is the first one, and in particular~\cite[Theorem B.1]{BafnaMinzer}. Below we verify that the result of~\cite[Theorem B.1]{BafnaMinzer} in fact works for all $\Sigma$, and combined with~\cite{BLM24} this implies \Cref{thm:dp_intro}. 
The main result of this section is the following strengthening of~\cite[Theorem B.1]{BafnaMinzer}:
\begin{theorem}\label{thm:stronger-coboundary_formal}
There is $c>0$ such that for all $\delta>0$ there are $m,r\in\mathbb{N}$ such that for sufficiently large $k$, sufficiently large $d$ and $\gamma$ small enough function of $d$ the following holds. If a $d$-dimensional simplicial complex $X$ is a $\gamma$-spectral expander and $(m,r,2^{-o(r)},c)$ weak UG coboundary expander, then the direct product test over $X(k)$ with respect to every alphabet $\Sigma$ has soundness $\delta$. Namely, if $F\colon X(k)\to \Sigma^k$ passes the $(k,\sqrt{k})$ direct product tester with respect to $X$ with probability at least $\delta$, then there is $f\colon X(1)\to\Sigma$
such that 
\[
\Pr_{A\sim \mu_k}
[\Delta(F[A], f|_A)\leq \delta]\geq \poly(\delta).
\]
\end{theorem}

\subsection{Direct Product Testing over Moderately Sized Alphabets}
In this section we prove Theorem~\ref{thm:dp_intro} for ``moderately large'' alphabet. 
By that, we mean that final result will 
depend on the alphabet size $|\Sigma|$. We use the notation from~\cite[Section 4]{BafnaMinzer}, 
and more specifically: the definition of agreement of a global function $g:[d]\rightarrow \Sigma$ with respect to $G:X(k)\rightarrow\Sigma^k$, denoted by $\agr_\nu(g,G)$, 
the definition of Unique-Games coboundary expansion, and the definition of list-agreement-testing.
We refer the reader
to~\cite[Section 4]{BafnaMinzer} for a formal
presentation of these notions.
\begin{lemma}\label{lem:moderate-alph}
There is $c>0$ such that for all $\delta>0$ there are  $m,r\in\mathbb{N}$ such that for all $R\in\N$, sufficiently large $k$, sufficiently large $d$ and $\gamma$ small enough function of $d$ the 
following holds. If a $d$-dimensional simplicial complex $X$ is a $\gamma$-spectral expander and $(m,r,\exp(-o(r)),c)$ weak UG coboundary expander, then the direct product test over $X(k)$ with respect to every alphabet $\Sigma$ with $|\Sigma|\leq R$ has soundness $\delta$. Namely, if $F\colon X(k)\to \Sigma^k$ passes the $(k,\sqrt{k})$ direct product tester with respect to $X$ with probability at least $\delta$, 
then there is $f\colon X(1)\to\Sigma$
such that 
\[
\Pr_{A\sim \mu_k}
[\Delta(F[A], f|_A)\leq \delta]\geq \poly(\delta).
\]
\end{lemma}

Note here that $d$ is allowed to depend on the alphabet size, which means that $|\Sigma|$ cannot be arbitrarily large as a function of $n$ or $d$. The proof of Lemma~\ref{lem:moderate-alph} follows the argument in~\cite{BafnaMinzer} 
closely. That proof has two major components:
\begin{enumerate}
    \item The first of which is~\cite[Lemma B.2]{BafnaMinzer}, which reduces the $1\%$ agreement testing problem to the $99\%$-list-agreement testing.
    \item The second of which is~\cite[Lemma B.5]{BafnaMinzer}, which deduces
the soundness of the list-agreement test from coboundary expansion. 
\end{enumerate}
Below we explain how to adapt each one of
these components in the case of moderate size
alphabets.

\subsubsection{Modifying\texorpdfstring{~\cite[Lemma B.2]{BafnaMinzer}}{BM23, Lemma B.2}}
Their proof utilizes~\cite[Theorem 1]{AlonVKK02} that says that the value of a dense max-$k$-CSP is preserved up-to small additive factors under sub-sampling of the variables of the CSP. This result is proved only for CSPs over $\{0,1\}$. Below we use the generalization of the statement from~\cite{AlonVKK02} to larger alphabet due to~\cite{BarakHHS11}. We state it in a format analogous to~\cite[Theorem 2.3]{BafnaMinzer} and show how it follows from~\cite[Theorem 8.2]{BarakHHS11}.
\begin{lemma}\label{lem:alon-mod-alph}
For all $k,R \in \N$ and $d \geq \max(\poly(k),\poly\log R)$, consider a $k$-CSP $\Psi$ over the alphabet $[R]$ with $\binom{d}{k}$ constraints that each depend on a unique $k$-set of variables. Then 
\[\Pr_{Q \subset_{d/2} [d]}\left[|\val(\Psi|_Q) - \val(\Psi)| \leq \frac{1}{d^{1/8}}\right] \geq 1-O\left(\frac{1}{d^{1/8}}\right).\]
\end{lemma}
\begin{proof}
Let $\zeta=1/d^{1/4}$. We start by proving that $\val(\Psi|_Q)$ is at least $\val(\Psi)-\zeta$ with high probability. Let $f$ be the maximizer of $\val(\Psi)$. Using Lemma~\ref{lem:sampling}, with the set of constraints $B \subseteq \binom{[d]}{k}$ that $f$ satisfies, we get that,
\[\Pr_{Q \subset_{d/2} [d]}[\val(\Psi|_Q) \leq \val(\Psi)-\zeta]\leq \frac{k\val(\Psi)}{d\zeta^2}\leq \frac{k}{\sqrt{d}}.\]
For proving the other direction, let $p$ denote $\Pr_{Q \subset_{d/2} [d]}[\val(\Psi|_Q)\geq \val(\Psi)+\sqrt{\zeta}]$. Combining with the above equation, we can lower bound the expectation of $\val(\Psi|_Q)$ as follows,
\[\E_{Q \subset_{d/2}[d]}\left[\val(\Psi|_Q)\right] \geq p(\val(\Psi)+\sqrt{\zeta})+\left(1-\frac{k}{\sqrt{d}}-p\right)(\val(\Psi)-\zeta).\]
To upper bound the expectation we use~\cite[Theorem 8.2]{BarakHHS11} which asserts that
\[\E_{Q\subset_{d/2} [d]}[\val(\Psi|_Q)] \leq  \val(\Psi)+\frac{1}{\sqrt{d}},\]
in the parameter regime under consideration here.
Combining the lower and upper bound on the expectation and solving for $p$ we get that, $p \leq O(\sqrt{\zeta})$ as required.     
\end{proof}

Analogously to~\cite[Lemma 4.12]{BafnaMinzer}, Lemma~\ref{lem:alon-mod-alph} implies the following result (we omit the formal proof).
\begin{lemma}\label{lem:no-jump-moderate-alph}
For all $k\in \N$, alphabets $\Sigma$, $d \geq \max(\poly(k),\poly\log |\Sigma|)$, and all functions $G: \binom{[d]}{k} \rightarrow \Sigma^k$ that satisfy $\agr_{t}(g,G) \leq \alpha$ for all $g\colon [d]\to\Sigma$, the following holds:
\[\Pr_{B\subseteq_{d/2} [d]}
\left[\max_g \agr_t(g|_B,G|_B) < \alpha+\frac{1}{d^{1/8}}\right] \geq 1-O\left(\frac{1}{d^{1/8}}\right).\]
\end{lemma}

We remark that the dependence of $d$ on the alphabet size in the above lemma is ultimately why \Cref{lem:moderate-alph} only works for moderately sized alphabet. Using this statement we get the following lemma that reduces the 1\% agreement test to the 99\% list agreement test.
\begin{lemma}\label{lem:agr-to-list-agr-improved-cbdry}
For all $\delta>0$, all alphabets $\Sigma$,
for sufficiently large $k,d \in \N$,
sufficiently small $\gamma$ compared to 
$d$, some $i \in [1/\delta^{80}]$, and $\tau = 1/\text{tower}_{i}(1/\delta)$, the following holds. Suppose that $X$ is a $d$-dimensional simplicial complex which is a $\gamma$-spectral expander, and $F: X(k) \rightarrow \Sigma^k$ passes the $(k,\sqrt{k})$-agreement-test with probability $\delta$. Then, there exists lists $(L[D])_{D \in X(d)}$ satisfying: 
\begin{enumerate}
\item 
\textbf{Short, non-empty lists:} With probability $1-O(\tau)$ over the choice of $D\sim X(d)$, the list $L[D]$ is non-empty and has size at most $O(1/\delta^{12})$.
\item  
\textbf{Good agreement:} For all $D\in X(d)$ and every $f \in L[D]$, we have that $\agr_\nu(f, F|_D) \geq \Omega(\delta^{12})$ for $\nu = 1/k^{\Omega(1)}$.
\item 
\textbf{Distance in the lists:} With probability at least $1-O(\tau)$ over the choice of $D\sim X(d)$, the list $L[D]$ has distance at least $\Omega(1/\log(1/\tau))$.
\end{enumerate}
Furthermore the lists above pass the List-Agreement-Test with parameter $\Theta(\tau)$, with probability $1-\tau$. 
\end{lemma}

\begin{proof}
In the proof of~\cite[Lemma B.2]{BafnaMinzer}, the direct product testing result for the complete complex from~\cite{DinurG08} is used to get a list of functions $L[D]$, such that for most $D \sim X(d)$ the lists satisfy properties (1), (2) and (3) from the lemma statement. The result of~\cite{DinurG08} is alphabet independent hence this part of the proof ports over easily. To show that these lists are consistent with each other, i.e. they pass the list-agreement-test, they use~\cite[Lemma 4.12]{BafnaMinzer} that asserts that the maximum agreement that $G$ has with any global function $g:[d]\rightarrow\{0,1\}$ doesn't increase after restricting to a random subset $B\subset_{d/2} D$ with high probability. We replace that
invocation with the analogous statement for large alphabet-- namely~\Cref{lem:alon-mod-alph} in place of ~\cite[Lemma 4.12]{BafnaMinzer}, gives us that the list-agreement test passes 
with probability $1-\tau$.
\end{proof}

\subsubsection{Modifying\texorpdfstring{~\cite[Lemma B.5]{BafnaMinzer}}{[BM23, Lemma B.5]}}
We now explain how the analysis of the list
agreement testing problem is reduced to 
coboundary expansion, again following the 
argument in~\cite{BafnaMinzer} closely.
\begin{lemma}\label{lem:list-agr-test-improved-cbry}
Assume there exists a collection of lists $\{L[D]\}_{D \in X(d)}$ that satisfy the premise of Lemma~\ref{lem:agr-to-list-agr-improved-cbdry}, and assume that $X$ is a $\gamma$-spectral expander for $\gamma < 1/\poly(d)$ and a weak $(O(1/\delta^{12}),t, \exp(-o(t)), c)$ UG coboundary expander for $t=\Theta(\text{tower}_{i-1}(1/\delta)^2)$.
Then there exists $G: X(1) \rightarrow \Sigma$ such that
\[
\Pr_{D \sim X(d)}\left[\Delta(G(D),L[D]) \leq \delta/3\right] \geq 1-O(c^{1/2} + \exp(-\sqrt{t}) + \gamma).\]
\end{lemma}

\begin{proof}
The proof uses the UG coboundary expansion of $X$ to reduce the 99\% list agreement testing problem to the 99\%-agreement testing problem on HDX. This part of the proof is alphabet-independent. It then uses the result of~\cite{DinurK17} that showed soundness for the 99\%-agreement test on HDX to get a global function $G$ that agrees with some element of the lists on $D$ with high probability. It is easy to verify that the~\cite{DinurK17} result holds for any alphabet $\Sigma$, therefore in the large alphabet case too we get a global function $G$ as required,
finishing the proof of the lemma.
\end{proof}

We can now prove Lemma~\ref{lem:moderate-alph}.
\begin{proof}[Proof of Lemma~\ref{lem:moderate-alph}]
Combining Lemma~\ref{lem:agr-to-list-agr-improved-cbdry} and Lemma~\ref{lem:list-agr-test-improved-cbry} we get that there are lists $L[D]$ as in the former lemma and a function $G\colon X(1)\to\Sigma$ as in the latter lemma. Sampling $D\sim X(d)$ we get that with probability at least $1/2$ we have that $G|_{D}$ is $\delta/3$-close to some $f\in L[D]$. Conditioned on
that, with probability at least $\Omega(\delta^{12})$ we have that $\Delta(f|_{K}, F[K])\leq \nu$ 
and $\Delta(f|_{K}, G|_{K})\leq 2\delta/3-\nu$, 
in which case $\Delta(G|_{K}, F[K])\leq \delta$, 
as required.
\end{proof}

\subsection{Direct Product Testing over All Alphabets}
In this section we complete the proof of \Cref{thm:stronger-coboundary_formal} by a reduction to \Cref{lem:moderate-alph}. We start with a simple claim that relates the singular value of a large induced subgraph in terms of the singular value of the full graph.

\begin{claim}\label{claim:induced-sval}
Let $G=(V,E)$ be a graph and $\cR:L^2(V) \rightarrow L^2(V)$ be a random walk with stationary distribution $\mu$ over $V$ and second singular value $\sigma$. Then for all subgraphs $H \subseteq V(G)$ with $\mu(H) \geq 2\sigma$, the random walk $\cR_H: L^2(H) \rightarrow L^2(H)$  defined as $\cR$ conditioned on $H$, i.e. 
\[
\cR_H f(v)=\E_{u\sim \cR(v)}[f(u)|u \in H],
\]
has singular value at most $O(\sigma/\mu(H))$. 
\end{claim}
\begin{proof}
We know that the second singular value of $\cR$ is $\sigma_2(\cR) = \sup_{f \perp \Vec{1}}\frac{\ip{f,\cR f}}{\ip{f,f}}$ and the same holds for $\cR_H$. Let $f \perp \Vec{1} \in L^2(H)$ be a vector achieving $\sigma_2(\cR_H)$ and define $\wt{f} \in L^2(G)$ as the vector which is equal to $f$ on $H$ and $0$ outside it. Note that $\wt{f}$ is also perpendicular to $\Vec{1}$.

Let $(u,v) \sim E(G)$ and $(u,v)\sim E(H)$ denote an edge picked according to the random walk $\cR$ and $\cR_H$ respectively. We have that,
\begin{align*}
\sigma_2(\cR_H)=\frac{\ip{f,\cR_H f}}{\ip{f,f}}&=\frac{\E_{(u,v)\sim E(H)}[f(u)f(v)]}{\E_{u \sim \pi_H}[f(u)^2]}\\
&=\frac{\E_{(u,v)\sim E(G)}[\wt{f}(u)\wt{f}(v)|u\in H,v\in H]}{\E_{u \sim \pi}[\wt{f}(u)^2|u \in H]}\\
&=\frac{\E_{(u,v)\sim E(G)}[\wt{f}(u)\wt{f}(v)]}{\Pr_{(u,v)\sim E(G)}[u \in H,v\in H]}\cdot \frac{\Pr_{u \sim \pi}[u \in H]}{\E_{u \sim \pi}[\wt{f}(u)^2]}.  
\end{align*}
By the expander mixing lemma we have that,
\[\Pr_{(u,v)\sim E(G)}[u \in H,v\in H] \geq \mu(H)^2 - \sigma_2(\cR)\mu(H)(1-\mu(H)) \geq \mu(H)^2/2.\]
Plugging this in we get,
\[\sigma(\cR_H)\leq \frac{2}{\mu(H)}\cdot \frac{\E_{(u,v)\sim E(G)}[\wt{f}(u)\wt{f}(v)]}{\E_{u \sim \pi}[\wt{f}(u)^2]}\leq \frac{2}{\mu(H)}\sigma_2(\cR).\qedhere
\]
\end{proof}

We are now ready to prove \Cref{thm:stronger-coboundary_formal}.

\begin{proof}[Proof of \Cref{thm:stronger-coboundary_formal}]
We first explain the high level overview of the
argument. Given the function $F$ that passes the test with probability $\delta$, we choose a large constant $R\in \Z_+$ and create a function $G\colon X(k)\to [R]^k$ using a random hash function $h:\Sigma\rightarrow [R]$. We apply the direct product testing result, \Cref{lem:moderate-alph} on $G$ to get a global function $g$ agreeing with $G$ on a large fraction of $k$-sets. Then using $g$ we finally deduce a global function $f$ taking values in $\Sigma$ that agrees with $F$ on a large fraction of $k$-sets.
\vspace{-2ex}
\paragraph{Hashing to smaller alphabet:}
Given $\delta$ we choose $\delta_1 \ll \delta$ appropriately, and then set $\eta_1 = \poly(\delta_1)$ as dictated by \Cref{lem:moderate-alph}, and finally choose $R \gg 1/\eta_1$. The constants $k,d$ for the complex are chosen to be large in terms of $\delta_1,\delta,R$ as required by \Cref{lem:moderate-alph}. To summarize our parameters satisfy, $$\delta \gg \delta_1 \gg \eta_1 \gg 1/R \gg 1/k \gg 1/d.$$ 

Fix such a choice henceforth.
Let $h:\Sigma \rightarrow [R]$ be a randomly chosen function. For $A \in \Sigma^k$ let $h(A)$ denote the string obtained by applying $h$ to every coordinate separately. Consider the function $G_h:X(k)\rightarrow [R]^k$ defined as $G_h = h \circ F$. We know that the distance between two distinct strings  $S_1,S_2 \in \Sigma^t$ can only decrease under hashing, therefore for every $h$, $G_h$ passes the agreement test with probability $\geq \delta$.
Moreover for all $S_1 \neq S_2 \in \Sigma^t$, $\Pr_h[\Delta(h(S_1),h(S_2))<\Delta(S_1,S_2)-\delta_1]\leq \frac{1}{\delta_1R}\leq \frac{1}{\sqrt{R}}$. Therefore,
\begin{equation}\label{eq:dist-collapse}
\E_h\left[\Pr_{\substack{B\sim X(\sqrt{k})\\A,A'\supset_k B}}\left[\Delta(G_h[A]|_B,G_h[A']|_B)\leq\Delta(F[A]|_B,F[A']|_B)-\delta_1\right]\right]\leq \frac{1}{\sqrt{R}}.    
\end{equation}
Fix any hash function $h$ for which the inner probability above is at most $1/\sqrt{R}$. Since $G_h$ passes the agreement test with probability $\delta$, we can apply the agreement testing result for moderate alphabet, Lemma~\ref{lem:moderate-alph}, with $\delta_1$ to get a function $g:X(1)\rightarrow [R]$ such that
\[\Pr_{A\sim X(k)}[\Delta(G_h[A],g|_A)\leq \delta_1] \geq \eta_1.
\] 
Let $\cA=\{A \in X(k): \Delta(G_h[A],g|_A)\leq \delta_1\}$; then the above inequality translates to $\mu(\cA) \geq \eta_1$. 
For each $A\in \cA$, applying Chernoff-type bound gives that
\[\Pr_{\substack{B\sim X(\sqrt{k})\\A\supset_k B}}[\Delta(G_h[A]|_B,g|_B)\leq 2\delta_1 \mid A \in \cA] \geq 1-\exp(-\sqrt{k}).
\]
From here on we will draw the tuple $(B,A,A')$ from the distribution $B\sim X(\sqrt{k}), A,A'\supset_k B$ and omit this notation.
Applying the triangle inequality and a union bound we get,
\begin{equation}\label{eq:hash-agr}
\Pr_{B,A,A'}[\Delta(G_h[A]|_B,G_h[A']|_B)\leq 4\delta_1 \mid A,A' \in \cA] \geq 1-\exp(-\sqrt{k}).    
\end{equation}
By the expander mixing lemma we have that $\Pr_{B,A,A'}[A,A' \in \cA] \geq \mu(\cA)^2 - O\left(\frac{1}{\sqrt{k}}\right) \geq \poly(\eta_1)$.
We now show that $F[A]|_B,F[A']|_B$ are close, analogously to \eqref{eq:hash-agr}, 
and for that we use \eqref{eq:dist-collapse} and a union bound as follows:
\begin{align*}
&\Pr_{B,A,A'}[\Delta(F[A]|_B,F[A']|_B)>5\delta_1\mid A,A'\in \cA]\\
&\leq \Pr_{B,A,A'}[\Delta(F[A]|_B,F[A']|_B)> \Delta(G_h[A]|_B,G_h[A']|_B)+\delta_1\mid A,A'\in \cA]\\&+\Pr_{B,A,A'}[\Delta(G_h[A]|_B,G_h[A']|_B)> 4\delta_1 \mid A,A' \in \cA]\\
&\leq \frac{1}{\poly(\eta_1) \sqrt{R}}+\exp(-\sqrt{k}) \leq \delta_1,   
\end{align*}
where the last inequality holds since $R,k$ are both much larger than $\delta_1,\delta$ and therefore also $\poly(\eta_1)$. 

For the rest of the proof it will be convenient to work with the simplicial complex $Y(k)=\cA$ and its downward closure, endowed with the set of measures $\{\pi_i\}_{i\in [k]}$, where $\pi_k$ is the conditional distribution $\mu_k|\cA$, and as is usually the case, for all $i<k$, $\pi_i$ is  defined as $A \sim \pi_k, I \subset_i A$. Rewriting the above equation we get that $F$ passes the direct-product test (allowing for approximate equality on the intersection) with probability close to $1$:
\begin{equation}\label{eq:dp-test-on-y}
\Pr_{\substack{B \sim Y(k)\\A,A' \supset_k B}}[\Delta(F[A]|_B,F[A']|_B)>5\delta_1] \leq \delta_1.  
\end{equation}
We will use this observation to find a global function $f$ that agrees with $F$ on $Y(k)$. The proof is essentially the same as proving the soundness of the direct product test on HDX in the 99\% regime. We use \Cref{claim:induced-sval} to get that the relevant random walks on $Y$ have good expansion since they are derived from restricting $X$ to $\cA$ which has large measure.

\paragraph{Finding a global function on $Y(1)$:}
To find a global function that agrees with $F$ on $Y(k)$, let us first define a set of good indices in $I \subseteq Y(1)$ as follows:
\begin{enumerate}
\item For every $i\in Y(1)$ define the quantity,
${p_i := \Pr_{\substack{B \sim Y_i(\sqrt{k}-1) \\B \subset A,A' \sim Y(k)}}[F[A]|_i\neq F[A']|_i]}$. If $p_i > \sqrt{6\delta_1}$, 
do not include $i$ in $I$.
\item Consider $\mu_i(\cA)$, the measure of $\cA$ in the link of $i$ (with respect to $X$).
Do not include $i$ in $I$ if $\mu_i(\cA) < \mu(\cA)/2$.
\end{enumerate}
Let $\cA_i:=\{A\in \cA\mid i\in A\}$ which also equals $Y_i(k-1)$. Our global function $f$ is defined as: $f(i)=\maj(F[A]|_i \mid A \sim \cA_i)$ if the majority exists and arbitrary otherwise. 

For every $i\in I$, we will show that this is an overwhelming majority, i.e.\ with high probability over $A\sim \cA_i$, $F[A]|_i=f(i)$. 
To do so we will first bound the second singular value of the down-up random walk $\cR_i$ on $Y_i(k-1)$ defined as: $B \sim Y_i(\sqrt{k}-1), B \subset A,A'\sim Y_i(k-1)$. Let $\cR'_i$ be the random walk: $B \sim X_i(\sqrt{k}-1), B \subset A,A'\sim X_i(k-1)$. One can check that $\cR_i = \cR'_i\mid A,A'\in \cA_i$. Lemma~\ref{lem:spectral_gap_of_graphs_from_HDX} implies that the second singular value of $\cR'_i$  is bounded by $O(1/k)$. Since $\mu(\cA_i) \geq \mu(\cA)/2 \gg 1/k$, by Claim~\ref{claim:induced-sval} we get that the induced random walk $\cR_i$ has singular value $\sigma_2(\cR_i) \leq O(1/k\mu(\cA_i)) \leq 1/\sqrt{k}$. 

We are ready to prove that $f(i)$ agrees with $F[A]|_i$ for most $A \sim \cA_i$. For every $\sigma\in \Sigma$ let $\cA_{i,\sigma}$ be the set of $A \in \cA_i$ where $F[A]|_i=\sigma$ and let $\pi_{i,\sigma}$ denote its measure inside $Y_i(k-1)$. Using Cheeger's inequality we get that,
\[\frac{1}{2}\Pr_{(A,A') \sim \cR_i}[A \in \cA_{i,\sigma},A' \notin \cA_{i,\sigma}] \geq (1-\lambda_2(\cR_i))\pi_{i,\sigma}(1-\pi_{i,\sigma}).\]
Using that $\lambda_2(\cR_i)\leq 1/\sqrt{k}$ and summing up the above over $\sigma \in \Sigma$ we get,
\[\frac{1}{2}\sum_\sigma \Pr_{(A,A') \sim \cR_i}[A \in \cA_{i,\sigma},A' \notin \cA_{i,\sigma}] \geq (1-1/\sqrt{k})(1-\sum_{\sigma}\pi_{i,\sigma}^2).\]
The LHS above is equal to $\Pr_{(A,A')\sim \cR_i}[F[A]|_i \neq F[A']|_i]$ which equals $p_i$. By the assumption that $i$ is in $I$, $p_i$ is less than $\sqrt{6\delta_1}$ so rearranging the above equation we get $\max_{\sigma}\pi_{i,\sigma} \geq 1-O(\sqrt{\delta_1})$. Since $f(i)=\arg\max_{\sigma}\pi_{i,\sigma}$ it follows that
\begin{equation}\label{eq:majority-I}
\Pr_{A \sim Y_i(k-1)}[F[A]|_i = f(i)] \geq 1-O(\sqrt{\delta_1}).    
\end{equation}

\paragraph{Bounding the measure of $I\subseteq Y(1)$:}
We bound the measure of $I$, $\pi_1(I)$, by showing that each of the two conditions defining $I$ is violated with small probability. Let us start with condition (1). First we bound the expectation of $p_i$ as follows,
\begin{align*}
\E_{i\sim Y(1)}[p_i]=\Pr_{\substack{i\sim Y(1)\\ B\sim Y_i(\sqrt{k}-1)\\ A,A'\supset_k B}}[F[A]|_i\neq F[A']|_i]=\E_{\substack{B\sim Y(\sqrt{k})\\ A,A'\supset_k B}}[\Delta(F[A]|_B,F[A']|_B)] \leq 5\delta_1,
\end{align*}
where we used \eqref{eq:dp-test-on-y} in the last inequality.
Applying Markov's inequality we get that,
\[\Pr_{i\sim Y(1)}[p_i >\sqrt{6\delta_1}]\leq \sqrt{6\delta_1}.\]
Now we will bound the probability of violating condition (2). Using Lemma~\ref{lem:sampling} we have that,
\[\Pr_{i \sim X(1)}\left[\mu_i(\cA) \leq \frac{\mu(\cA)}{2}\right]\leq O\left(\frac{1}{k\eta_1}\right)\leq \frac{1}{\sqrt{k}}.\]
We can translate this bound to $i \sim Y(1)$ in a straightforward way:
\[\Pr_{i \sim Y(1)}\left[\mu_i(\cA) \leq \frac{\mu(\cA)}{2}\right]= \sum_{i\in Y(1)}\pi_1(i)\frac{\mu_i(\cA)}{\mu(\cA)}\Ind\left[\mu_i(\cA) \leq \frac{\mu(\cA)}{2}\right]\leq \frac{1}{2}\Pr_{i\sim X(1)}\left[\mu_i(\cA) \leq \frac{\mu(\cA)}{2}\right] \leq \frac{1}{2\sqrt{k}}.\]
By a union bound we conclude that
\begin{equation}\label{eq:measure-I}
\Pr_{i\sim Y(1)}[i \notin I] \leq O\left(\frac{1}{\sqrt{k}}\right)+\sqrt{6\delta_1}\lll \sqrt{\delta_1}. 
\end{equation}

\paragraph{Direct-Product Test Soundness on $Y$:}
We are ready to conclude that $\Delta(F[A],f|_A) \lll \delta_1^{1/4}$ for a large fraction of $A \sim Y(k)$. We do so by calculating the expectation of $\Delta(F[A],f|_A)$:
\begin{align*}
\E_{A \sim Y(k)}[\Delta(F[A],f|_A)]&=\E_{A \sim Y(k)}\E_{i\sim A}[\Ind[F[A]|_i \neq f(i)]]\\
&=\E_{i\sim Y(1)}[\Pr_{A \sim Y_i(k-1)}[F[A]|_i \neq f(i)]]\\
&\leq \Pr_{i \sim Y(1)}[i \notin I]+\E_{i \sim Y(1)}[\Pr_{A \sim Y_i(k-1)}[F[A]|_i \neq f(i)] \mid i \in I]\\
&\lll \sqrt{\delta_1},
\end{align*}
where we used \eqref{eq:measure-I} and \eqref{eq:majority-I} in the last inequality. Applying Markov's inequality we get
\[\Pr_{A \sim Y(k)}[\Delta(F[A],f|_A) \leq \delta_1^{1/4}] \geq 1-O(\delta_1^{1/4}).\]
Moving from the complex $Y$ to the complex $X$
we conclude that
\[\Pr_{A \sim X(k)}[\Delta(F[A],f|_A) \leq \delta_1^{1/4}] \geq \mu(\cA)(1-O(\delta_1^{1/4})) \geq \frac{\eta_1}{2},\]
which gives us the desired conclusion if we set  $\delta_1=\Theta(\delta^4)$.
\end{proof}

\section{Construction of Decodable PCPs}\label{sec:app-dpcp}
In this section we prove Lemmas~\ref{lem:had-dpcp} and~\ref{lem:rm-dpcp}. Both constructions are based on low-degree testing, and in Section~\ref{sec:prelim-low-deg} we begin by covering the necessary background about it. In Section~\ref{sec:had-dpcp} we construct a decodable PCP that has  exponential size but constant alphabet size, establishing Lemma~\ref{lem:had-dpcp}. This construction is based on the Hadamard Code. In Section~\ref{sec:rm-dpcp} we construct a decodable PCP that has polynomial size and quasi-polynomial alphabet size, establishing Lemma~\ref{lem:rm-dpcp}. This construction is based on the Reed-Muller code.

\subsection{Preliminaries of Low Degree Testing}\label{sec:prelim-low-deg}
Let $\F$ be a finite field. A linear plane $P \subseteq \F^m$ is associated with two points $x,y \in \F^m$ and is equal to the set $P=\{t_1 x+t_2 y: t_1,t_2 \in \F\}$. Suppose $f: \F^m \rightarrow \F$ is a purported linear function. Let $\cA$ be an oracle that assigns every plane $P$ in $\F^m$ a linear function $\cA(P)$ that is supposedly the restriction of $f$ onto $P$. Then one can perform the following ``plane-vs-point'' test to verify if $f$ is indeed linear.
\begin{definition}\label{def:plane-vs-point} 
Given a function $f:\F^m \rightarrow \F$ and a planes oracle $\cA$ the plane-vs-point tester proceeds as follows:
\begin{enumerate}
    \item Sample a uniformly random linear plane $P \subset \F^m$ and a random point $x \in P$.
    \item Accept if $\cA(P)(x) = f(x)$, reject otherwise.
\end{enumerate}    
\end{definition}
In general we can perform a subspaces-vs-point test given an arbitrary distribution over subspaces that is well-behaved in the following sense: 
\begin{definition}\label{def:eta-good}
Let $\pi$ be a distribution over tuples of vectors $(x_1,\ldots,x_t)\in \F_q^t$. Abusing notation, we use $\pi$ to also denote the induced distribution of $\spn(x_1,\ldots,x_t)$ where $(x_1,\ldots,x_t)$ are sampled according to $\pi$. Let $\cD_1$ be the joint distribution over $(\spn(x_1,\ldots,x_t),P)$ where $(x_1,\ldots,x_t) \sim \pi$ and $P=\spn(\sum c_i x_i, \sum c'_i x_i)$ for $c_i,c'_i \in \F_q$ chosen independently and uniformly. Let $\cD_2$ be the distribution over $(\Omega, P)$ where a plane $P$ is drawn uniformly from $\F_q^n$ and $\Omega$ is then drawn from $\pi$ conditioned on containing $P$. We say that $\pi$ is $\eta$-good if the total-variation distance of $\cD_1$ and $\cD_2$ is at most $\eta$.    
\end{definition}

The following linearity-testing theorem, first proved in~\cite{HastadW01}, provides a list-decoding guarantee for the plane-vs-point linearity test. One can get soundness for the generalized subspaces-vs-point test using a simple reduction, so long as the associated 
distribution over subspaces is good. We use a version of the linearity-testing as stated in~\cite{MoshkovitzR08}.
\begin{lemma}\label{lem:linearity-testing}
There exists $c>0$ such that the following holds. For all $m \in \Z_{+}$ and primes $q$ that are large enough the following holds. Let $\mathbb{F}$ be a field of size $q$ and let $\delta\in (0,1)$ be such that 
$\delta \geq \frac{1}{q^c}$. For any function $f : \mathbb{F}^m \to \mathbb{F}$, there exists a list of linear functions 
$L_1, L_2, \ldots, L_t$ for $
t = O\left(\frac{1}{\delta^3}\right)$ such that the following holds for any planes oracle $\cA$ (even for a randomized one):
\[ \Pr_{\cA, P,x \in P}\left[\cA(P)(x) \neq f(x) \vee \exists i \in [t], L_i|_P \equiv \cA(P)\right] \geq 1 - \delta.\]
Furthermore the same holds for subspaces $\Omega$ sampled from an $\eta$-good distribution $\pi$,
\[ \Pr_{\cA, \Omega \sim \pi,x \in \Omega}\left[\cA(\Omega)(x) \neq f(x) \vee \exists i \in [t], L_i|_P \equiv \cA(P)\right] \geq 1 - \delta-\eta. \]
\end{lemma}
\begin{proof}
The proof of the first statement can be found in~\cite{HastadW01,MoshkovitzR08}, and we provide a proof of the second statement by reducing it to the first. 

Let $\pi$ be an $\eta$-good distribution over subspaces and let $\cA$ be any subspace oracle. Consider a randomized planes oracle $\cB$ defined as follows: given a plane \(P\), $\cB(P)$ is defined as $\cA(\Omega)|_P$ for $\Omega \sim \pi~|~\Omega \supset P$. By the soundness of linearity testing, there is a list of \(t = \poly(1/\delta)\) linear functions \(L_1, \ldots, L_t\) such that
\[
\Pr_{\cB, P, x\in P}\left[\cB(P)(x) = f(x) \vee \exists i, L_i|_P \equiv \cB(P)\right] \geq 1 - \frac{\delta}{2}.
\]
Let $\cD_1$ and $\cD_2$ be the distribution over $(\Omega,P)$ as specified in \Cref{def:eta-good}. Rewriting the above we get, 
\[
\Pr_{\cA,(\Omega,P)\sim \cD_2,x}\left[A(\Omega)(x) = f(x) \vee \exists i : L_i|_P \equiv \cA(\Omega)|_P\right] \geq 1 - \frac{\delta}{2},
\]
and since $\pi$ is $\eta$-good we conclude that
\[
\Pr_{\cA,(\Omega,P)\sim \cD_1,x}\left[\cA(\Omega)(x) = f(x) \vee \exists i : L_i|_P = \cA(\Omega)|_P\right] \geq 1 - \frac{\delta}{2} - \eta.
\]
Above we have the agreement of \(L_i\) with \(\cA(\Omega)\) on a random \(P\) chosen from \(\Omega\) instead of over all of $\Omega$. However, by a standard Schwartz-Zippel argument, for any \(i\) and \(\Omega\), since $P$ contains a random point in $\Omega$ we get
\[
\Pr_{\cA,P} \left[L_i \neq \cA(\Omega) \wedge L_i|_P = \cA(\Omega)|_P \right] \leq \frac{1}{q}.
\]
Hence, by a union bound over \(i\), we have:
\[
\Pr_{\cA,\Omega\sim\pi,x} \left[\cA(\Omega)(x) = f(x) \vee \exists i \colon L_i|_{\Omega} = \cA(\Omega) \right] \geq 1 - \frac{\delta}{2} - \eta - \frac{t}{q} \geq 1-\delta-\eta.\qedhere
\]
\end{proof}

Suppose that now we want to verify whether a function $f: \F^m \rightarrow \F$ has degree at most $d$. Let $\cA$ be an oracle that assigns every affine plane $P\subset \F^m$ (namely, a set of the form $P=\{x+t_1y+t_2z~|~t_1,t_2\in \F\}$ for some $x,y,z\in\F^m$) a polynomial of degree at most $d$, denoted by $\cA(P)$. 
The polynomial $\cA(P)$ is supposedly the restriction of $f$ onto $P$. Then one can perform the same plane-vs-point test as above. The theorem below, first proved in~\cite{RazS97}, provides a list-decoding guarantee for the plane-vs-point test, and by the same reduction as above we also get soundness for the subspaces-vs-point low-degree test. The analogous statement for $\eta$-good distributions for subspaces follows by a reduction to the plane-vs-point test along the lines of the proof of \Cref{lem:linearity-testing}, hence we omit it.
\begin{lemma}\label{lem:low-deg-testing}
There exists $c>0$ such that the following holds. Let $\F$ be a field of size $q$. 
Let $m, d \in \mathbb{Z}^{\geq 0}$ and 
$\delta \in (0, 1)$ be such that 
$\delta > \left(\frac{md}{q}\right)^c$. For any function $f : \mathbb{F}^m \to \mathbb{F}$, there exists a list of polynomials 
$Q_1, Q_2, \ldots, Q_t$ of degree at most $d$ where $t = O\left(\frac{1}{\delta}\right)$, 
such that the following holds for any planes table $\cA$ (even a randomized one):
\[ \Pr_{\cA, P,x \in P}\left[\cA(P)(x) \neq f(x) \vee \exists i \in [t], Q_i \mid_P \equiv \cA(P)\right] \geq 1 - \delta.\]
Furthermore the same holds for subspaces sampled from an $\eta$-good distribution $\pi$,
\[ \Pr_{\cA, \Omega \sim \pi,x \in \Omega}\left[\cA(\Omega)(x) \neq f(x) \vee \exists i \in [t], Q_i \mid_P \equiv \cA(P)\right] \geq 1 - \delta-\eta. \]
\end{lemma}

\subsection{Proof of Lemma~\ref{lem:had-dpcp}: Hadamard-based dPCP}\label{sec:had-dpcp}
In this section we build a dPCP for Circuit-SAT that encodes a satisfying assignment of $n$ bits using $2^{O(n^2)}$ symbols from an alphabet of size $O(1)$. Our reduction goes through the Quadratic Equations problem, defined as follows:
\begin{definition}
An instance $(X,E)$ of Quadratic Equations over a field $\F$, abbreviated as $\qe_{m,n}(\F)$, is a system of $m$ quadratic equations $E$ in the variables $X = (x_1,\ldots,x_n)$ and the goal is to decide whether the system is satisfiable or not. The value of an instance $Q$ of $\qe$ is the maximum fraction of equations satisfied by any $x\in \F$ and is denoted by $\val(Q)$.
\end{definition}

\begin{lemma}[\Cref{lem:had-dpcp} restated]\label{lem:had-dpcp-restated}
For all $\delta >0$, for $q = 1/\delta^{O(1)}$ and for all alphabets $\Sigma$, the language $\csat_{\Sigma}(N,S)$ has a regular decodable PCP with the following parameters:
\begin{enumerate}
    \item Robust soundness error $\delta$.
    \item Proof alphabet size $q$. 
    \item Proof length $q^{O(S^2)}$.
    \item Randomness complexity $O(S^2\log(q))$.
    \item Query complexity $q^{O(\log|\Sigma|)}$.
    \item Decision complexity $q^{O(\log|\Sigma|)}$.
    \item List size $1/\delta^{O(1)}$.
\end{enumerate}
\end{lemma}
\begin{proof}
We start with an overview of this proof. We reduce the $\csat$ problem to Gap-Quadratic-Equations ($\qe$) over $\F_q$, i.e.\ the problem of deciding whether a system of quadratic equations over $\F_q$ is satisfiable or has value at most $O(1/q)$. Then we reduce this to Gap-Generalized Label Cover, with the right side of vertices being points in $\F_q^{n+n^2}$ and the left side of vertices is low-dimensional subspaces, where $n$ is the number of variables in the $\qe$ instance $Q$. The assignment  $X \in \F_q^{n}$ to $Q$ is thought of as coefficients of a linear function; in fact, to facilitate checking quadratic equations in $X$ we encode the vector $(X,X^{\otimes 2})\in \F_q^{n+n^2}$ using 
the Hadamard code. The prover is supposed to provide the evaluation of this function on the left and right side of vertices. Additionally the left side is constructed so that it contains a random equation and vectors corresponding to the locations of $X$ that the verifier wants to decode at. Using the soundness of linearity testing, by querying a random left vertex and a random right vertex inside it, the verifier can reject if the prover assignment is not an evaluation of $(X,X^{\otimes 2})$ or if $X$ does not satisfy $Q$. If it does not reject, then with high probability, it is able to decode the required values of $X$.

We now proceed to the formal proof.
\vspace{-2ex}
\paragraph{Reduction to Quadratic Equations:} Let $\Sigma=[r]$ and $C$ be an instance of $\csat_\Sigma(N,S)$, i.e.\ $C$ is a Boolean function, $C:\Sigma^N\rightarrow \{0,1\}$ also represented by an $S$-sized circuit $C$ that computes the equivalent function, $C:\{0,1\}^{N\log r}\rightarrow \{0,1\}$. Let the $\Sigma$-valued variables be denoted by $Y=(y_1,\ldots,y_N)$, and each $y_i$ is associated to a block of $\log r$ Boolean variables, denoted by 
$(z_{(i,1)},\ldots,z_{(i,\log r)})$. We will use this identification to move back and forth between the $Y$ and $Z$ variables.

Let us start by reducing the $\csat$ problem to $\qe(\F_2)$ while preserving satisfiability. 
This reduction is standard, following the proof of the Cook-Levin theorem. Formally, we get an instance $Q_1=(X,E_1)$ on $n=O(S)$ Boolean variables denoted by $X = (Z,B)$ (where $B$ is a set of auxiliary variables that the reduction produces) and $m_1=O(n)$ equations, and we have the property that $X$ is a satisfying assignment for $Q_1$ if and only if $Z$ is satisfies the circuit $C$.

\paragraph{Generating a gap for $\qe$:} Fix a prime number $q\in \N$ with $q=1/\delta^C$ for a large enough absolute constant $C>0$ to be chosen later.  Firstly, consider the instance $Q_2=(X,E_2)$ over $\F_q$ where $E_2 = E_1 \cup \{x_i^2-x_1=0\}_{i \in [n]}$ with $|E_2|=m_2$. It is easy to see that $X$ is a satisfying assignment for $E_1$ if and only if it satisfies $E_2$. Let $C$ be any linear code with the generating matrix $G \in \F_q^{m \times m_2}$ with $m=O(m_2)$ and distance $\geq 1-2/q$ (such a code can be constructed in polynomial time by concatenation of standard error correcting codes). Then consider the instance $Q=(X,E)$ with $|E|=m$, where the $j^{th}$-equation in $E$ is the linear combination of equations from $Q_1$ where the $k$th equation is multiplied by $G_{j,k}$. If $X$ satisfies $Q_1$ then it also satisfies $Q_2$, but if $Q_1$ is unsatisfiable then using the distance of $C$ we get that every assignment $X$ satisfies at most $2/q$-fraction of the equations in $E$, i.e.\ $\val(Q_2) \leq 2/q$.
\vspace{-2ex}
\paragraph{Construction of Label-Cover Instance:} We now construct a generalized label cover instance $\Psi$ using the Hadamard code.
The left vertex set $L$ of $\Psi$, is $(\log r+O(1))$-dimensional linear subspaces of $\F_q^{n+n^2}$ endowed with a distribution $\pi_L$ and the right side $R$ is points in $\F_q^{n+n^2}$. 

To describe $\pi_L$ we start with some notation.
Recall that the $i^{th}$ variable $y_i$ for $i \in [N]$ is associated to a block of variables $(z_{(i,1)},\ldots,z_{(i,\log r)})$ whose indices are a subset of $[n]$. Each variable $z_{(i,k)}$ corresponds to the vector $\Vec{e}_{(i,k)} \in \F_q^{n+n^2}$ that has a $1$ in the $(i,k)$-location (that occurs in the first block of $n$ indices) and $0$ everywhere else. Let $S_i = \{\Vec{e}_{(i,k)}:k\leq \log r\}$. Additionally, let the $j^{th}$-equation in $E$ be $\ip{(X,X\otimes X),E_j} = b_j$ for some $E_j \in \F_q^{n+n^2},b_j \in \F_q$.

To pick a random vertex from $\pi_L$, sample $i \sim [N], j \sim [m]$, $y \sim \F_q^n$, $z,z' \sim \F_q^{n+n^2}$ and then pick the subspace $\Omega_{i,j,y,z,z'} \subset \F_q^{n+n^2}$ defined as $\spn(S_i,E_j,(y,0),(0,y\otimes y),z,z')$. For notational convenience, we drop the subscript in $\Omega$ when clear from context. 

We now discuss the alphabets for $\Psi$, also viewed as a prover assignment. As an assignment to the right-side, the prover is supposed to provide us with a linear function $L=(A,A\otimes A)$ mapping a point $C \in \F_q^{n+n^2}$ to $\ip{L,C}=\sum_{i\in [n]} C_i L_i+\sum_{i,j\in [n]} C_{ij}L_{ij}$, where $A$ is a satisfying assignment for $Q$. On the left side the prover is supposed to provide the restriction of $L$ to each subspace. Formally,
\begin{enumerate}
\item \textbf{Right Alphabet:} For each point in $V=\F_q^{n+n^2}$ the prover provides a value in $\F_q$. That is, an assignment of the prover to the vertics on the right side is thought of as a points oracle $f:\F_q^{n+n^2}\rightarrow \F_q$.
\item \textbf{Left Alphabet:} For each subspace $\Omega_{i,j,y,z,z'}$ the prover provides a degree $1$ polynomial $\cA(\Omega)$ via its coefficients ($\dim(\Omega)$-many) on the subspace. For convenience of notation we represent $\cA(\Omega)$ as a vector in $\F_q^{n+n^2}$, although this choice is not unique. The evaluations of $\cA(\Omega)$ must satisfy,
\begin{itemize}
    \item $\ip{\cA(\Omega),E_j} = b_j$.
    \item $\ip{\cA(\Omega), (0,y\otimes y)}= \ip{\cA(\Omega),(y,0)}^2$.
\end{itemize}
\end{enumerate}
Note that the right alphabet size is $q$ and left alphabet size is at most $q^{O(\log r)}$. Given this we have the following PCP decoder-- at input $i \in [N]$,
\begin{enumerate}
\item Randomly sample $\Omega_{i,j,y,z,z'} \sim \pi_L|i$ and $x \sim \Omega_{i,j,y,z,z'}$.
\item If $\ip{\cA(\Omega),x} \neq L(x)$ output $\perp$, else output the symbol $F(\Omega,x) \in \Sigma$ corresponding to the tuple $(\ip{\cA(\Omega),\Vec{e}_{(i,1)}},\ldots,\ip{\cA(\Omega),\Vec{e}_{(i,\log r)}})$. 
\end{enumerate}

\vspace{-2ex}
\paragraph{Completeness:} 
Suppose the $\csat$ instance $C$ we started with is satisfiable, and let $A'$ be a satisfying assignment. In that case the $\qe$ instance $Q$ we generated is satisfiable, and we can pick an assignment $B$ to the auxiliary variables so that the assignment $A=(A',B)$ satisfies $Q$.  Assign the right-side of the label cover according to the linear function $L=(A,A\otimes A)$, i.e. every point $v \in V$ is assigned the value $\ip{L,v}$. For each subspace $\Omega \in U$ assign the linear function $\cA(\Omega)=L|_{\Omega}$.
It is easy to check that the left assignment satisfies all the conditions that the left alphabet is supposed to. Furthermore,
\[\Pr_{\substack{i \sim [N] \\ \Omega_{i,j,y,z,z'}\sim \pi_L|i \\x \sim \Omega_{i,j,y,z,z'}}}\left[F(\Omega,x) = (A(i,1),\ldots,A(i,\log r))\right] = 1.\]

\vspace{-2ex}
\paragraph{Soundness:} We will now verify the soundness condition, and assume that the initial $\csat$ instance $C$ is unsatisfiable. 
Fix an assignment $f$ to the right vertices of the label cover instance. We start by verifying that the distribution $\pi_L$ is good.
\begin{claim}\label{claim:object-good-had}
The distribution $\pi_L$ is $O(1/q)$-good.
\end{claim}
\begin{proof}
Consider the distribution $\cD_1$ that samples $\Omega_{i,j,y,z,z'} \sim \pi$ where \[
\Omega_{i,j,y,z,z'}=\spn(S_i,E_j,(y,0),(0,y\otimes y),z,z'),\] and then samples $P \subseteq \Omega$ with 
\[
P=\spn(c_1 \Vec{e}_{(i,1)}+\ldots+c_{r+4}z+c_{r+5}z',c'_1\Vec{e}_{(i,1)}+\ldots+c'_{r+4}z+c'_{r+5}z')
\]
for uniformly and independently chosen $c_i,c'_i \in \F_q$. If $c_{r+4},c'_{r+5}$ are both not equal to zero and both $z,z' \neq 0$, which happens with probability at least $1-O(1/q)$, then the marginal on $P \sim \cD_1$ is the same as a uniformly random plane. Therefore the total variation distance between the distributions $P \sim \cD_1, \Omega \sim \cD_1|P$ and $P \sim \cD_2, \Omega \sim \cD_2|P$ is at most $O(1/q)$, as required.   
\end{proof}

By Claim~\ref{claim:object-good-had} we may use \Cref{lem:linearity-testing} to get a list of linear functions 
$L_1,\ldots, L_t \in \F_q^{n+n^2}$ for 
$t = O\left(\frac{1}{\delta^3}\right)$ such that for all plane oracles $\cA$,
\begin{equation}\label{eq:linearity-test}
\Pr_{\Omega, x}\left[\ip{\cA(\Omega),x} \neq f(x) \vee \exists j \text{ such that } L_j|_{\Omega} \equiv \cA(\Omega)\right] \geq 1-\frac{\delta}{4}-O\left(\frac{1}{q}\right).   
\end{equation}
We will now prune the above list of linear functions so that we are only left with \(L_j\) such that:
\begin{enumerate}
    \item $L_j = (A_j,A_j \otimes A_j)$ for some $A_j \in \F_q^n$.
    \item $L_j$ satisfies the quadratic system $Q$, i.e.\ $\ip{L_j,E_k}=b_k$ for all $k \in [m]$. 
\end{enumerate}
Denote by $\good$ the set of indices $j \in [t]$ for which $L_j$ satisfies both of the conditions above. First note that if $L_j$ is good then $A_j$ is a satisfying assignment for $Q$. Therefore let us bound the probability that for some $j \notin \good$, $\cA(\Omega) \equiv L_j|_{\Omega}$. 

Fix such an index $j$. Suppose condition (1) is violated for $L_j = (A_j,B_j)$, i.e. $B_j \neq A_j\otimes A_j$. Then consider the degree $2$ polynomials $B_j(y) = \ip{B_j,y \otimes y}$ and $A'_j(y) = \ip{A_j \otimes A_j,y \otimes y}=\ip{A_j,y}^2$ for $y \in \F_q^{n}$. By the Schwartz-Zippel lemma 
$B_j(y)\neq A_j'(y)$ for at least $\left(1-\frac{2}{q}\right)$-fraction of $y$. Since 
$\Omega \sim \pi_L$ contains a random $y$ we get that $\ip{L_j,(0,y\otimes y)} \neq \ip{A_j,y}^2$, thus implying that $\ip{L_j,(0,y\otimes y)} \neq \ip{L_j,(y,0)}^2$ with probability at least $1-O(1/q)$ over $\pi_L$. However, since our assignment $\cA$ always satisfies $\ip{\cA(\Omega),(0,y\otimes y)}=\ip{\cA(\Omega),(y,0)}^2$, we see that the probability, for a random $\Omega$, 
that 
$L_j|_{\Omega} \equiv \cA(\Omega)$ 
is at most $O(1/q)$. 

Let us now suppose that $L_j$ violates (2). Then it can only satisfy $\frac{2}{q}$-fraction of the equations in $E$ since $\val(Q) \leq \frac{2}{q}$ (when it is unsatisfiable). Again since a random $\Omega$ contains a random equation $E_k \sim E$, and $\cA(\Omega)$ satisfies $E_k$, we get that $L_j|_{\Omega} \equiv \cA(\Omega)$ with probability at most $O(1/q)$.
Thus, we have shown that for any bad \(L_j\),
\[ \Pr_{\Omega}\left[ L_j\mid_{\Omega} \equiv \cA(\Omega)\right] \lll 1/q. \]
Hence, a simple union bound gives us that a modification of \eqref{eq:linearity-test} holds,
\begin{equation}\label{eq:dpcp-hadamard}
\Pr_{\Omega,x}\left[\ip{\cA(\Omega),x} \neq f(x) \vee \exists j \in \good: L_j|_{\Omega} \equiv \cA(\Omega)\right] \geq 1 - \frac{\delta}{4} - O\left(\frac{t}{q}\right) \geq 1-\frac{\delta}{2},   
\end{equation}
where the last inequality holds by choosing $q \geq \Omega(1/\delta^4)$. Reformulating \eqref{eq:dpcp-hadamard} we get that there is a list of satisfying assignments $(A_j)_{j \in \good}$ for $Q$ such that for all $\cA$,
\begin{equation}\label{eq:had-soundness}
\Pr_{\substack{i \sim [N] \\ \Omega_{i,j,y,z,z'}\sim \pi_L|i \\x \sim \Omega_{i,j,y,z,z'}}}\left[F(\Omega,x) \in \{\perp\} \cup \{(A_j(i,1),\ldots,A_j(i,\log r)): j\in \good\} \right] \geq 1-\frac{\delta}{2},    
\end{equation}
which completes the proof of soundness of $\Psi$.
\vspace{-2ex}
\paragraph{Modifying $\Psi$ to be regular:} The label cover instance $\Psi$ may not be regular, but this is easy to fix as we explain now. First for simplicity we put in a vertex on the left for every choice of randomness so that we now have a uniform distribution over the left-side of vertices (instead of $\pi_L$). Note that the degree of a subspace $\Omega$ on the left is equal to $q^{\dim(\Omega)}$, therefore we can make it regular by throwing away the subspaces that have small dimension, which are at most a $O(1/q)$-fraction of all the subspaces. 

To make the instance right-regular first note that the distribution on the right side of $\Psi$ is $\eps:=O(1/q)$-TV-close to uniform (the proof is the same as that of \Cref{claim:object-good-had}). Let $d_u$ be the degree of $u \in R$, let $d$ be the average right degree, and let $N$ be the number of vertices on the right. We discard the right vertices $u$ for which $|d_u-d| \geq d\sqrt{\eps}$, which is at most $\leq \sqrt{\eps}$-fraction of all points in $\F_q^{n+n^2}$. Next we add some dummy vertices on the left, and then to each vertex on the right we add at most $\sqrt{\eps} d$ edges to the dummy vertices so that the resulting right vertex set is regular. By discarding the high-degree vertices on the right we might have ruined the left regularity by a bit, so we add some dummy vertices on the right to bring the left degree back to the original. This whole operation costs us at most $1/\sqrt{q}$ in the soundness, which means that \eqref{eq:had-soundness} holds with probability $\geq 1-\delta$.

\vspace{-2ex}
\paragraph{Converting to a robust PCP:}
Using the equivalence between generalized Label Cover and robust PCPs in \Cref{lem:lc-robust-equivalence}, one gets that this is a regular robust decodable PCP, thus finishing the proof. 
\end{proof}

\subsection{Proof of Lemma~\ref{lem:rm-dpcp}: Reed-Muller based dPCP}\label{sec:rm-dpcp}
In this section we build a dPCP for $\csat_{\Sigma}$ that has polynomial size and uses an alphabet of quasi-polynomial size. We follow the exposition from the lecture notes~\cite{lec-notes} very closely.

\subsubsection{Zero on Subcube Test}
The Zero-on-Subcube Testing problem is the following: given a subset $H \subseteq \F$ and a function $f$ we want to test if $f \equiv 0$ on $H^m$ and $\deg(f)\leq d$. This section contains several helpful tools about
the zero-on-subcube testing problem, and as the proofs are straightforward we omit them. 
The claim below is a useful starting point in designing a local test for this problem.
\begin{claim}\label{claim:zero-subcube-structure}
A polynomial \(f\) of degree at most \(d\) is identically zero over \(H^m\) if and only if there exist polynomials \(P_1, P_2, \ldots, P_m\) of degree at most \(d\) such that
\[ f(x) = \sum_{i = 1}^{m} g_H(x_i) P_i(x_1, x_2, \ldots, x_m),\]
where $g_H(x)$ denotes the univariate polynomial $\prod_{h\in H}(x-h)$. 
\end{claim}

Let us now state a local test for verifying if  $f\equiv 0$ on $H^m$ along with $\deg(f)\leq d$. Let $\cA$ be a planes oracle such that for each affine plane \(P\), \(\cA(P)\) is an \((m+1)\)-tuple \((P_0, P_1, P_2, \ldots, P_m)\) of polynomials of degree at most \(d\) such that
\[ P_0 = \sum_i g_H(x_i) P_i. \]
Similarly let $\overline{f}$ be a points oracle that is an \((m+1)\)-tuple of functions \(\overline{f} = (f, f_1, \ldots, f_m)\) from \(\mathbb{F}^m\) to \(\mathbb{F}\). \begin{definition} The zero-on-subcube test then proceeds as follows,
\begin{enumerate}
\item Sample an affine plane \(P\) uniformly at random and a random point \(x\) from it.
\item Query the planes oracle for \(\cA(P)\) and the points oracle for \(\overline{f}(x)\).
\item Accept iff \(\cA(P)(x) = \overline{f}(x)\).
\end{enumerate}    
\end{definition}

It is easy to prove the soundness of the Zero-on-Subcube test using the soundness of the plane-vs-point test, namely \Cref{claim:zero-subcube-structure}, and an application of the Schwartz-Zippel lemma. 
\begin{lemma}\label{lem:zero-on-subcube-test}
There exists $c>0$ such that the following holds. 
Let $\F$ be a field of size $q$, 
let $H \subseteq \F$, let 
$m, d \in \mathbb{Z}^{\geq 0}$ and 
let $\delta \in (0, 1)$ be such that 
$\delta \geq \left(\frac{(d+|H|)m}{q}\right)^c$. For any function 
$f : \mathbb{F}^m \to \mathbb{F}$, there exists a list of polynomial maps 
$\overline{Q^{(1)}}, \ldots, \overline{Q^{(t)}}$, with $\deg(\overline{Q^{(i)}}) \leq d$ and $\overline{Q^{(i)}} \equiv 0$ on $H^m$, where $t = O\left(\frac{1}{\delta}\right)$ such that the following holds for any planes oracle $\cA$ (even for a randomized one):
\[ \Pr_{\cA, P,x \sim P}\left[\cA(P)(x) \neq \overline{f}(x) \vee \exists i \in [t], \overline{Q^{(i)}} \mid_P \equiv \cA(P)\right] \geq 1 - \delta.\]
Furthermore the same holds for subspaces sampled from an $\eta$-good distribution $\pi$,
\[ \Pr_{\cA, \Omega \sim \pi, x \sim \Omega}\left[\cA(\Omega)(x) \neq \bar{f}(x) \vee \exists i \in [t], \bar{Q^{(i)}} \mid_P \equiv \cA(\Omega)\right] \geq 1 - \delta-\eta. \]
\end{lemma}
\begin{proof}
The proof can be found in~\cite{lec-notes}.
\end{proof}

\subsubsection{Proof of Lemma~\ref{lem:rm-dpcp}}
In this section we prove Lemma~\ref{lem:rm-dpcp}, restated below.
\begin{lemma}[\Cref{lem:rm-dpcp} restated]\label{lem:rm-dpcp-restated}
For all $\delta >0$ and all alphabets $\Sigma$, $\csat_{\Sigma}(N,S)$ has a regular decodable PCP with the following parameters: 
\begin{enumerate}
    \item Robust soundness error $\delta$.
    \item Proof alphabet size and proof length 
    at most $S^{O(1)}$.
    \item Randomness complexity at most $O(\log S)$.
    \item Query and decision complexity at most
    $(\log(S))^{O(\log|\Sigma|)}$.
    \item List size at most $1/\delta^{O(1)}$.
\end{enumerate} 
\end{lemma}

\begin{proof}
The proof is broken into several steps.
\vspace{-2ex}
\paragraph{Reduction to 3-SAT:} Let $\Sigma=[r]$ and $\cC$ be an instance of $\csat_\Sigma(N,S)$, i.e.\ $\cC$ is a Boolean function, $\cC:\Sigma^N\rightarrow \{0,1\}$ also represented by an $S$-sized circuit $\cC$ that computes the equivalent function, $\cC:\{0,1\}^{N\log r}\rightarrow \{0,1\}$. Let the $\Sigma$-valued variables be denoted by $Y=(y_1,\ldots,y_N)$, and each $y_i$ is associated to a block of $\log r$ Boolean variables, denoted by 
$(z_{(i,1)},\ldots,z_{(i,\log r)})$. We will use this identification to move back and forth between the $Y$ and $Z$ variables.

Let us start by reducing the $\csat$ problem to 3-SAT while preserving satisfiability, again using the Cook-Levin reduction. Formally, we get a 3-SAT instance $\varphi$ on $n=O(S)$ Boolean variables denoted by $X = (z,b)$ where $b$ is a set of auxiliary variables that the reduction produces, and we have the property that $X$ is a satisfying assignment for $\varphi$ if and only if $Z$ satisfies $\cC$.
\vspace{-2ex}
\paragraph{Arithmetization:} Next, we perform an ``arithmetization'' procedure on $\varphi$. Let \(\mathbb{F}\) be a field of size \(q = \log^C n,\) for some large absolute constant $C>0$ to be chosen later. Fix any subset \(H\) of \(\mathbb{F}\) such that \(H\) contains \(\{0, 1\}\), $|H|=\Theta(\log n)$ and there is an integer \(m=\Theta(\log n/\log\log n)\) such that \(|H|^m = n\); we will identify \([n]\) with \(H^m\) and use these interchangeably.
We will get a polynomial representation of the formula \(\varphi\). For each possible clause of 3 variables, the polynomial encodes whether the clause belongs to \(\varphi\) or not. We think of the formula \(\varphi\) as a function mapping \([n]^3 \times \{0, 1\}^3\) to \(\{0, 1\}\) as follows:
\[
\varphi(i, j, k, b_1, b_2, b_3) = 
\begin{cases}
1 & \text{if } x_i^{b_1} \vee x_j^{b_2} \vee x_k^{b_3} \text{ is a clause in }\varphi,\\
0 & \text{otherwise},
\end{cases}
\]
where \(x_i^0\) and \(x_i^1\) represent the negative and positive instances of \(x_i\), respectively. Since we have identified \(H^m\) with \([n]\) and \(H\) contains \(\{0, 1\}\), we can think of \(\varphi\) as a function from \(H^{3m+3}\) to \(\mathbb{F}\) (define \(\varphi\) to be 0 outside the points mentioned above). As in the case of the assignment, we can define a polynomial \(\wt{\varphi}\) over \(3m + 3\) variables of degree \(O(m|H|)\) that agrees with \(\varphi\) on \(H^{3m}\). Similarly, every Boolean assignment \(A : [n] \to \{0, 1\}\) can also be thought of as a function mapping \(H^m\) to \(\mathbb{F}\). Let $A(x)$ also denote the polynomial of degree \(O(m |H|)\) on $\F^m$ that agrees with \(A\) when evaluated on inputs from \(H^m\).
Given the polynomials \(\wt{\varphi}\) and \(A\) define \(p_{\varphi, A}\) on $\F^{3m+3}$ as follows,
\[ p_{\varphi, A}(i, j, k, b_1, b_2, b_3) = \wt{\varphi}(i, j, k, b_1, b_2, b_3)(A(i) - b_1)(A(j) - b_2)(A(k) - b_3).\]
Note that $\deg(p_{\varphi, A}) \leq O(m|H|)$. We have the following claim (we omit the straightforward proof).
\begin{claim}\label{claim:arithmetization}
Let \(A\) be any polynomial defined on \(m\) variables. Assume the polynomial \(p_{\varphi, A}\) is constructed from \(A\) as above. Then, \(p_{\varphi, A}\) is identically zero on \(H^{3m+3}\) if and only if \(A|_{H^m}\) is a satisfying assignment for the formula \(\varphi\).  
\end{claim}
\vspace{-2ex}
\paragraph{Construction of Label-Cover Instance:}
Given $\varphi$ we will construct the label cover instance $\Psi$. The left-side $L$ will be $O(\log r)$-dimensional linear subspaces of $\F^{3m+3}$ endowed with a distribution $\pi_L$ and the right side $R$ will be $\F^{3m+3}$. First define the linear map $\rho:\F^{3m+3} \rightarrow \F^{3m+3}$ as  
\[
\rho(i,j,k,b_1,b_2,b_3)=(k,i,j,b_1,b_2,b_3).
\]
Recall that the $i^{th}$ variable for $i \in [N]$ is associated to a block of variables $\{(i,1),\ldots,(i,\log r)\} \subseteq [n]$ which in turn corresponds to a $\log r$-sized set $S_i \subset H^m$. To pick a random subspace from $\pi_L$ first sample $i \sim [N]$, extend each element in $S_i$ randomly to $3m+3$ coordinates to get a $\log r$-sized set $\wt{S_i} \subset \F^{3m+3}$. Then sample $y,y',z \sim \F^m$, and $z' \sim \F^m$ such that the first \(m\) coordinates of \(z'\) are the same as the first \(m\) coordinates of \(z\) and the remaining coordinates are uniformly chosen from $\F^{2m+3}$. Then pick the subspace $\Omega_{\wt{S_i},y,y',z,z'}$ defined as $\spn(\wt{S_i},y,y',z,z',\rho(z),\rho^2(z))$. For notational convenience, we drop the subscript in $\Omega_{\wt{S_i},y,y',z,z'}$ when clear from context. The PCP we construct is based on the zero-on-subcube test for subspaces, where the prover is supposed to prove that $p_{\varphi,A}$ is zero over $H^m$ while also allowing us to decode coordinates of $A$.
\begin{enumerate}
\item \textbf{Right Alphabet:} For each point in $R = \F^{3m+3}$, the prover provides a $(3m+5)$-tuple of values in $\F$. This can also be thought of as a ``points oracle'' or a collection of functions \(\overline{f} : \mathbb{F}^{3m+3} \to \mathbb{F}^{3m+5}\), $\overline{f}=(f_{-1},f_0,\ldots,f_{3m+3})$. 
\item \textbf{Left Alphabet:} For each subspace \(\Omega_{\wt{S_i},y,y',z,z'} \in L\), the prover provides $\cA(\Omega)$ which is a \((3m+5)\)-tuple of polynomials \((p_{-1}, p_0, p_1, \ldots, p_{3m+3})\) of degree \(O(m|H|)\) defined on \(\Omega\) such that:
\begin{itemize}
\item \(p_0(x) = \sum_{1 \leq j \leq 3m+3} g_H(x_j)p_j(x)\) for each \(x \in L(\Omega)\).
\item \(p_{-1}(z) = p_{-1}(z')\). %
\item \(p_0(z) = \wt{\varphi}(z)(p_{-1}(z) - z_{3m+1})(p_{-1}(\rho(z)) - z_{3m+2})(p_{-1}(\rho^2(z)) - z_{3m+3})\). 
\end{itemize}
This can be thought of as a ``subspaces oracle''. Each polynomial $p_i$ is provided via its values on the subspace.
\end{enumerate}
Note that the right alphabet has size $q=\poly\log(S)$ and the left alphabet for each subspace is a subset of $[q]^{(3m+5)D}$, where $D$ denotes the number of points in any $O(\log r)$-dimensional subspace, which is at most $q^{O(\log r)}$. It is easy to check that given $\sigma \in [q]^{(3m+5)D}$, one can decide whether it belongs to $\Sigma_L(\Omega)$ (that is, it satisfies the three properties that the left-alphabet is supposed to) with a circuit of polynomial size, i.e.\ size equal to $(3m+5)q^{O(\log r)}\log q=(\log S)^{O(\log r)}$. 

Given this our PCP decoder is simple -- given $i \in [N]$:
\begin{enumerate}
    \item Randomly sample $\Omega_{\wt{S_i},y,y',z,z'} \sim \pi_L|i$ and $x \sim \Omega_{\wt{S_i},y,y',z,z'}$.
    \item If $\cA(\Omega)(x) \neq \overline{f}(x)$, output $\perp$, else output the symbol $F(\Omega,x) \in \Sigma$ that corresponds to the tuple $(p_{-1}(z))_{z \in \wt{S_i}}$.
\end{enumerate}
\vspace{-2ex}
\paragraph{Completeness:}
Suppose $\cC$ is satisfiable.  Then $\varphi$ is also satisfiable, and we let $A$ be some satisfying assignment for it (whose first $N\log r$ variables correspond to a satisfiable assignment of $\cC$). Additionally let  $\wt{A}:\F^{3m+3} \rightarrow \F$ be the polynomial $\wt{A}(i,j,k,b_1,b_2,b_3) = A(i)$. Let $f_{-1} = \wt{A}$ and $f_0 = p_{\varphi, A}$. We know that $p_{\varphi,A}$ is zero on $H^m$ therefore by \Cref{claim:zero-subcube-structure} we get the witness polynomials $f_1 = P_1, \ldots, f_{3m+3} = P_{3m+3}$, with $p_{\varphi,A} = \sum_{1\leq j\leq 3m+3}g_H(x_j)P_j(x)$. Then assign the right-side of the label cover to be $\bar{f} = (f_{-1},f_0,f_1,\ldots,f_{3m+3})$. To assign the left-side of $\Psi$, for each subspace $\Omega$ let $\cA(\Omega) = (p_{-1},\ldots,p_{3m+3})$ with $p_i$ being the restriction of $f_i$ to $\Omega$. It is easy to check that the $p_i$'s satisfy all the conditions they are supposed to. Furthermore,
\[\Pr_{\substack{i \sim [N] \\ \Omega_{\wt{S_i},y,y',z,z'}\sim \pi_L|i \\x \sim \Omega_{\wt{S_i},y,y',z,z'}}}\left[F(\Omega,x) = (A(i,1),\ldots,A(i,\log r))\right] = 1.\]
\vspace{-2ex}
\paragraph{Soundness:} 
Towards etablishing the soundness of the 
reduction, we first prove that $\pi_L$ is 
good.
\begin{claim}\label{claim:object-good-rm}
The distribution $\pi_L$ is \(O(1/q)\)-good.
\end{claim}
\begin{proof}
The proof is identical to the proof of \Cref{claim:object-good-had}.
\end{proof}
Fix an assignment $f$ to the label cover
instance. Using \Cref{lem:zero-on-subcube-test} we get that there exists a short list of polynomial maps $\bar{Q^{(1)}},\ldots, \bar{Q^{(t)}}$ for 
$t = O\left(\frac{1}{\delta}\right)$, such that for all $j$, $Q^{(j)}_0$ is zero on the subcube $H^m$ and $\deg(\bar{Q^{(j)}})\leq d$, and 
for all plane oracles $\cA$
\begin{equation}\label{eq:zero-on-subcube}
\Pr_{\Omega \sim \pi_L, x\sim \Omega}\left[\cA(\Omega)(x) \neq \bar{f}(x) \vee \exists j \text{ such that } \bar{Q^{(j)}}\mid_{\Omega} \equiv \cA(\Omega)\right] \geq 1 - \frac{\delta}{2}.    
\end{equation}

We will now prune the above list of polynomial maps so that we are only left with those tuples \(Q^{(j)}\) such that \(Q^{(j)}_0\) is \(p_{\varphi, A}\) for some satisfying assignment \(A\) of the formula \(\varphi\), and yet the above condition holds for this smaller list of polynomials.
Let $\good$ be the set of indices $j \in [t]$ for which $\bar{Q^{(j)}}$ satisfies:
\begin{enumerate}
\item For all $x \in \F^{3m+3}$,
\[ \bar{Q^{(j)}_0}(x) = \wt{\varphi}(x)(\bar{Q^{(j)}_{-1}}(x) - x_{3m+1})(\bar{Q^{(j)}_{-1}}(\rho(x)) - x_{3m+2})(\bar{Q^{(j)}_{-1}}(\rho^2(x)) - x_{3m+3}). \]
\item For all $z_1 \in \F^{m}$ and \(z_2,z_3 \in \F^{2m+3}\),
\[\bar{Q^{(j)}_{-1}}(z_1, z_2) = \bar{Q^{(j)}_{-1}}(z_1, z_3).\] 
\end{enumerate}
First note that if $j$ is good then $\bar{Q_0^{(j)}}$ can be associated with a satisfying assignment $A^{(j)}$ for $\varphi$. Therefore let us bound the probability that for some $j \notin \good$, $\cA(\Omega) \equiv \bar{Q^{(j)}}|_{\Omega}$. 

Fix such an index $j$. Suppose condition (1) is violated for $\bar{Q^{(j)}}$. By the Schwartz-Zippel lemma, this implies that with probability at least \(1 - \frac{d}{q}\) over the choice of a random \(\Omega\), the above inequality continues to hold when restricted to \(\Omega\), since $\Omega$ contains a random \(z\) such that \(z, \rho(z),\) and \(\rho^2(z)\) lie in \(\Omega\). However, since our assignment \(\cA\) always satisfies condition (1) with equality, we see that the probability, for a random \(\Omega\), that \(\bar{Q^{(j)}}\mid_{\Omega} \equiv \cA(\Omega)\) is at most \(\frac{d}{q}\). Similarly, if (2) above is violated, then it is violated with probability at least \(1-\frac{d}{q}\) over \(\Omega \sim \pi_L\), since \(\Omega\) contains a random \(z, z'\) that agree on the first \(m\) coordinates. Hence, the probability that \(\bar{Q^{(j)}}\mid_{\Omega} \equiv \cA(\Omega)\) is at most $d/q$. Thus, we have shown that for any bad $j$,
\[ \Pr_{\Omega}\left[\bar{Q^{(j)}}\mid_{\Omega} \equiv \cA(\Omega)\right] \lll \frac{d}{q}. \]
Hence, a simple union bound gives us that a modification of \eqref{eq:zero-on-subcube} holds:
\begin{equation}\label{eq:dpcp1}
\Pr_{\Omega,x}\left[\cA(\Omega)(x) \neq \bar{f}(x) \vee \exists j \in \good: \bar{Q^{(j)}}\mid_{\Omega} \equiv \cA(\Omega)\right] \geq 1 - \frac{\delta}{2} - O\left(\frac{td}{q}\right) \geq 1-\delta, 
\end{equation}
where in the last inequality we used that $\delta \geq \poly(m,|H|/q)$. Reformulating \eqref{eq:dpcp1} we get that there is a list of satisfying assignments $(A^{(j)})_{j \in \good}$ such that
\[\Pr_{\substack{i \sim [N] \\ (\wt{i},y,y',z,z')\sim \pi_L|i \\x \sim \Omega_{\wt{i},y,y',z,z'}}}\left[F(\Omega,x) \in \{\perp\} \cup \{(A^{(j)}(i,1),\ldots,A^{(j)}(i,\log r)): j\in \good\} \right] \geq 1-\delta,\]
which completes the proof of soundness.  
\vspace{-2ex}
\paragraph{Modifying $\Psi$ to be regular:} The proof of regularization of $\Psi$ is the same as that of \Cref{lem:had-dpcp-restated}, hence we omit it here.
\vspace{-2ex}
\paragraph{Converting to Robust dPCP:}
Using the equivalence between generalized Label Cover and robust PCPs in \Cref{lem:lc-robust-equivalence}, one gets that this is a regular robust decodable PCP.
\end{proof}

\section{Proof of\texorpdfstring{~\Cref{thm:expander-routing}}{Theorem 3.3}}
\label{app:path-expander}

\paragraph{Notation:} For a graph $G = (V, E)$ let $E(u, S)$ denote the fraction of edges incident of $u$ which are also incident on set $S \subseteq V$. For a set of vertices $\cV \subseteq V$, let $G(\cV)$ denote the induced graph with vertices $\cV$.
\begin{definition}
For a graph $G = (V, E)$ and a subset of vertices $T \subseteq V$ define $Q(T)$ by the following algorithm: Set $Q(T) = V\setminus T$. We now iteratively remove a vertex $u$ from $Q(T)$ if $E(u, V\setminus Q(T)) \geq 1/5$, until $Q(T)$ halts.
\end{definition}

The proof of~\Cref{thm:expander-routing} requires the following two lemmas, which are easy consequences of~\cite[Lemma 1]{Upfal} and~\cite[Lemma 2]{Upfal}.

\begin{lemma}\label{lem:clo-expander}
There exists an absolute constant $\alpha > 0$ 
such that the following holds.
Let $G = (V, E)$ be a regular expander graph with second largest singular value $\sigma_2(G) \leq \alpha$. Let $T \subset V$ be any set such that $|T| \leq \alpha n$. At convergence $|Q(T)| \geq |V|-\mu|T|$ for some universal constant $\mu$.
\end{lemma}

\begin{lemma}\label{lem:paths-expander}
There exists an absolute constant $\alpha>0$ 
such that the following holds.
Let $G = (V, E)$ be a regular expander graph with second largest singular value $\sigma_2(G) \leq \alpha$. Let $T \subset V$ be any set such that $|T| \leq \alpha n$. For all $v_1, v_2 \in Q(T)$ there exists a path of length $O(\log n)$ between $v_1$ and $v_2$ in $G(V\setminus T)$.
\end{lemma}
\begin{proof}
The conclusion follows immediately from~\cite[Lemma 2]{Upfal} by setting $T_2 = \emptyset$ and $T_1 = T$ therein.
\end{proof}

\begin{proof}[Proof of~\Cref{thm:expander-routing}]
Let $|V| = n, |E| = m$, and $\delta = \frac{n}{m}$.
For a pair $u, \pi(u)$, we wish to construct a path between them, which we denote by $P(u)$. 
Fix $\ell = \Theta(\log n)$ and $t = \Theta\left(\log^{c+1}(n)\right)$. 
To construct the paths $P(u)$, we present an algorithm that simply picks the shortest paths between a pair of vertices iteratively and deletes any vertex that has been used at least $t$ times. Note that since paths are of length $\ell$, in an ideal scenario where every edge occurred equally then each edge would belong to  $O\left(\delta \cdot \log n\right)$ paths. Therefore, by taking $t$ to be larger, we allow some slack, which still gives us the uniformity over edges that is sufficient for the later arguments.

We now proceed to the formal argument. Our algorithm proceeds as follows:
\begin{mdframed}
\begin{enumerate}[label=(\alph*)]
\item Instantiate $\forall u, P(u)=\perp$, $\cV=V$.
\item For every $u \in V$ do the following:
\begin{enumerate}[label=(\roman*)]
\item If $u \notin \cV$ set $P(u)=\perp$.
\item Otherwise, find the shortest path $p$ between $u$ and $\pi(u)$ in the graph $G(\cV)$. If the length of $p$ is at most $\ell$ then set $P(u) = p$, else set $P(u)=\perp$.
\item If any vertex $v$ in $\cV$ has been used $\geq t$ times in the paths $\{P(u)\}_u$, then remove it from $\cV$.
\end{enumerate}
\end{enumerate}
\end{mdframed}

It is easy to see that the above algorithm runs in polynomial time in $|E|$ and that every vertex is used at most $t$ times over all paths in $\{P(u)\}_u$. It remains to argue that the algorithm finds paths of length at most $\ell$ between all but $O\left(\frac{1}{\log^c(n)}\right)$ fraction of the $(u, \pi(u))$ pairs. Let $\cV_f$ denote the set $\cV$ when the algorithm terminates, and let $\cV_0=V$ denote the set at the start. It is easy to check that the number of vertices that the algorithm removes from $\cV_0$ to get $\cV_f$ is at most $\frac{|V|\cdot\ell}{t} = \Theta\left(\frac{n}{\log^{c}(n)}\right)$ implying that for $T = V\setminus\cV_f$, we have that $|T| = |V\setminus\cV_f| \leq O\left(\frac{n}{\log^{c} (n)}\right)$. Using~\Cref{lem:clo-expander} we get that $|V\setminus Q(T)| \leq \mu \cdot O\left(\frac{n}{\log^{c}(n)}\right) = O\left(\frac{n}{\log^{c}(n)}\right)$. Finally, using~\Cref{lem:paths-expander} we get that for all $v_1, v_2 \in Q(T)$ there exists a path of length $O(\log n)$ between $v_1$ and $v_2$ in $G(\cV_f)$. Hence, our algorithm would also have found a path of length $\ell$ between them. Since the set $V\setminus Q(T)$ can touch at most $2 \cdot O\left(\frac{1}{\log^{c}(n)}\right)$ of the $(u, \pi(u))$ pairs, we get that we can find a path of length $\ell$ for all but $O\left(\frac{1}{\log^{c}(n)}\right)$ fraction of the $(u, \pi(u))$ pairs.
\end{proof}

\section{Proof of\texorpdfstring{~\Cref{thm:edge-routing}}{Theorem 1.5}}\label{sec:cor-proof}
In this section we prove~\Cref{thm:edge-routing}. 
For that purpose, we need an analogue of~\Cref{lem:link-routing} for \emph{regular graphs}, and towards this end we use complexes of~\cite{LSV2} (instead of variants of the Chapman-Lubotzky complexes from~\Cref{thm:cl}).\footnote{We remark that any sparse complex that satisfying the conditions 1 to 6 of \Cref{thm:cl}, in which all  vertex links are isomorphic to each other, would work for us.} We begin by stating the result of~\cite{LSV2}, and throughout we denote by $\binom{d}{k}_q$ the number of subspaces of dimension $k$ of $\mathbb{F}_q^d$, and by $B_d(F)$ the affine building of type A over a field $F$.
\begin{theorem}[Theorem 1.1 in~\cite{LSV2}]\label{thm:lsv}
Let $q$ be a prime power, $d \geq 2, e \geq 1$. Then, the group $G = \mathrm{PGL}_d(\mathbb{F}_{q^e})$ has an explicit set $S$ of 
\[
\binom{d}{1}_q + \binom{d}{2}_q + \cdots + \binom{d}{d-1}_q
\]
generators, such that the Cayley complex of $G$ with respect to $S$ is a Ramanujan complex, covered by $B_d(F)$, when $F = \mathbb{F}_q((y))$.    
\end{theorem}
Using the complexes from~\ref{thm:lsv} one gets an analogues result to~\Cref{thm:cl} where the complex satisfies properties 1 to 6 of \Cref{thm:cl} as well as the additional property that all the vertex-links are the isomorphic -- spherical buildings of type A. 

First, to see that one can choose an LSV complex with $\{n,\ldots, O(n)\}$ vertices with $q=\pl n$, note that for a given prime $q$ and $e\geq 1$, the number of vertices in the Cayley graph from~\Cref{thm:lsv} is $q^{\Theta(ed^2)}$. Since for sufficiently large $x$ there is always a prime between $[x,x+x^{0.6}]$ (e.g.~due to~\cite{baker2001difference}), it easy to check that for every $d$, one can choose a prime $q=\log^{\Theta(1)} n$ and $e$ to get a $d$-dimensional LSV-complex with number of vertices between $n$ and $2n$. The polynomial time constructability of the complex is immediate given that there is an explicit description of the generators of the Cayley graph (see~\cite[Algorithm 9.2]{LSV2}). The other properties from 1 to 6 follow along similar lines to the argument in~\Cref{thm:cl}, using the fact that the the vertex links of the above complex are spherical buildings of type A over $\F_q^d$. We therefore get the following result:
\begin{theorem}\label{thm:lsv_for_routing}
For large enough $d\in \N$, for large enough $n \in \N$ and some prime $q =\Theta(\log^C n)$, one can construct in time ${\sf poly}(n)$ a $d$-dimensional complex $X$ for which $n \leq |X(1)|\leq 2n$ such that for  the following holds:
\begin{enumerate}
\item The complex $X$ is $d$-partite.
\item Every vertex participates in at most $q^{O(d^2)}$ $d$-cliques, and all vertices have the same degree.
\item All vertex links are isomorphic to type $A$ spherical buildings over $\mathbb{F}_q^d$. In particular, for every vertex link $L$ of $X$ there is a group $\sym(L)$ that acts transitively on the $d$-faces $L(d)$.
\item For links $L\neq L_{\emptyset}$ 
and any $i\neq j$, 
the bipartite graph $(L_i(1),L_j(1))$ is uniformly weighted and has diameter $O\left(\frac{d}{|i-j|}\right)$.
\item For all links $L$ of $X$, every bipartite graph $(L_i(1),L_j(1))$ for $i\neq j \in [d] \setminus \text{col}(L)$ has second largest eigenvalue at most 
$\frac{2}{\sqrt{q}}$. 
\item The second largest singular value of $G=(X(1),X(2))$ is at most $\frac{1}{d-1}+\frac{2}{\sqrt{q}}\leq \frac{2}{d}$.
\end{enumerate}
\end{theorem}

With ~\Cref{thm:lsv_for_routing} in hand, we can repeat the argument of Theorem~\Cref{lem:link-routing} to get~\Cref{thm:edge-routing}.
There are a couple of minor changes required to make the proof go through, since in~\Cref{lem:link-routing} we think of the underlying graph as a weighted graph (instead of a graph with multi-edges) and the initial function is on the links instead of the vertices. We elaborate on these changes below.

\begin{lemma}\label{thm:routing-polylog}
There is a universal constant $\eps_0>0$, such that for all $n \in \mathbb{N}$ there exists a regular graph  $G = (V, E)$ (with multiedges) on $\Theta(n)$ vertices with degree $\pl n$ such that for all permutations $\pi$ on $V(G)$, there is a set of routing protocols $\cR=\{R(u,\pi(u))\}_{u\in G}$ that is $(\eps,O(\eps))$-edge-tolerant for all $\eps \leq \eps_0$, with round complexity $O(\log n)$ and work complexity $\pl n$. Furthermore, the total work done by all vertices to implement each protocol $R(u,\pi(u))$ is $\pl n$ and both the graph and the routing protocols can be constructed in time $\poly(n)$. 
\end{lemma}
Lemma~\Cref{thm:routing-polylog} gives~\Cref{thm:edge-routing} immediately, since one can write the edges of a complete graph as a union of $n$ disjoint matchings when $n$ is even. In the rest of the section we give a proof sketch of~\Cref{thm:routing-polylog}.

\begin{proof}[Proof Sketch of \Cref{thm:routing-polylog}]
Let $X$ be the complex from~\Cref{thm:lsv_for_routing} with the number of vertices lying between $n$ and $O(n)$. First consider the weighted graph $G=(X(1),X(2))$ where the weight on each edge is proportional to the probability of sampling it in $X(2)$, which is, $p_{uv} = |X_{uv}(d-2)|/|X(d)|$. Also consider the graph $G'$ with $|X_{uv}(d-2)|$ multiedges in the place of every edge, which is a regular graph with degree $q^{O(d^2)}=\pl n$, since the link of every vertex in $G$ is isomorphic (and equal to the spherical building of type A). It is easy to check that if $G$ is $(\eps/2,\nu)$-tolerant to edge corruptions then $G'$ is $(\eps,\nu)$-tolerant to edge corruptions. This is because a single message transfer along an edge $e = (u,v)$ of $G$ can be simulated by transferring the same message along all the multiedges in $G'$ that correspond to $e$ and taking a majority vote at $v$. This fails only if $>1/2$-fraction of these multiedges are bad, giving us an error tolerance of $\eps$. Therefore it suffices to prove the lemma with the weighted graph $G$ where the probability of edges that are corrupted is $\eps/2$.

Let us now show that~\Cref{lem:link-routing} (for LSV complexes instead of the complexes from \Cref{thm:cl}) implies the lemma statement for $G$. First note that the main difference between~\Cref{lem:link-routing} and~\Cref{thm:edge-routing} for the graph $G$ above is that in~\Cref{lem:link-routing} the initial message is on links $X_{v_1}$, where $v$ is a vertex in the zig-zag product $Z=G\zz H$, and in~\Cref{thm:edge-routing} the initial message is held by vertices in $G$. To resolve this, the first step in the protocol is to have each vertex \(u\) send its message to all the vertices in its link \(L_u\), and this is duplicated on each link $\deg(u)$-many times, so that we can think of this as an initial message on every link $X_{v_1}(1)$ for $v\in Z$. Then consider the natural lifted permutation $\pi'$ on $V(Z)$, where a vertex $(u,c)$ maps to $(\pi(u),c)$, and apply the link to link protocol from~\Cref{lem:link-routing} to $G$ with the permutation $\pi'$. To transfer the messages back to the vertices of $G$, every link \(X_{(v,c)_1}(1)\) sends a message to $v$ and $v$ takes a majority over all the messages received. The work complexity of the protocol is $\pl n$, and it is easy to check that in each protocol $R(u,\pi(u))$ the total work done in the message transfer is also $\pl n$.

To analyze the protocol, first note that by Markov's inequality, if \(\leq \eps\) fraction of the edges are corrupted, at most \(O(\eps)\) links \(X_{v_1}(1)\) for $v\in V(Z)$, will have greater than \(0.01\) fraction of vertices that do not hold the message sent by $v_1$. Therefore, after implementing the link to link protocol from~\Cref{lem:link-routing} only $O(\eps)$-fraction of the links $X_{v_1}$ have the incorrect majority value. For the last step of the message transfer from the links to the vertices, again by Markov's inequality, only $O(\eps)$-fraction of the vertices $u\in G$ can have more than $1/8$-fraction of links corresponding to $u$ where the majority value is incorrect, or where more than $1/8$-fraction of the edges on $u$ are corrupted.  Therefore all but $O(\eps)$-fraction of vertices $u$ compute the correct majority value and hence receive the message sent by $\pi^{-1}(u)$. 
\end{proof}

\section{Construction of Variants of the Chapman-Lubotzky Complexes}\label{app:cl}
\centerline{\large{\textsc{By Zhiwei Yun}}}

\swapnumbers
\numberwithin{equation}{section}

\setlength{\textwidth}{460pt}
\setlength{\oddsidemargin}{0pt}
\setlength{\evensidemargin}{0pt}
\setlength{\topmargin}{0pt}
\setlength{\textheight}{620pt}

\def\AA{\mathbb{A}}
\def\BB{\mathbb{B}}
\def\CC{\mathbb{C}}
\def\DD{\mathbb{D}}
\def\EE{\mathbb{E}}
\def\FF{\mathbb{F}}
\def\GG{\mathbb{G}}
\def\HH{\mathbb{H}}
\def\II{\mathbb{I}}
\def\JJ{\mathbb{J}}
\def\KK{\mathbb{K}}
\def\LL{\mathbb{L}}
\def\MM{\mathbb{M}}
\def\NN{\mathbb{N}}
\def\OO{\mathbb{O}}
\def\PP{\mathbb{P}}
\def\QQ{\mathbb{Q}}
\def\RR{\mathbb{R}}
\def\SS{\mathbb{S}}
\def\TT{\mathbb{T}}
\def\UU{\mathbb{U}}
\def\VV{\mathbb{V}}
\def\WW{\mathbb{W}}
\def\XX{\mathbb{X}}
\def\YY{\mathbb{Y}}
\def\ZZ{\mathbb{Z}}

\def\calA{\mathcal{A}}
\def\calB{\mathcal{B}}
\def\calC{\mathcal{C}}
\def\calD{\mathcal{D}}
\def\calE{\mathcal{E}}
\def\calF{\mathcal{F}}
\def\calG{\mathcal{G}}
\def\calH{\mathcal{H}}
\def\calI{\mathcal{I}}
\def\calJ{\mathcal{J}}
\def\calK{\mathcal{K}}
\def\calL{\mathcal{L}}
\def\calM{\mathcal{M}}
\def\calN{\mathcal{N}}
\def\calO{\mathcal{O}}
\def\calP{\mathcal{P}}
\def\calQ{\mathcal{Q}}
\def\calR{\mathcal{R}}
\def\calS{\mathcal{S}}
\def\calT{\mathcal{T}}
\def\calU{\mathcal{U}}
\def\calV{\mathcal{V}}
\def\calW{\mathcal{W}}
\def\calX{\mathcal{X}}
\def\calY{\mathcal{Y}}
\def\calZ{\mathcal{Z}}

\newcommand\cI{\mathcal{I}}
\newcommand\cJ{\mathcal{J}}
\newcommand\cN{\mathcal{N}}
\newcommand\cO{\mathcal{O}}
\newcommand\cQ{\mathcal{Q}}
\newcommand\cW{\mathcal{W}}
\newcommand\cX{\mathcal{X}}
\newcommand\cY{\mathcal{Y}}
\newcommand\cZ{\mathcal{Z}}

\def\bA{\mathbf{A}}
\def\bB{\mathbf{B}}
\def\bC{\mathbf{C}}
\def\bD{\mathbf{D}}
\def\bE{\mathbf{E}}
\def\bF{\mathbf{F}}
\def\bG{\mathbf{G}}
\def\bH{\mathbf{H}}
\def\bI{\mathbf{I}}
\def\bJ{\mathbf{J}}
\def\bK{\mathbf{K}}
\def\bL{\mathbf{L}}
\def\bM{\mathbf{M}}
\def\bN{\mathbf{N}}
\def\bO{\mathbf{O}}
\def\bP{\mathbf{P}}
\def\bQ{\mathbf{Q}}
\def\bR{\mathbf{R}}
\def\bS{\mathbf{S}}
\def\bT{\mathbf{T}}
\def\bU{\mathbf{U}}
\def\bV{\mathbf{V}}
\def\bW{\mathbf{W}}
\def\bX{\mathbf{X}}
\def\bY{\mathbf{Y}}
\def\bZ{\mathbf{Z}}

\newcommand\frA{\mathfrak{A}}
\newcommand\frB{\mathfrak{B}}
\newcommand\frC{\mathfrak{C}}
\newcommand\frD{\mathfrak{D}}
\newcommand\frE{\mathfrak{E}}
\newcommand\frF{\mathfrak{F}}
\newcommand\frG{\mathfrak{G}}
\newcommand\frH{\mathfrak{H}}
\newcommand\frI{\mathfrak{I}}
\newcommand\frJ{\mathfrak{J}}
\newcommand\frK{\mathfrak{K}}
\newcommand\frL{\mathfrak{L}}
\newcommand\frM{\mathfrak{M}}
\newcommand\frN{\mathfrak{N}}
\newcommand\frO{\mathfrak{O}}
\newcommand\frP{\mathfrak{P}}
\newcommand\frQ{\mathfrak{Q}}
\newcommand\frR{\mathfrak{R}}
\newcommand\frS{\mathfrak{S}}
\newcommand\frT{\mathfrak{T}}
\newcommand\frU{\mathfrak{U}}
\newcommand\frV{\mathfrak{V}}
\newcommand\frW{\mathfrak{W}}
\newcommand\frX{\mathfrak{X}}
\newcommand\frY{\mathfrak{Y}}
\newcommand\frZ{\mathfrak{Z}}

\newcommand\fra{\mathfrak{a}}
\newcommand\frb{\mathfrak{b}}
\newcommand\frc{\mathfrak{c}}
\newcommand\frd{\mathfrak{d}}
\newcommand\fre{\mathfrak{e}}
\newcommand\frf{\mathfrak{f}}
\newcommand\frg{\mathfrak{g}}
\newcommand\frh{\mathfrak{h}}
\newcommand\fri{\mathfrak{i}}
\newcommand\frj{\mathfrak{j}}
\newcommand\frk{\mathfrak{k}}
\newcommand\frl{\mathfrak{l}}
\newcommand\fm{\mathfrak{m}}
\newcommand\frn{\mathfrak{n}}
\newcommand\frp{\mathfrak{p}}
\newcommand\frq{\mathfrak{q}}
\newcommand\frr{\mathfrak{r}}
\newcommand\frs{\mathfrak{s}}
\newcommand\frt{\mathfrak{t}}
\newcommand\fru{\mathfrak{u}}
\newcommand\frv{\mathfrak{v}}
\newcommand\frw{\mathfrak{w}}
\newcommand\frx{\mathfrak{x}}
\newcommand\fry{\mathfrak{y}}
\newcommand\frz{\mathfrak{z}}

\newcommand\tilA{\widetilde{A}}
\newcommand\tilB{\widetilde{B}}
\newcommand\tilC{\widetilde{C}}
\newcommand\tilD{\widetilde{D}}
\newcommand\tilE{\widetilde{E}}
\newcommand\tilF{\widetilde{F}}
\newcommand\tilG{\widetilde{G}}
\newcommand\tilH{\widetilde{H}}
\newcommand\tilI{\widetilde{I}}
\newcommand\tilJ{\widetilde{J}}
\newcommand\tilK{\widetilde{K}}
\newcommand\tilL{\widetilde{L}}
\newcommand\tilM{\widetilde{M}}
\newcommand\tilN{\widetilde{N}}
\newcommand\tilO{\widetilde{O}}
\newcommand\tilP{\widetilde{P}}
\newcommand\tilQ{\widetilde{Q}}
\newcommand\tilR{\widetilde{R}}
\newcommand\tilS{\widetilde{S}}
\newcommand\tilT{\widetilde{T}}
\newcommand\tilU{\widetilde{U}}
\newcommand\tilV{\widetilde{V}}
\newcommand\tilW{\widetilde{W}}
\newcommand\tilX{\widetilde{X}}
\newcommand\tilY{\widetilde{Y}}
\newcommand\tilZ{\widetilde{Z}}

\newcommand\tila{\widetilde{a}}
\newcommand\tilb{\widetilde{b}}
\newcommand\tilc{\widetilde{c}}
\newcommand\tild{\widetilde{d}}
\newcommand\tile{\widetilde{e}}
\newcommand\tilf{\widetilde{f}}
\newcommand\tilg{\widetilde{g}}
\newcommand\tilh{\widetilde{h}}
\newcommand\tili{\widetilde{i}}
\newcommand\tilj{\widetilde{j}}
\newcommand\tilk{\widetilde{k}}
\newcommand\tilm{\widetilde{m}}
\newcommand\tiln{\widetilde{n}}
\newcommand\tilo{\widetilde{o}}
\newcommand\tilp{\widetilde{p}}
\newcommand\tilq{\widetilde{q}}
\newcommand\tilr{\widetilde{r}}
\newcommand\tils{\widetilde{s}}
\newcommand\tilt{\widetilde{t}}
\newcommand\tilu{\widetilde{u}}
\newcommand\tilv{\widetilde{v}}
\newcommand\tilw{\widetilde{w}}
\newcommand\tilx{\widetilde{x}}
\newcommand\tily{\widetilde{y}}
\newcommand\tilz{\widetilde{z}}

\def\hatG{\widehat{G}}

\newcommand\dB{B^{\vee}}
\newcommand\dC{C^{\vee}}
\newcommand\dD{D^{\vee}}
\newcommand\dG{G^{\vee}}
\newcommand\dH{H^{\vee}}
\newcommand\dK{K^{\vee}}
\newcommand\dL{L^{\vee}}
\newcommand\dM{M^{\vee}}
\newcommand\dN{N^{\vee}}
\newcommand\dP{P^{\vee}}
\newcommand\dQ{Q^{\vee}}
\newcommand\dS{S^{\vee}}
\newcommand\dT{T^{\vee}}
\newcommand\dU{U^{\vee}}

\newcommand\aff{\textup{aff}}
\newcommand\Aff{\textup{Aff}}
\newcommand\alg{\textup{alg}}
\newcommand\Alg{\textup{Alg}}
\newcommand\AS{\textup{AS}}
\newcommand\Av{\textup{Av}}
\newcommand\av{\textup{av}}
\newcommand\BM{\textup{BM}}
\newcommand{\Bun}{\textup{Bun}}
\newcommand{\can}{\textup{can}}
\newcommand{\ch}{\textup{char}}
\newcommand{\codim}{\textup{codim}}
\newcommand{\Coh}{\textup{Coh}}
\newcommand{\coker}{\textup{coker}}
\newcommand\Cone{\mathrm{Cone}}
\newcommand\cont{\mathrm{cont}}
\newcommand{\Coor}{\textup{Coor}}
\newcommand{\Corr}{\textup{Corr}}
\newcommand{\CS}{\textup{CS}}
\newcommand\CT{\mathrm{CT}}
\newcommand\Div{\textup{Div}}
\newcommand\Dyn{\textup{Dyn}}
\newcommand\Eis{\mathrm{Eis}}
\newcommand\ev{\textup{ev}}
\newcommand{\even}{\textup{even}}
\newcommand{\Fil}{\textup{Fil}}
\newcommand{\Fl}{\textup{Fl}}
\newcommand{\fl}{f\ell}
\newcommand\Forg{\textup{Forg}}
\newcommand\Four{\mathrm{Four}}
\newcommand\Fr{\textup{Fr}}
\newcommand\Frac{\textup{Frac}}
\newcommand\Frob{\textup{Frob}}
\newcommand\FT{\mathrm{FT}}
\newcommand\Fun{\mathrm{Fun}}
\newcommand\Gal{\textup{Gal}}
\newcommand\geom{\textup{geom}}
\newcommand\glob{\textup{glob}}
\newcommand{\gr}{\textup{gr}}
\newcommand\Hk{\mathrm{Hk}}
\newcommand{\Hecke}{\textup{Hecke}}
\newcommand{\Hilb}{\textup{Hilb}}
\newcommand{\Hod}{\textup{Hod}}
\newcommand\IC{\textup{IC}}
\newcommand\id{\textup{id}}
\renewcommand{\Im}{\textup{Im}}
\newcommand{\ind}{\textup{ind}}
\newcommand\inv{\textup{inv}}
\newcommand\Irr{\textup{Irr}}
\newcommand\Jac{\textup{Jac}}
\newcommand\lmod{\mathrm{-mod}}
\newcommand\Lie{\textup{Lie}\ }
\newcommand\Loc{\textup{Loc}}
\newcommand\loc{\textup{loc}}
\newcommand\Mod{\textup{Mod}}
\newcommand\mon{\mathrm{mon}}
\newcommand{\mult}{\textup{mult}}
\newcommand\nil{\textup{nil}}
\newcommand{\Nm}{\textup{Nm}}
\newcommand{\odd}{\textup{odd}}
\newcommand{\opp}{\textup{opp}}
\newcommand\Out{\textup{Out}}
\newcommand{\Perf}{\textup{Perf}}
\newcommand\Perv{\textup{Perv}}
\newcommand{\Pic}{\textup{Pic}}
\newcommand\pr{\textup{pr}}
\newcommand\pro{\textup{pro}}
\newcommand\Proj{\textup{Proj}}
\newcommand\Prym{\textup{Prym}}
\newcommand\pt{\textup{pt}}
\newcommand{\Quot}{\textup{Quot}}
\newcommand\rank{\textup{rank}\ }
\newcommand{\red}{\textup{red}}
\newcommand{\reg}{\textup{reg}}
\newcommand{\Reg}{\textup{Reg}}
\newcommand\Rep{\textup{Rep}}
\newcommand{\Res}{\textup{Res}}
\newcommand\res{\textup{res}}
\newcommand\rk{\textup{rk}}
\newcommand\rs{\textup{rs}}
\newcommand\Sat{\mathrm{Sat}}
\newcommand\sgn{\textup{sgn}}
\newcommand\Sht{\textup{Sht}}
\newcommand\Span{\textup{Span}}
\newcommand\Spec{\textup{Spec}\ }
\newcommand\Spf{\textup{Spf}\ }
\newcommand\sph{\mathrm{sph}}
\newcommand\St{\textup{St}}
\newcommand\Stab{\textup{Stab}}
\newcommand\Supp{\textup{Supp}}
\newcommand\Sym{\textup{Sym}}
\newcommand{\tors}{\textup{tors}}
\newcommand\Tot{\textup{Tot}}
\newcommand{\Tr}{\textup{Tr}}
\newcommand{\tr}{\textup{tr}}
\newcommand\triv{\textup{triv}}
\newcommand\unip{\textup{unip}}
\newcommand{\univ}{\textup{univ}}
\newcommand{\Vect}{\textup{Vect}}
\newcommand{\vol}{\textup{vol}}

\newcommand\Aut{\textup{Aut}}
\newcommand\Hom{\textup{Hom}}
\newcommand\End{\textup{End}}
\newcommand\Map{\textup{Map}}
\newcommand\Mor{\textup{Mor}}
\newcommand\uAut{\underline{\Aut}}
\newcommand\uHom{\underline{\Hom}}
\newcommand\uEnd{\underline{\End}}
\newcommand\uMor{\underline{\Mor}}
\newcommand{\RHom}{\bR\Hom}
\newcommand\RuHom{\bR\uHom}
\newcommand{\Ext}{\textup{Ext}}
\newcommand{\uExt}{\underline{\Ext}}

\newcommand\gl{\mathfrak{gl}}
\newcommand\PGL{\textup{PGL}}
\newcommand\Ug{\textup{U}}
\newcommand\ug{\mathfrak{u}}
\newcommand\SL{\textup{SL}}
\renewcommand\sl{\mathfrak{sl}}
\newcommand\SU{\textup{SU}}
\newcommand\su{\mathfrak{su}}
\newcommand\GU{\textup{GU}}
\newcommand\SO{\textup{SO}}
\newcommand\GO{\textup{GO}}
\newcommand\Og{\textup{O}}
\newcommand\Sp{\textup{Sp}}
\renewcommand\sp{\mathfrak{sp}}
\newcommand\GSp{\textup{GSp}}

\newcommand{\Gm}{\GG_m}
\def\Ga{\GG_a}

\newcommand{\ad}{\textup{ad}}
\newcommand{\Ad}{\textup{Ad}}
\renewcommand\sc{\textup{sc}}
\newcommand{\der}{\textup{der}}

\newcommand\xch{\mathbb{X}^*}
\newcommand\xcoch{\mathbb{X}_*}

\newcommand\upH{\textup{H}}
\newcommand{\leftexp}[2]{{\vphantom{#2}}^{#1}{#2}}
\newcommand{\pH}{\leftexp{p}{\textup{H}}}
\newcommand{\trpose}[1]{\leftexp{t}{#1}}
\newcommand{\Ql}{\QQ_{\ell}}
\newcommand{\Qlbar}{\overline{\QQ}_\ell}
\newcommand{\const}[1]{\overline{\QQ}_{\ell,#1}}
\newcommand{\twtimes}[1]{\stackrel{#1}{\times}}
\newcommand{\hotimes}{\widehat{\otimes}}
\newcommand{\htimes}{\widehat{\times}}
\renewcommand{\j}[1]{\langle{#1}\rangle}
\newcommand{\wh}[1]{\widehat{#1}}
\newcommand\quash[1]{}
\newcommand\mat[4]{\left(\begin{array}{cc} #1 & #2 \\ #3 & #4 \end{array}\right)}  %
\newcommand\un{\underline}
\newcommand{\bu}{\bullet}
\newcommand{\ov}{\overline}
\newcommand{\bs}{\backslash}
\newcommand{\tl}[1]{[\![#1]\!]}
\newcommand{\lr}[1]{(\!(#1)\!)}
\newcommand\sss{\subsubsection}
\newcommand\Qp{\mathbb{Q}_{p}}
\newcommand\Zp{\mathbb{Z}_{p}}
\newcommand\Zl{\mathbb{Z}_{\ell}}

\newcommand\op{\oplus}
\newcommand\ot{\otimes}
\newcommand\bt{\boxtimes}
\newcommand{\sslash}{\mathbin{/\mkern-6mu/}}
\renewcommand\c\circ
\newcommand\vn{\varnothing}
\renewcommand\div{\mathrm{div}}

\newcommand{\homo}[2]{\mathbf{H}_{#1}({#2})}   %
\newcommand{\homog}[2]{\textup{H}_{#1}({#2})}  %
\newcommand{\coho}[2]{\mathbf{H}^{#1}({#2})}    %
\newcommand{\cohog}[2]{\textup{H}^{#1}({#2})}     %
\newcommand{\cohoc}[2]{\textup{H}_{c}^{#1}({#2})}     %
\newcommand{\hBM}[2]{\textup{H}^{\textup{BM}}_{#1}({#2})}  %

\newcommand\mt{\mapsto}
\newcommand{\incl}{\hookrightarrow}
\newcommand{\isom}{\stackrel{\sim}{\to}}
\newcommand{\bij}{\leftrightarrow}
\newcommand{\surj}{\twoheadrightarrow}
\newcommand{\oll}{\overleftarrow}
\newcommand{\orr}{\overrightarrow}
\newcommand{\olr}{\overleftrightarrow}
\newcommand\xr{\xrightarrow}
\newcommand\w{\wedge}
\newcommand{\ari}{\ar@{^{(}->}} %

\renewcommand\a\alpha
\renewcommand\b\beta
\newcommand\g\gamma
\newcommand\G\Gamma
\renewcommand\d\delta
\newcommand\D\Delta
\newcommand{\ph}{\varphi}
\renewcommand\r\rho
\renewcommand{\t}{\tau}
\newcommand{\x}{\chi}
\newcommand{\y}{\eta}
\newcommand{\z}{\zeta}
\newcommand{\ep}{\epsilon}
\newcommand{\vp}{\varpi}
\renewcommand{\l}{\lambda}
\renewcommand{\L}{\Lambda}
\newcommand{\om}{\omega}
\newcommand{\Om}{\Omega}
\newcommand{\Sig}{\Sigma}

\newcommand\dm{\diamondsuit}
\newcommand\hs{\heartsuit}
\newcommand\na{\natural}
\newcommand\sh{\sharp}
\newcommand\da{\dagger}

\renewcommand\v{\vee}

\newcommand{\kbar}{\overline{k}}
\newcommand{\Qbar}{\overline{\QQ}}

\newcommand{\Gk}{\Gal(\kbar/k)}
\newcommand{\GQ}{\Gal(\overline{\QQ}/\QQ)}
\newcommand{\GQp}{\Gal(\overline{\QQ}_{p}/\QQ_{p})}

\newcommand\Yun[1]{{\color{red} #1}}
\newcommand\add{\textup{add}}
\newcommand\adm{\textup{adm}}
\newcommand\Lag{\textup{Lag}}
\newcommand\PFT{\textup{PFT}}
\newcommand\inj{\textup{inj}}
\newcommand\cen{\textup{cen}}
\newcommand\rot{\textup{rot}}
\newcommand\iso{\textup{iso}}
\newcommand\tw{\textup{tw}}
\newcommand\Wa{W_{\textup{aff}}}
\newcommand\WWa{\WW_{\textup{aff}}}
\newcommand\colim{\textup{colim}}
\newcommand\NP{\textup{NP}}
\newcommand\Crit{\textup{Crit}}
\newcommand\disj{\textup{disj}}
\newcommand\Ch{\textup{Ch}}
\newcommand\Isom{\textup{Isom}}
\newcommand\RD{\textup{RD}}
\newcommand\RLG{\textup{RLG}}
\newcommand\gen{\textup{gen}}

\newcommand\cSW{\mathcal{SW}}

\newcommand\WhdP{U_{\dP}(U^{-}_{\dL}, \psi_{\dL})}

\newcommand\HX{V_{X}}
\newcommand\uni{\textup{uni}}
\newcommand\sst{\textup{sst}}
\newcommand\et{\textup{\'et}}
\newcommand\diam{{\sf diam}}

\vspace{10pt}

\subsection{Setup}
Fix an odd prime $\ell\equiv -1\mod 4$.  Let $D$ be the quaternion algebra over $\QQ$ ramified at $\ell$ and $\infty$. To be concrete, take $D$ to be the $\QQ$-vector space with basis $1$ (which is the unit for the ring structure), $i,j$ and $k$ satisfying  
\begin{equation*}
i^{2}=-1, \quad j^{2}=-\ell, \quad k=ij=-ji.
\end{equation*}
For $x=a+bi+cj+dk\in D$, let $\ov x=a-bi-cj-dk$, and
\begin{equation*}
N(x)=x\ov x=a^{2}+b^{2}+\ell c^{2}+\ell d^{2}.
\end{equation*}
Let $\cO_{D}\subset D$ be the subring
\begin{equation*}
\cO_{D}=\ZZ\op \ZZ i\op\ZZ j\op\ZZ k\subset D.
\end{equation*}
Then for all primes $p$ (including $p=\ell$), the $p$-adic completion $\cO_{D,p}$ is a maximal order in $D_{p}=D\ot_{\QQ}\QQ_{p}$. When $p=\ell$, $\cO_{D,\ell}\subset D_{\ell}$ is the unique maximal order.

Let $V=D^{g}$ be a (right) $D$-vector space of dimension $g$, equipped with a standard Hermitian form $h(x_{1},\cdots, x_{g})=\sum_{i=1}^{g}N(x_{i})$.  Let $G=\Ug(V,h)$ be the unitary group for $(V,h)$, which is a reductive group over $\QQ$. Concretely, 
\begin{equation*}
    G(\QQ)=\{A\in M_g(D)|\ov{A^t}A=I_g\}.
\end{equation*}
For any $\QQ$-algebra $R$, it makes sense to consider $G(R)$, which is defined in the same way as $G(\QQ)$, with $D$ replaced by the $R$-algebra $R\ot_\QQ D$.

The algebraic group $G$ is a form of the symplectic group $\Sp_{2g}$ over $\QQ$.  We have a conjugation-invariant algebraic function (over $\QQ$)
\begin{equation*}
\Tr: G\to \AA^{1}
\end{equation*}
that sends $A: V\to V$, written as a $g\times g$ matrix $(a_{st})$ with entries $a_{st}\in D$ under a basis of $V$, to the reduced trace of $\sum_{s=1}^{g}a_{ss}$.

For each prime $p$,  we have a compact open subgroup $K_{p}\subset G(\Qp)$ consisting of matrices where entries are in $\cO_{D,p}$. %
When $p\ne2$ and $p\ne \ell$, upon choosing an isomorphism of $\Zp$-algebras $\cO_{D,p}\cong M_{2}(\ZZ_{p})$ (which induces $D_{p}\cong M_{2}(\Qp)$), we get an isomorphism $G_{\Qp}\cong\Sp_{2g,\Qp}$, under which $K_{p}$ corresponds to the maximal compact subgroup that is the stabilizer of a self-dual $\Zp$-lattice under the symplectic form. The trace function base changed to $\Qp$ is the usual trace for $\Sp_{2g,\Qp}$.

The Lie group $G(\RR)$ is a compact form of $\Sp_{2g}(\CC)$. By writing $\HH=\CC\op \CC j$, we may identify $G(\RR)$ with a subgroup of the compact unitary group $\Ug_{2g}$. The trace function on $G(\RR)$ becomes the usual trace of a $2g\times 2g$ unitary matrix. In particular, if $A\in G(\RR)$ has $\Tr(A)=2g$, then all eigenvalues of $A$ are equal to $1$, hence $A=I$ is the identity element.

\subsection{Building}
Now fix a prime $p\ne2,p\ne\ell$. Fix an isomorphism $\cO_{D,p}\cong M_2(\ZZ_p)$, hence an isomorphism $G_{\Qp}\cong \Sp_{2g,\Qp}$ such that $K_p$ is the stabilizer of the standard lattice $\Zp^{2g}\subset \Qp^{2g}$.

Let $\cB$ be the building of $G_{\Qp}\cong \Sp_{2g,\Qp}$. This is a $g$-dimensional simplicial complex with a simplicial action of $G(\Qp)$. The set of $i$-simplicies of $\cB$ is denoted by $\cB(i)$. In particular, $\cB(0)$ denotes the set of vertices of $\cB$. Let us recall the linear-algebraic meaning of the vertices $\cB(0)$.  A vertex of $\cB$ corresponds uniquely to a $\Zp$-lattice $L\subset \Qp^{2g}$ such that 
\begin{eqnarray*}
   pL^\vee\subset  L\subset L^\vee\subset p^{-1}L.
\end{eqnarray*}
Here, for a $\Zp$-lattice $L\subset \Qp^{2g}$, we write $L^\vee=\{x\in \Qp^{2g}|\langle x, y\rangle\in \Zp,\forall y\in L\}$ (where $\langle -, -\rangle$ is the symplectic form on $\Qp^{2g}$). Such a lattice $L$ is {\em of type $i$}, where $i=0,1,\cdots, g$, if $\dim_{\FF_{p}}(L^{\vee}/L)=2i$. For each $i=0,1,\cdots, g$, let $\cB[i]$ be the set of vertices of $\cB$ of type $i$. 

Two vertices $v$ and $v'$ are adjacent if and only if their corresponding lattices $L$ and $L'$ satisfy 
\begin{equation*}
  \mbox{either }  L\subset L'\subset L^\vee \mbox{ or }L'\subset L\subset (L')^\vee.
\end{equation*}
In either case, we have
\begin{equation}\label{adj latt bound}
    pL \subset L'\subset p^{-1}L.
\end{equation}

A collection of vertices $\{v_{0},v_{1},\cdots, v_{g}\}$ are the vertices of a $g$-simplex if and only if the following are satisfied
\begin{itemize}
\item Up to renaming the vertices, we have $v_{i}\in \cB[i]$ for $i=0,1,\cdots, g$. 
\item Let $L_{i}\subset \Qp^{2g}$ be the type $i$ lattice corresponding to $v_{i}$. Then
\begin{equation*}
pL_{0}\subset L_{g}\subset L_{g-1}\subset \cdots \subset L_{1}\subset L_{0}.
\end{equation*}
\end{itemize}

Let $\G'\subset G(\QQ)$ be the subset of those $g\times g$ matrices with entries in $\cO_{D}[1/p]$. For a subgroup $H\subset K_{\ell}$ of finite index, let
\begin{equation*}
\G'_{H}=\{\g\in \G'|\mbox{the image of $\g$ in $K_{\ell}$ lies in $H$}\}.
\end{equation*}
In particular, $\G'_{H}$ acts on $\cB$ via the embedding
\begin{equation}\label{emb Gamma}
    \iota: \G'_{H}\subset G(\QQ)\subset G(\Qp)\cong \Sp_{2g}(\Qp).
\end{equation}

\begin{prop}\label{p:d} Suppose the image of $H$ under the trace map is contained in $2g+\ell^{b}\Zl$ with $\ell^{b}>4gp^3$, then for any $1\ne \g\in \G'_{H}$ and any vertex $v\in \cB(0)$, $d(v,\g v)\ge 4$.
\end{prop}
\begin{proof}
Suppose   $\g\in \G'_{H}$ and $v\in \cB(0)$ are such that $d(v,\g v)< 4$. Consider the rational number $\Tr(\g)$. By Lemma \ref{l:denom p} below, the $p$-adic valuation of $\Tr(\g)$ is at least $-3$. For any prime $q\ne \ell$ and $q\ne p$,   since $\g\in \Sp_{2g}(\ZZ_{q})$ under an isomorphism $G(\QQ_{q})\cong \Sp_{2g}(\QQ_{q})$, we have that the $q$-adic valuation of $\Tr(\g)$ is $\ge0$. By assumption, the $\ell$-adic valuation of $\Tr(\g)-2g$ is at least $b$. Finally, by considering trace on $G(\RR)\subset \Ug_{2g}$, we have $|\Tr(\g)|\le 2g$. Combining the above information we see that $2g-\Tr(\g)$ takes the form $\frac{a}{p^3}$ for some $\ell^{b}|a\in \ZZ$ and $|a/p|\le 4g$, or $|a|\le 4gp^3$. So if $\ell
^{b}>4gp^3$, we must have $a=0$, which means $\Tr(\g)=2g$. Viewing $\g$ as an element in $\Ug_{2g}$, $\Tr(\g)$ is the sum of $2g$ eigenvalues of $\g$ all of which have complex norm $1$. The fact $\Tr(\g)=2g$ then forces all eigenvalues of $\g$ to be equal to $1$, hence $\g=I\in G(\QQ)\subset \Ug_{2g}$.
\end{proof}

\begin{lemma}\label{l:denom p}
If $\g\in G(\Qp)$ and $v\in \cB(0)$ is a vertex such that $d(v,\g v)\le k$, then $\Tr(\g)\in p^{-k}\Zp$.  
\end{lemma}
\begin{proof}
Fix an isomorphism $G_{\Qp}\cong \Sp_{2g,\Qp}$. Let $v=v_0, v_1,\cdots, v_k=\g v$ be a sequence of vertices in $\cB$ such that $v_i$ is adjacent to $v_{i-1}$ for $i=1,\cdots, k$. 

Let $L_i$ be the lattice corresponding to $v_i$. Since $v_i$ is adjacent to $v_{i-1}$, we have $L_i\subset p^{-1}L_{i-1}$. Therefore $\g L=L_k\subset p^{-k}L_0=p^{-k}L$. Under a $\Zp$-basis of $\L$, $\g$ is then a matrix with $p^{-k}\Zp$-entries, hence $\Tr(\g)\in p^{-1}\Zp$. 
\end{proof}

\subsection{Construction of \texorpdfstring{$H$}{H} at \texorpdfstring{$\ell$}{l}}\label{ss:H}
Recall $\cO_{D,\ell}\subset D_{\ell}=D\ot_{\QQ}\Ql$ is the maximal order. Let $\vp\in \cO_{D,\ell}$ be an element such that $N(\vp)$ has $\ell$-adic valuation $1$. Under our convention of $D$ we can simply take $\vp$ to be $j$.

Now $\vp^{i}\cO_{D,\ell}$ is the two-sided ideal of $\cO_{D,\ell}$ consisting of elements whose reduced norm is in $\ell^{i}\Zl$. We have $\vp^{i}\cO_{D,\ell}/\vp^{i+1}\cO_{D,\ell}\cong \FF_{\ell^{2}}$. The reduced traces of elements in $\vp^{i}\cO_{D,\ell}$ lie in $\ell^{\lceil i/2\rceil}\Zl$.

Identify $G(\Ql)$ with a subgroup of $g\times g$ matrices with entries in $D_{\ell}$. Let $H(0)=K_{\ell}$. For $i>0$, let $H(i)\subset G(\Ql)$ be the subgroup consisting of elements $A\in M_{g}(\cO_{D,\ell})$ such that $A\equiv 1$ mod $\vp^{i}\cO_{D,\ell}$. Then $\Tr(H(i))\subset 2g+\ell^{\lceil i/2\rceil}\Zl$.
 
We have $H(0)/H(1)\cong \Ug_{g}(\FF_{\ell})$ (unitary group for a Hermitian space of dimension $g$ over $\FF_{\ell^{2}}$), whose cardinality is $\ell^{O(g^{2})}$. Direct calculation shows that $H(i)/H(i+1)\cong \Sym^{2}(\FF^{g}_{\ell^{2}})$ if $i$ is odd, which has cardinality $\ell^{g^{2}+g}$, and $H(i)/H(i+1)$ can be identified with $g\times g$ skew-Hermitian matrices with entries in $\FF_{\ell^2}$, if $i>0$ is even, which has cardinality $\ell^{g^{2}}$. We conclude that
\begin{equation}\label{Hi}
[H(0):H(i)]= \ell^{i/2\cdot O(g^{2})}.
\end{equation}

The above discussion also shows that $H(1)/H(i)$ is an $\ell$-group for $i\ge1$. Since $H(1)=\varprojlim_i H(1)/H(i)$, $H(1)$ is a pro-$\ell$-group. Therefore, for any $i\ge1$, $H(i)$ is also a pro-$\ell$-group.

\subsection{Upper bound for \texorpdfstring{$|\G'_{H}\bs \cB(g)|$}{|G'{H} cB(g)|}}

\begin{lemma}\label{l:type 0 vert}
There is a constant $C_{0}$ depending only on $g$ and $\ell$ (and not on $p$), such that $ |\G'\bs \cB[0]|\le C_{0}$.
\end{lemma}
\begin{proof} Note that  $\cB[0]=\Sp_{2g}(\Qp)/\Sp_{2g}(\Zp)=G(\Qp)/K_{p}$. By construction we have an injection
\begin{equation*}
\G'\bs \cB[0]=\G'\bs G(\Qp)/K_{p}\incl G(\QQ)\bs G(\AA)/\prod_{q}K_{q}.
\end{equation*}
The right side is a product of all primes $q$; it is independent of $p$ and depends only on $\ell$ and $g$. We take $C_{0}$ to be the cardinality of the right side. 
\end{proof}
  
Recall that $\cB(g)$ is the set of $g$-dimensional (maximal) simplices of $\cB$.  

\begin{lemma}\label{l:quot H0} 
We have $|\G'\bs \cB(g)|\le C_{0}p^{O(g^{2})}$ for the constant $C_{0}$ in Lemma \ref{l:type 0 vert}.
\end{lemma}
\begin{proof}

Consider the map $v_{0}: \cB(g)\to \cB[0]$ sending each $g$-dimensional simplex to its  unique type 0 vertex. The fibers of this map are in bijection with $X(\FF_{p})$ where $X$ is the flag variety of $\Sp_{2g}$. Therefore fibers of $v_{0}$ have cardinality $p^{O(g^{2})}$. Passing to the quotient, the fibers of the map $\G'\bs \cB(g)\to \G'\bs \cB[0]$ then also have cardinality $\le p^{O(g^{2})}$.
\end{proof}

\begin{cor}\label{cor:Hi} 
There exists a compact open subgroup $H\subset G(\Ql)$ such that:
\begin{enumerate}
\item  For any $1\ne \g\in \G'_{H}$ and any vertex $v\in \cB(0)$, $d(v,\g v)\ge 4$.
\item $|\G'_{H}\bs \cB(g)|\le C_{\ell, g}p^{O(g^{2})}$ for some constant $C_{\ell, g}$ depending only on $\ell$ and $g$.
\end{enumerate}
\end{cor}
\begin{proof}
Take a positive integer $i$ such that $4gp^3<\ell^{\lceil i/2\rceil}\le 4gp^3\ell$. Let $H=H(i)$ as constructed in \S\ref{ss:H}. Since $\Tr(H(i))\subset \ell^{\lceil i/2\rceil}\Zl$ and $\ell^{\lceil i/2\rceil}>4gp^3$,  $H=H(i)$ satisfies the hypothesis of Proposition \ref{p:d}, therefore (1) is satisfied.

For (2), we have
\begin{equation}\label{compare quot}
|\G'_{H}\bs \cB(g)|\le [H(0):H]|\G'_{H(0)}\bs \cB(g)|.
\end{equation}
By \eqref{Hi}, we have
\begin{equation}\label{H0H}
[H(0):H]=\ell^{i/2\cdot O(g^{2})}\le (4gp^3\ell)^{O(g^{2})}.
\end{equation}
By Lemma \ref{l:quot H0} we have
$|\G'_{H(0)}\bs \cB(g)|\le C_{0}p^{O(g^{2})}$ where $C_{0}$ depends only on $\ell$ and $g$. Using \eqref{compare quot} and \eqref{H0H} we get
\begin{equation*}
|\G'_{H}\bs \cB(g)|\le [H(0):H]|\G'_{H(0)}\bs \cB(g)|\le (4gp^3\ell)^{O(g^{2})}\cdot C_{0}p^{O(g^{2})}=C_{\ell, g}p^{O(g^{2})}.
\qedhere
\end{equation*}
\end{proof}

Now we give an estimate of the complexity for constructing the graph underlying $\G'_{H(i)}\bs \cB$ constructed in Corollary \ref{cor:Hi}.

Let $v_{0}\in \cB[0]$ be the standard self-dual lattice $L_{0}=\ZZ_{p}^{2g}$, whose stabilizer is $K_{p}$. 

\begin{lemma}\label{l:vert0 diam} 
Any vertex of  $\G'_{H(i)}\bs \cB$ can be represented by a lattice $L$ satisfying
\begin{equation}\label{bound L}
p^{\diam(\G'_{H(i)}\bs \cB)}L_{0}\subset L\subset p^{-\diam(\G'_{H(i)}\bs \cB)}L_{0}.
\end{equation}
\end{lemma}
\begin{proof} If $L$ is a lattice, viewed as a vertex of $\cB$, and $d(L,L_0)\le n$. By \eqref{adj latt bound} and induction on $n$ we get
\begin{equation*}
    p^nL_0\subset L\subset p^{-n}L_0.
\end{equation*}
Since each vertex in the quotient $\G'_{H}\bs \cB$ is represented by a lattice with distance at most the diameter of $\G'_{H}\bs \cB$, the lemma follows.

\end{proof}

\begin{lemma}\label{lem:time_bd_construct}
    Fix $\ell,g,p\in\mathbb{N}$ and the groups $H(i)$ as above, and let $i$ such that $\ell^{\lceil i/2\rceil}> 4gp^3$. Then the graph $\Gamma'_{H(i)}\backslash\mathcal{B}(g)$ 
     can be constructed in time 
     $p^{O_{\ell,g}\left({\sf diam}(\Gamma'_{H(i)}\backslash\mathcal{B})\right)}.
     $%
\end{lemma}

\begin{proof}  
Let $d=\diam(\Gamma'_{H(i)}\backslash\mathcal{B})$. For $j=0,1,\cdots, g$, 
let $\L[j]$ be the set of type $i$ lattices $L\subset \Qp^{2g}$ satisfying  \eqref{bound L} that correspond to vertices of $\cB$. Let $\L=\cup_{j=0}^g\L[j]$. Thus the vertices of $\G'_{H(i)}\bs \cB$ are images of lattices in $\L$.

Each $L\in \L[0]$ is of the form $hL_0$ for some $h\in \Sp_{2g}(\Qp)$ with entries in $p^{-d}\Zp$. Moreover, if two such $h,h'$ satisfy $h\equiv h'\mod p^{d}\Zp$, then $hL_{0}=h'L_{0}$ for $(h-h')L_{0}\subset p^{d}L_{0}$ which is already contained in both $hL_{0}$ and $h'L_{0}$. Therefore, to list all  lattices in $\L[0]$, we can do it in time $C_{g}p^{2d}$, where $C_g$ depends only on $g$. Similar estimate applies to listing $\L[j]$ for all $i$.

For each $L\in \L$, its neighbors in $\L$ can be listed in time $p^{O(g^{2})}$ again (we only need to look for those $L'$ such that $pL\subset L'\subset p^{-1}L$). Therefore it remains to list all pairs of vertices $(L,L')$ in $\L$ that are in the same $\G'_{H(i)}$-orbit efficiently.

Suppose $\g\in \G'_{H(i)}$ is such that $\g L=L'$ for two vertices $L,L'\in \L$, we claim that the entries $\g_{st}\in \cO_D[1/p]$ of $\g$ lie in the following finite set of cardinality $\le 8p^{2d}$
\begin{equation*}
    \g_{st}\in p^{-2d}\cO_D, \quad N(\g_{st})\le 1.
\end{equation*}
Indeed, since $L,L'$ both satisfy \eqref{bound L}, we see that $\iota(\g)$ has entries in $p^{-2d}\Zp$, which implies the entries of $\g$ are in $p^{-2d}\cO_D$. The condition $N(\g_{st})\le 1$ is always satisfied for unitary matrices. Thus we only need to search within a finite subset of $\G'_{H(i)}$ of cardinality $\le (8p^{2d})^{g^2}=8^{g^2}p^{2g^2d}$ to decide whether a given pair of lattices $L$ and $L'$ are in the same $\G'_{H(i)}$-orbit.  This shows that the equivalence relation on $\L$ given by the $\G'_{H(i)}$-action can be determined within time  $p^{O_{g,\ell}(d)}$.
\end{proof}

\end{document}